\documentclass{ws-ijmpa}
\usepackage[super,compress]{cite}
\usepackage{mathrsfs}
\usepackage{stmaryrd}
\usepackage[all,cmtip]{xy}
\usepackage{graphicx}
\usepackage[usenames,dvipsnames]{color}
\usepackage[colorlinks=true,linkcolor=blue,citecolor=red,urlcolor=Violet,bookmarksdepth=subsection,final]{hyperref}

  \newcommand{\oo}{\infty}
  \newcommand{\del}{\partial}
\renewcommand{\d}{\mathrm{d}}
\renewcommand{\dh}{\mathrm{d}_{\mathsf{h}}}
  \newcommand{\dv}{\mathrm{d}_{\mathsf{v}}}
  \newcommand{\eps}{\varepsilon}
         \def\oast{\circledast} % in case \oast is not defined
  
  \newcommand{\sso}{\subset}
  \newcommand{\sse}{\subseteq}

  \newcommand{\id}{\mathrm{id}}
  \newcommand{\im}{\operatorname{im}}
  \newcommand{\coker}{\operatorname{coker}}
  \newcommand{\supp}{\operatorname{supp}}
  
  \newcommand{\Secs}{\mathrm{\Gamma}}
  \newcommand{\Forms}{\mathrm{\Omega}}

  \newcommand{\E}{\mathcal{E}}
  \newcommand{\EE}{\mathrm{E}}
  \newcommand{\EL}{\mathrm{EL}}

  \newcommand{\G}{\mathrm{G}}
\renewcommand{\H}{\mathrm{H}}
  \newcommand{\J}{\mathrm{J}}
  \newcommand{\K}{\mathrm{K}}
\renewcommand{\L}{\mathcal{L}}

\renewcommand{\S}{\mathcal{S}}

  \newcommand{\R}{\mathbb{R}}
  \newcommand{\Lie}{\mathcal{L}}

	\newcommand{\gK}{{}^g\K}

\begin{document}
\markboth{Igor Khavkine}
{Covariant phase space, constraints, gauge and the Peierls formula}

%%%%%%%%%%%%%%%%%%%%% Publisher's Area please ignore %%%%%%%%%%%%%%%
%
\catchline{}{}{}{}{}
%
%%%%%%%%%%%%%%%%%%%%%%%%%%%%%%%%%%%%%%%%%%%%%%%%%%%%%%%%%%%%%%%%%%%%

\title{COVARIANT PHASE SPACE, CONSTRAINTS, GAUGE AND THE PEIERLS FORMULA}

\author{IGOR KHAVKINE}

\address{Department of Mathematics, Trento University, Via Sommarive, 14
- 38123 Povo (TN) Italy\\ igor.khavkine@unitn.it}

\maketitle

\begin{history}
\received{Day Month Year}
\revised{Day Month Year}
\end{history}

\begin{abstract}
It is well known that both the symplectic structure and the Poisson
brackets of classical field theory can be constructed directly from the
Lagrangian in a covariant way, without passing through the non-covariant
canonical Hamiltonian formalism. This is true even in the presence of
constraints and gauge symmetries. These constructions go under the names
of the covariant phase space formalism and the Peierls bracket. We
review both of them, paying more careful attention, than usual, to the
precise mathematical hypotheses that they require, illustrating them in
examples. Also an extensive historical overview of the development of
these constructions is provided. The novel aspect of our presentation is
a significant expansion and generalization of an elegant and quite
recent argument by Forger~\& Romero showing the equivalence between the
resulting symplectic and Poisson structures without passing through the
canonical Hamiltonian formalism as an intermediary. We generalize it to
cover theories with constraints and gauge symmetries and formulate
precise sufficient conditions under which the argument holds. These
conditions include a local condition on the equations of motion that we
call hyperbolizability, and some global conditions of cohomological
nature. The details of our presentation may shed some light on subtle
questions related to the Poisson structure of gauge theories and their
quantization.
\end{abstract}

\ccode{PACS numbers:}

\setcounter{tocdepth}{3}
\tableofcontents

\section{Introduction}
A classical field theory is essentially defined by a local variational
principle for a given set of dynamical fields on a given spacetime
manifold.  The variational principle determines a set of partial
differential equations (PDEs), the equations of motion, to be emposed on
the dynamical fields. The equations of motion are typically hyperbolic
(or can be made so with gauge fixing if gauge invariance is present).
What distinguishes variational PDEs among the more general class of
hyperbolic PDEs is that their solution spaces can be naturally endowed
with symplectic and hence Poisson structure, making it into a phase
space. The algebra of smooth functions on the phase space with the
corresponding Poisson bracket then constitutes the Poisson algebra of
observables. The existence of this algebraic formulation is what allows
for quantization.

It is clear that the symplectic and Poisson structure on its phase space is a
crucial ingredient in the description of a classical field theory.  The
most common way of building these structures is via the canonical
formalism\cite{dirac-lect,ht} (sometimes known as the 3+1 Hamiltonian
formalism, when spacetime is 4-dimensional), which requires an explicit
choice of a time function, even when no one such choice is natural, and
the application of a Legendre transform, which may be only ambiguously
defined, for instance, in the presence of gauge invariance. However, it
is also well known that both can be defined completely
\emph{covariantly} (that is, without choosing an explicit time function
or applying a Legendre transform) directly from the Lagrangian, without
going through the canonical formalism. These methods are known,
respectively, as the \emph{covariant phase space
formalism}\cite{lw,cw,zuckerman} and the \emph{Peierls
bracket}.\cite{peierls,dewitt-qft} They are clearly preferable when the
canonical formalism explicitly breaks some of the natural symmetries of
the theory (any relativistic theory is an example).  The symplectic
2-form and the Poisson bivector constructed in this way are equivalent
(they are mutual inverses). Despite the covariance of both constructions,
until recently, their equivalence was only known via the intermediary of
the non-covariant canonical formalism.\cite{bhs,henneaux-elim} That is,
until Forger~\& Romero\cite{fr-pois} provided an elegant, completely
covariant proof of the equivalence in the case of a scalar field.

In this review, we describe in depth the constructions of symplectic and
Poisson structures of classical field theory, as well as their
equivalence, all in a covariant way. A novel contribution of our
exposition is an extension of the Forger-Romero argument to field
theories where constraints and gauge invariance are present. Another is
that we do not restrict ourselves to the class of wave-like equations
defined on Lorentzian manifolds (though that will be the main source of
our examples). We also pay special attention to several aspects that are
often omitted or left implicit in the existing literature. Neither the
covariant phase space nor the Peierls constructions are automatic. That
is, besides a given Lagrangian density, several conditions must be
fulfilled for the corresponding formulas to make sense. We make these
conditions explicit in both cases: existence of a Cauchy surface and
spacelike compact support for solutions, in the construction of the
symplectic structure, and existence of hyperbolic PDE system, with
retarded and advanced Green functions, closely related to the equations
of motions, in the construction of the Poisson structure. Furthermore,
the statement of equivalence also requires a certain sufficient
condition on the cohomological properties of the constraint and gauge
generator differential operators.

Note that we restrict our attention to the geometric and
algebraic aspects of the constructions and systematically avoid
analytical details. In particular, the construction of the space of
solutions of a PDE system on a spacetime of dimension greater than one
(which corresponds to an ordinary mechanical system) requires a theory
of infinite dimensional differential geometry. We treat the minimal
amount of the needed infinite dimensional geometry in a formal way. On
the other hand, an honest attempt to present the relevant functional
analytical details can be found in Refs.~\citen{bsf,fr-bv,rejzner-thesis,bfr}.
Our results could then form the core identities of a future
investigation along similar lines that could extend them beyond the
formal level.

In Sec.~\ref{sec:lin-inhom}, we present some background material on
Green functions and hyperbolic PDEs. There we define the notion of Green
hyperbolicity that will be crucial in the later construction of the
Peierls formula and the study of its properties. The main body of this
review is contained in Sec.~\ref{sec:classquant}. There, we review the
formal differential geometry of the (possibly infinite dimensional)
solution spaces (Sec.~\ref{sec:formal-dg}). Also, we review the
covariant phase space formalism (Secs.~\ref{sec:var-sys},
\ref{sec:symp-pois}) and the Peierls formula (Sec.~\ref{sec:symp-pois}).
Finally, Sec.~\ref{sec:equiv} shows the equivalence between the
corresponding symplectic and Poisson structures by a generalization of
the Forger-Romero argument, under precise sufficient conditions
(Secs.~\ref{sec:constr}, \ref{sec:gauge} and~\ref{sec:constr-gf}). We
conclude with some examples in Sec.~\ref{sec:examples} and a discussion
in Sec.~\ref{sec:discuss}. In addition, some further needed background
information is given in the appendices, which includes the jet bundle
approach to PDEs, conservation laws and variational forms, as well as a
generalization of causal structure on smooth manifolds beyond Lorentzian
geometry. In particular, \ref{sec:jets-pdes} sets up some notation that
is used throughout the paper to describe PDE systems.

Finally, we conclude this introductory section with an extensive (though
still incomplete) historical overview of the literature on the covariant
phase space formalism and the Peierls formula. This historical material
may be safely skipped at first reading.

\subsection{Historical overview}
The canonical formalism in
mechanics\cite{whittaker,goldstein,abraham-marsden,arnold,souriau} (what we
would now call the construction of the symplectic and Poisson structures
on the phase space of a mechanical system) has a long history and is
most closely associated with the names of Hamilton and Jacobi. Though,
undoubtedly, its roots go back even to Lagrange. Its main
features are (a) the identification of the phase space with the
space of initial data and (b) the use of the Legendre transform to
determine special coordinates on the phase space in which the symplectic
form takes a certain canonical form (hence naming the formalism).
Because of the relation of Poisson brackets to quantization, the problem
of the quantization of fields in the early 20th century required the
translation (see for instance Refs.~\citen{rosenfeld,dirac-lect}) of the
canonical formalism from mechanics to field theories (multiple
independent variables instead of just one).

As already mentioned earlier, the main features of the canonical
formalism quickly began to clash with the relativistic nature and
spacetime covariance of field theories. That was already clear in the
works of Rosenfeld\cite{rosenfeld} and Dirac\cite{dirac-lect}. The possibility that
both of these unpleasant features could be avoided became realized very
slowly and is still not in mainstream use by theoretical physicists.
It appears that parts of it were rediscovered multiple times and it is
hard to trace them to any one source. Below, we discuss some key
references that made it clear that it is possible to construct the phase
space itself, as well as its symplectic and Poisson structures in a
fully covariant way, avoiding both features (a) and (b) of the canonical
formalism.

An important figure in some of the developments described below is that
of Souriau, despite lack of many explicit references to his work.
Perhaps his importance is not surprising, since he was one of the people
responsible for abstracting (in the 1960s) the modern notion of a
symplectic manifold as the appropriate arena for mechanics.\cite{souriau}
In particular, the name of Souriau is closely associated to identifying
the classical phase space with the set of solutions of the equations of
motion, rather than the set of initial data. Also, in Souriau's
book\cite{souriau} can be found a construction of the symplectic structure
on the phase space directly from the Lagrangian, which he attributes to
Lagrange himself.\cite{lagrange} Even though his book treated only
mechanical systems and not field theories, these ideas seem to have been
rather influential.

\subsubsection{Peierls formula}
The covariant construction of Poisson brackets in field theory can in
fact be traced to a single source: the seminal 1952 paper of
Peierls.\cite{peierls} In that paper, he introduced what is now known as
the \emph{Peierls bracket}, which we prefer to call the \emph{Peierls
formula}, as reviewed in Sec.~\ref{sec:formal-pois}. The formula for the
causal Green function as the difference of the retarded and advanced
Green functions, $\G = \G_+ - \G_-$, appeared there, though in a
somewhat implicit form. It is likely that Peierls was guided by
experience. Having seen the unequal time Poisson bracket (or rather the
quantum commutator) of point fields in many examples, computed using the
canonical method, but expressed in relativistically invariant form, he
was probably lead to a guess for its general formula. He showed that
this formula is in fact antisymmetric and is equivalent to the canonical
bracket for equal time fields, in the non-singular case (without gauge
invariance). He, however, did not give an independent proof of its non-degeneracy
or the Jacobi identity. Peierls showed that gauge invariance was not an
obstacle to defining the Poisson bracket by his method, as long as one
restricted oneself to gauge invariant observables. He also showed how
the formula extends to fermionic fields: each fermionic field can be
reduced to a bosonic one after multiplication by a formal anticommuting
parameter.

A somewhat later 1957 paper of Glaser, Lehmann~\& Zimmermann\cite{glz}
treated the perturbative expansion of interacting fields in terms of
retarded products of incoming free fields, in contrast to the usual
expansion in terms of time ordered products of asymptotic fields. As
their name suggests, retarded fields are defined using retarded Green
functions. They did not explicitly discuss Poisson structures, but their
work folds into the thread of ideas we are discussing in a slightly
different way. Their formula for the commutator of interacting fields
involved differences of of retarded and advanced Green functions, what
we would now call causal Green functions, which also occur in the
Peierls formula. This is not so surprising given the intimate
relationship between quantum commutators and Poisson brackets.

In 1960 came a paper of Segal\cite{segal-jmp} where he discussed the
canonical quantization of field theories with non-linear hyperbolic
equations of motion by identifying their phase space with the space of
solutions and endowing it with a Poisson structure in a covariant way.
His formula for the Poisson bracket also involved the causal Green
function $\G$. His construction appears to have been independent of
Peierls, but motivated much in the same way. In particular, he
constructs $\G$ not as the difference of retarded and advanced Green
functions, but as a distributional solution of the linearized equations
with specific initial conditions designed to reproduce the equal time
commutation relations. Further, though minor, developments of these
ideas appeared in a monograph\cite{segal-book} and in some conference
proceedings, including Ref.~\citen{segal-cargese}

Unfortunately, not many people paid attention to Peierls' paper.
Notable exceptions were Bergmann and DeWitt. In fact,
DeWitt\cite{dewitt-lect,dewitt-qft} quickly became an early adopter and
proponent. He can be said to be responsible for clarifying the role of
the causal Green function of the Jacobi equation (the linearized
Euler-Lagrange equation) in Peierls' construction and showing that it
can be consistently used with gauge fixing. Incidentally, he also
clarified the role of classical fermi fields in terms of anticommuting
Grassmann variables and thus seeded the germs of supergeometry.

In 1971, Steinmann\cite{steinmann} published a monograph where he adapted
the retarded products of Glaser, Lehmann~\& Zimmermann as a way of
formalizing renormalized perturbation theory within the context of
Axiomatic Quantum Field Theory.

DeWitt's formulation saw relatively few improvements until the early
'90s when Marolf\cite{marolf-thesis,marolf1,marolf2} (DeWitt's PhD
student at the time) realized that the Peierls formula can be taken
off-shell and define a Poisson bracket for off-shell observables. He
essentially showed that, by using the Peierls formula directly, the
(off-shell) field configuration space can be given the structure of a
degenerate (though regular) Poisson manifold. The symplectic leaves of
this Poisson structure are copies of the (on-shell) solution space with
the standard canonical symplectic structure. The following
interpretation is present though implicit in Marolf's papers:
given a Lagrangian density $\L[\phi] + \phi\cdot J$, where $\phi$
denotes the dynamical fields and $J$ the corresponding external
sources, the solution space corresponding to a fixed external source
profile $J$ is one leaf of the Poisson structure on the field
configuration space constructed from the source-less Lagrangian density
$\L[\phi]$. Unfortunately, most of Marolf's discussion is non-covariant,
as for simplicity it introduces a fixed time coordinate.

More recently, the Peierls formula became an important ingredient in the
construction of perturbative Algebraic Quantum Field Theory (pAQFT). Its
use came to prominence in the 2000s with
Refs.~\citen{df-peierls,brennecke-duetsch}. In these papers, cues were
taken partially from the perturbative QFT tradition of Refs.~\citen{glz,steinmann}, with
retarded, advanced and time-ordered products adapted to an off-shell
setting. At the same time, they realized that very similar formulas come
about at the classical level from the use of the off-shell Peierls
formula (as defined by Marolf) in perturbative classical field theory.
It was there that the use of the off-shell Peierls formula was
formulated in a systematic and covariant way. The first direct
demonstration of the Jacobi identity for the Peierls formula was given
in Ref.~\citen{brennecke-duetsch}. More recently, this formulation of
the off-shell Peierls formula has been extended to classical fermi
fields,\cite{rejzner-fermion} which paved the way for its inclusion in a
BV-BRST treatment of gauge theories in pAQFT.\cite{fr-bv,rejzner-thesis}

Another recent development is the incorporation of the off-shell Peierls
formula in a serious functional analytical effort to describe the
infinite dimensional spaces of field configurations and algebras of
observables using infinite dimensional differential
geometry.\cite{fr-bv,rejzner-thesis,bfr}

Finally, we should mention another recent paper\cite{hs-gauge} that has
a non-trivial overlap with this review. Its goals also include
clarifying the specific conditions needed to be satisfied by a linear
gauge theory that guarantee that the Peierls construction works. On the
other hand, their geometric set up is somewhat less general than ours
and they do not consider the relation with the covariant phase space
formalism in detail. Though, unlike here, they also treat fermi fields
and further consider quantization.

\subsubsection{Covariant phase space formalism}
Much of the impetus for the development of a covariant Poisson bracket
came from covariant quantization, or more precisely the need for a solid
classical analog of the unequal time commutation relations in QFT. On
the other hand, the development of the covariant phase space formalism
was spurred by the desire to understand conservation laws in field
theory, as well as trying to improve upon the canonical quantization
program of Dirac and Bergmann.

In the physics literature, the roots of this formalism, though in a
rather obscure form, can be found in the 1953 paper of Bergmann~\&
Schiller.\cite{bergmann-schiller} This paper was part of Bergmann's program
to study the implications of general covariance in General Relativity
(or any other second order covariant theory) for the structure of its
conservation laws, its stress energy tensor, and its canonical
quantization. There already appear formulas for what we call the
presymplectic current density and its potential.

This possibility of obtaining the symplectic structure of a theory
directly from the Lagrangian by the methods of Bergmann~\& Schiller
remained somewhat unknown, except possibly to a small group of experts.
For example, in 1962, Komar\cite{komar} (a student of Bergmann) used these
methods to define the symplectic structure on the space of initial data
on a null surface. Few other papers using this formalism appeared until
its apparently independent resurgence in the '80s.

The 1975 article by Ashtekar~\& Magnon\cite{ashtekar-magnon} used the
integrated symplectic potential current density as a symplectic form for
the Klein-Gordon field on curved spacetime, citing the work of
Segal\cite{segal-cargese} as a similar previous treatment in Minkowski
space, which goes back to the aforementioned Ref.~\citen{segal-jmp}.
However, it appears that the formalism of Ashtekar~\& Magnon was
privately inspired\cite{ashtekar-priv} by the ideas of Souriau, which
they generalized to field theories. Later Ashtekar also applied the same
formalism to general relativity.\cite{ashtekar-prl} Starting around this
time, perhaps again due to the influence of the ideas of Souriau, the
identification of the phase space with the space of solutions rather
than with that of initial data starts to become more prevalent.

Another independent appearance of the (pre)symplectic form as an
integral of the covariant (pre)symplectic current density is found in
the 1978 article of Friedman\cite{friedman} and a later article of
Friedman~\& Schutz.\cite{friedman-schutz} The examples considered there consisted
of a scalar field and of the combined system of linearized gravitational
and hydrodynamic modes describing a relativistic star. The latter
presymplectic form is degenerate due to general covariance
(gravitational gauge symmetry) and, in fact, its kernel was used in the
analysis of linear stability to discard possibly unstable unphysical
modes. Friedman did not cite any preceding sources, but apparently was
inspired\cite{friedman-priv} by some lectures on the covariant
treatment of conservation laws in variational theories (obviously
connected to symplectic structure by Noether's theorem) that were
delivered by Trautman at Chicago in 1971. The relevant content from
these lectures later appeared as Ref.~\citen{trautman}. Trautman's main
influence seems to have been the gradual refinement of the original
ideas of Noether\cite{noether} in connecting symmetries (including gauge
symmetries) with conservation laws in the calculus of variations with
multiple independent
variables.\cite{dedecker,takens,goldschmidt-sternberg} In the latter
formalism, the geometric structure that is closely related (but not
identical to) the covariant presymplectic (potential) current density is
the Poincar\'e-Cartan form.

Another independent instance of the covariant symplectic formalism,
though only for mechanical systems and not field theories, appeared in
the 1982 article of Henneaux\cite{henneaux-invprob} on the inverse problem
of the calculus of variations. No sources were cited in the article, but
apparently the main inspiration\cite{henneaux-priv} were the ideas of
Souriau.\cite{souriau}

The breakthrough point for a more widespread appreciation of the
covariant phase space formalism was the 1987 paper of Crnkovi\'c~\&
Witten\cite{cw}, which gave explicit covariant constructions for the
symplectic forms of scalar fields, Yang-Mills theory and General
Relativity. This paper was in fact an expansion of one section, where
this formalism was laid out in some generality, of an earlier paper of
Witten\cite{witten} on open string field theory. Witten did not cite any
preceding sources and it is not clear what his main influences were.

Almost simultaneously, very similar ideas were laid out by
Zuckerman.\cite{zuckerman} It is not completely clear what was Zuckerman's
main influence. It is likely to have been similar to Trautman's, as the
main preceding reference that he cites is the original article by
Noether.\cite{noether} He does, however cite\cite{sternberg,witten}
independent contemporary presentations of very similar ideas, including
a private letter of Deligne. Deligne's presentation of these ideas was
later recorded in the Ref.~\citen{deligne}, where only Zuckerman is cited
explicitly.

Soon thereafter, several reviews appear presenting the formalism in its
general and essentially modern form. Crnkovi\'c\cite{crnkovic} follows
Ref.~\citen{cw}, Lee~\& Wald follow Refs.~\citen{friedman,cw} and Ashtekar,
Bombelli~\& Reula\cite{abr} follow
Refs.~\citen{ashtekar-magnon,ashtekar-prl,souriau}. The Introduction in
Ref.~\citen{abr} contains further contemporary references. The paper by
Lee~\& Wald has to be singled out for its exceptional clarity of
presentation. The generality of the presentation was lacking
only the treatment of classical fermi fields, which appeared even in a
cursory way only in Ref.~\citen{deligne}. That was rectified only much later
in a paper of Hollands~\& Marolf.\cite{hollands-marolf}

In the mean time, some of the formal manipulations involved in deriving
the covariant presymplectic current density were formalized in the
framework of the variational bicomplex, as can be seen for instance in
Sec.~8.3 of Ref.~\citen{barnich-brandt}. It is essentially this presentation
that we review later in Sec.~\ref{sec:var-sys}.

\subsubsection{Equivalence}
While symplectic and Poisson structures are obviously related, as can be
seen from the above references, the two covariant formalisms that we have
described naturally appear in somewhat different problems. Thus, it is
not surprising that most researchers would prefer to use just one or
the other, without attempting to relate the two. Some have actually done
so, though, until recently only using the fact that they are both
equivalent (under appropriate assumptions of course) to the
corresponding canonical constructions.

Recall that in his original paper Peierls\cite{peierls} showed, by
explicit calculation, the equivalence of the Peierls formula with the
canonical Poisson bracket for non-singular field theories.
DeWitt\cite{dewitt-lect} extended that to the case of gauge fixed singular
theories, provided the canonical formalism was applied after gauge
fixing.

Even before that, having taken note of Peierls' paper, Bergmann's group
showed\cite{bgjn} that both the Bergmann-Schiller
and Peierls formulations of Poisson brackets are essentially equivalent
for non-singular field theories (those without gauge invariance). They
also compare the Peierls bracket with the Dirac bracket in singular
theories. Unfortunately, by modern standards their discussion is rather
obscure.

Much later, Barnich, Henneaux~\& Schomblond\cite{bhs} showed that the
Peierls bracket coincides with the Dirac bracket in the canonical
formalism even if both first and second class constraints are present
(hence including gauge theories). Their presentation is very clear. It
is also very insightful in the way that they included the covariant
phase space formalism. They showed that the canonical formalism is in
fact a special case of the covariant one, provided one uses Hamilton's
least action principle as the variational functional. Further, they
noted that the canonical variables can be introduced by adjoining
some auxiliary fields\cite{henneaux-elim} (the canonical momenta) to the
Lagrangian formulation, while showing that the symplectic structure is
invariant under the adjunction or elimination of auxiliary fields.

Finally, rather recently, a breakthrough appeared in a paper of
Forger~\& Romero.\cite{fr-pois} They managed to prove the equivalence of
the covariant constructions of the symplectic and Poisson structures in
an elegant and fully covariant way, thus without using the canonical
formalism as an intermediary. Their proof was restricted to the case of
a non-singular scalar field theory. It is their argument whose
generalization we present in expanded detail in the bulk of this review,
Sec.~\ref{sec:classquant}. It should be mentioned that Forger~\& Romero
also very clearly compared the geometric structure of the covariant
phase space formalism to the related but non-identical geometric
structure of the \emph{multisymplectic} formalism.

\section{Linear PDE theory}\label{sec:lin-inhom}
In this section, we describe a number of important technical results
about linear hyperbolic equations that will be crucial for the later
discussion of the Peierls bracket in Sec.~\ref{sec:classquant}. Here
we are mostly concerned with the linear algebra of the inhomogeneous
system of partial differential equation (PDE system)
\begin{equation}
	f[\phi] = \tilde{\alpha}^* ,
\end{equation}
on an $n$-dimensional spacetime manifold $M$, where $f\colon \Secs(F)
\to \Secs(\tilde{F}^*)$ is a linear partial differential operator acting
on smooth sections $\phi\in \Secs(F)$ of a \emph{field (vector) bundle}
$F\to M$ and taking values in the \emph{densitized dual} bundle
$\tilde{F}^*\to M$, where $\tilde{F}^* \cong F^*\otimes_M \Lambda^n M$
is the linear dual bundle $F^*$ tensored with the bundle of volume forms
on $M$.  Typically, $\tilde{\alpha}^* \in \Secs(F)$ is a compactly
supported \emph{dual density}. We choose to always have the PDE valued
in dual densities out of convenience, as will be evidenced later in
Sec.~\ref{sec:green-adj}. Keeping with the terminology of
\ref{sec:jets-pdes}, we use such an $(f,\tilde{F}^*)$ as our preferred
equation form for linear PDE systems.

We will consider the case where $f$ is hyperbolic and hence possesses
Green functions (Sec.~\ref{sec:green-hyper}). An important idea that is
often necessary to relate physical equations of motion to hyperbolic
equations is that of compatible constraints
(Sec.~\ref{sec:compat-constr}). Solution spaces can be conveniently
parametrized using special, causal Green functions
(Sec.~\ref{sec:caus-green-free} and~\ref{sec:caus-green}). Both the
differential operator $f$ and its Green functions have adjoints, which
are important in the definitions of various natural bilinear pairings
(Sec.~\ref{sec:green-adj}).

Sometimes we will refer to basic background information on jet
bundles, the interpretation of differential operators as maps between
jet bundles and the interpretation of PDEs as submanifolds of jet
bundles. The relevant information is summarized in \ref{sec:jets} and
\ref{sec:jets-pdes}. Also, hyperbolic PDEs naturally define a
generalized kind of causal structure on $M$. The necessary ideas and
definitions are summarized in \ref{sec:conal}, with the notation similar
to the standard one used in Lorentzian geometry. This causal structure
can be used to restrict the supports of field and dual density sections.
\begin{definition}
Consider a vector bundle $V\to M$. We define the following subspaces of
the space of sections $\Secs(V)$:
\begin{align}
	\Secs_0(V) &= \{ \phi\in\Secs(V)
		\mid \text{$\supp\phi$ is compact} \} , \\
	\Secs_+(V) &= \{ \phi\in\Secs(V)
		\mid \text{$\supp\phi$ is retarded} \} , \\
	\Secs_-(V) &= \{ \phi\in\Secs(V)
		\mid \text{$\supp\phi$ is advanced} \} , \\
	\Secs_{SC}(V) &= \{ \phi\in\Secs(V) \mid
		\text{$\supp\phi$ is spacelike compact} \} ,
\end{align}
where \emph{retarded support}, \emph{advanced support}, or
\emph{spacelike compact support} means, respectively, that $\supp\phi
\sso \overline{I^+(K)}$, $\supp\phi\sso \overline{I^-(K)}$, or
$\supp\phi\sso \overline{I(K)}$ for some compact $K\sso M$.  The
corresponding subspaces of the solution space $\S(F)$ of $f[\phi] = 0$
are denoted by
\begin{equation}
	\S_{0,\pm,SC}(F) = \S(F) \cap \Secs_{0,\pm,SC}(F) .
\end{equation}
\end{definition}

\subsection{Green hyperbolicity}\label{sec:green-hyper}
Below, we define the notion of a \emph{Green hyperbolic} PDE system as
one that possesses unique advanced and retarded Green functions. We will
rely heavily on the existence and properties of these Green functions in
later sections. The class of Green hyperbolic systems is quite large,
including for instance wave-like equations on globally hyperbolic
Lorentzian spacetimes\cite{bgp} as well as symmetric (or even
symmetrizable) hyperbolic systems (cf.~Refs.~\citen{baer-green}
and~\cite[Sec.4]{kh-caus}) that
satisfy a similar global causal condition. Note that the second class of
examples does not require a background Lorentzian metric to be
defined.

Before proceeding, we need the notion of a \emph{causal structure} (a
priori independent of any Lorentzian metric), with respect to which the
notions of advanced and retarded support will be defined. In the
literature on relativity, the two are usually introduced together.
However, a deeper investigation of hyperbolic PDE systems shows that the
notion of causality can be defined independently and intrinsically from
a given PDE. It so happens that, for equations with a d'Alambert-like
principal symbol, the causal relations deduced directly from the PDE, on
the one hand, and from the background Lorentzian metric, on the other,
actually coincide. The basic relevant notions and definitions are
summarized in \ref{sec:conal}. See Ref.~\cite[Secs.3,4]{kh-caus} for a more in depth
discussion.

\begin{definition}\label{def:green-hyp}
The PDE system $f[\phi] = 0$ is said to be \emph{Green hyperbolic} if
there exists a globally hyperbolic conal structure on $M$ such that (a)
the inhomogeneous equation $f[\phi_\pm] = \tilde{\alpha}^*$ is solvable
for $\tilde{\alpha}^*\in\Secs_\pm(\tilde{F}^*)$  and (b) a solution
$\phi_\pm\in \Secs_\pm(F)$ exists, is unique and satisfies the support
condition $\supp \phi_\pm \sse \overline{I^\pm(\supp\tilde{\alpha}^*)}$.
We denote the unique two-sided inverses by $\G_\pm\colon \Secs_\pm(\tilde{F}^*)
\to \Secs_\pm(F)$ and refer to them as the \emph{retarded} ($+$) and
\emph{advanced} ($-$) \emph{Green functions}. In adapted local
coordinates $(x^i,u^a)$ on $F$ and $(x^i,u_b)$ on $\tilde{F}^*$, where
$u^a(\phi(x)) = \phi^a(x)$, $u_b(\tilde{\alpha}^*(x)) = \alpha_b(x)$ and
$\d\tilde{x} = \d{x}^1\wedge\cdots\wedge\d{x}^n$, the Green functions
can be represented as integral kernels
\begin{equation}
	\phi^a_\pm(x)
	= (\G_\pm[\tilde{\alpha}^*])^a(x)
	= \int_M \G_\pm^{ab}(x;y) \alpha_b(y) \, \d\tilde{y} .
\end{equation}
\end{definition}
It is sufficient that $\G_\pm$ be defined on $\Secs_0(\tilde{F}^*)$. It
can then be extended to $\Secs_\pm(\tilde{F}^*)$ by an exhaustion
argument (Cor.5 in Ref.~\citen{bf-lcqft}). Ideally, the Green functions
should be well-defined distributions (be continuous in the appropriate
function space topology), but we will not discuss this issue here and
instead concentrate on their algebraic and geometric properties.

\begin{remark}
In this paper, we rely heavily on the notion of Green hyperbolicity. On
the other hand, many PDE references only treat well-posedness of the
Cauchy problem of homogeneous equations and do not address the
inhomogeneous problem. For instance, Cauchy well-posedness of linear
symmetric hyperbolic systems is established in~\cite[Ch.7]{ringstroem},
but the existence of Green functions is not addressed. Fortunately,
there is an argument in the classic PDE literature, known as
\emph{Duhamel's principle}, that essentially establishes the equivalence
between hyperbolic systems with a well-posed Cauchy problem and Green
hyperbolic systems. Usually, this argument is discussed only for
specific examples, but in principle it works quite generally. Thus, we
can appeal to this very large class of PDE systems when considering
examples, rather than restricting ourselves only to wave-like equations
on Lorentzian manifolds for which Green hyperbolicity is well
established.\cite{bgp,baer-green,waldmann-pde}
\end{remark}

\subsection{Compatible constraints}\label{sec:compat-constr}
Consider a linear PDE system on the field bundle $F\to M$ that consists
of
\begin{equation}
	f[\phi] = 0, \quad c[\phi] = 0 ,
\end{equation}
where $f\colon \Secs(F) \to \Secs(\tilde{F}^*)$ is a hyperbolic partial
differential operator and $c\colon \Secs(F) \to \Secs(E)$ is another
partial differential operator valued in a vector bundle $E\to M$. We
refer to $f[\phi] = 0$ as the \emph{hyperbolic subsystem} and to
$c[\phi] = 0$ as the \emph{constraints subsystem}, with $E\to M$ the
\emph{constraints bundle}. The equation form of the total system is
$(f\oplus c, \tilde{F}^*\oplus E)$.

The constraints are said to be \emph{hyperbolically integrable} if there
exists a pair of linear differential operators $h\colon \Secs(E) \to
\Secs(\tilde{E}^*)$ and $q\colon \Secs(\tilde{F}^*) \to
\Secs(\tilde{E}^*)$ that satisfy the
identity $h\circ c = q\circ f$, or 
\begin{equation}
	h[c[\phi]] = q[f[\phi]] ,
\end{equation}
for any $\phi\in \Secs(F)$, where the operator $h$ itself is hyperbolic.
The existence of this identity implies that the vanishing of the
constraints $c[\phi] = 0$ at some initial time implies that $c[\phi] =
0$ everywhere on $M$, provided $f[\phi] = 0$. We call the PDE $h[\psi] =
0$ on $E\to M$ the \emph{consistency subsystem} and the joint system
$f[\phi] = 0$, $h[\psi] = 0$ on $F\oplus_M E\to M$ the \emph{compound
system}. We write the corresponding equation form as $(f\oplus h,
\tilde{F}^* \oplus_M \tilde{E}^*)$.

If the above conditions are satisfied, the PDE system $f[\phi] = 0$,
$c[\phi] = 0$ is said to be \emph{hyperbolic with constraints}. Whenever
we refer to a causal structure induced by a hyperbolic system with
constraints, we actually mean the one induced by the corresponding
hyperbolic compound system.

\subsection{Causal Green function (without constraints)}\label{sec:caus-green-free}
Now that we are sure to have access to the retarded/advanced green
functions $\G_\pm$ for the Green hyperbolic system $f[\phi]
= 0$, we can define the so-called \emph{causal Green function}
\begin{equation}\label{eq:caus-green}
	\G = \G_+ - \G_- .
\end{equation}
This new Green function helps to conveniently parametrize the space of
solutions $\S_{SC}(F) \cong \ker f \sso \Secs_{SC}(F)$ by featuring in
the following
\begin{proposition}\label{prp:exact}
The sequence
\begin{equation}\label{eq:hyp-seq}
\vcenter{\xymatrix{
	0
		\ar[r] &
	\Secs_0(F)
		\ar[r]^{f} &
	\Secs_0(\tilde{F}^*)
		\ar[r]^{\G} &
	\Secs_{SC}(F)
		\ar[r]^{f} &
	\Secs_{SC}(\tilde{F}^*)
		\ar[r] &
	0 ,
}}
\end{equation}
is exact (in the sense of linear algebra).
\end{proposition}
That is, the image of each map coincides with the kernel of the next
map. The proof for wave-like equations, which is given
in~\cite[Thm.3.4.7]{bgp} and~\cite[Lem.3.2.1]{wald-qft} and
unfortunately excludes the final surjection, directly carries through to
the Green hyperbolic case. The final surjection is covered by the proof
of Cor.5 in Ref.~\citen{bf-lcqft}. A complete proof actually follows
from the identities given in Lem.~\ref{lem:exsplit} below and the fact
that $f$ is invertible on $\Secs_\pm(F)$, and hence a fortiori injective
on $\Secs_0(F)$.

We can interpret the above proposition in the following way. Since
$\S_{SC}(F) \cong \im \G$, we can express any solution to the
homogeneous problem as $\phi = \G[\tilde{\alpha}^*]$, where $\alpha\in
\Secs_0(\tilde{F}^*)$ is some smooth dual density of compact support.
Equivalently, by exactness, $\S_{SC}(F) \cong \coker f =
\Secs_0(\tilde{F}^*) / \im f$.  Also, since $\Secs_{SC}(F) \cong \im f$,
for any dual density $\tilde{\alpha}^*$ with spacelike compact support,
there exists a solution $\phi$ with spacelike compact support of the
inhomogeneous problem $f[\phi] = \tilde{\alpha}^*$.

\begin{definition}\label{def:adapt-pu}
Consider one Cauchy surface $\Sigma\sso M$ and two more Cauchy surfaces
$\Sigma^\pm\sso M$ to the past and future of $\Sigma$, where
$\Sigma^\pm\sso I^\pm(\Sigma)$, and let $S^\pm = I^\pm(\Sigma^\mp)$. Let
$\{\chi_+,\chi_-\}$ be a partition of unity adapted to the open cover
$\{S^+,S^-\}$ of $M$, that is, $\chi_+ + \chi_- = 1$ and $\supp \chi_\pm
\sso S^\pm$. We call $\{\chi_+, \chi_-\}$ a \emph{partition of unity
adapted to the Cauchy surface $\Sigma$}.
\end{definition}

\begin{lemma}\label{lem:exsplit}
The exact sequence of Prp.~\ref{prp:exact} splits at
\begin{equation}
	\Secs_0(\tilde{F}^*) \cong \Secs_0(F) \oplus \S_{SC}(F)
	\quad\text{and}\quad
	\Secs_{SC}(F) \cong \S_{SC}(F) \oplus \Secs_{SC}(\tilde{F}^*) .
\end{equation}
Given a partition of unity $\{\chi_+,\chi_-\}$ adapted to a Cauchy
surface $\Sigma$, there exist (noncanonical) splitting maps
\begin{align}
	f_\chi \colon & \im\G \to\Secs_0(\tilde{F}^*) , &
		f_\chi[\phi] &= \pm f^\pm_\chi[\phi] = \pm f[\chi_\pm\phi] , \\
	\G_\chi \colon & \Secs_{SC}(\tilde{F}^*)\to \Secs_{SC}(F) , &
		\G_\chi[\tilde{\alpha}^*]
		&= \G_+[\chi_+\tilde{\alpha}^*]
			+ \G_-[\chi_-\tilde{\alpha}^*] .
\end{align}
\end{lemma}
\begin{proof}
Note that these splitting maps are not canonical, as they depend on the
choice of a Cauchy surface and a partition of unity adapted to it.

When $\phi\in\S_{SC}(F) \cong \im\G$, it is clear that $f^\pm_\chi[\phi]
= f[\chi_\pm\phi]$ does in fact have compact support, as $\supp\phi$ is
spacelike compact while $\supp\d\chi_\pm \linebreak \sso S^+\cap S^-$ is
timelike compact,%
	\footnote{A set is \emph{timelike compact} if its intersection with
	every spacelike compact set is compact.\cite{sanders-tc}} %
and $f[\chi_\pm\phi] \ne 0$ only on $\supp\phi \cap \supp\d\chi_\pm$,
which is by definition compact.  Also, since $\d(\chi_+ + \chi_-) = 0$,
we have $f^+_\chi[\phi] + f^-_\chi[\phi] = 0$, which means that the map
$f_\chi = \pm f^\pm_\chi$ is well defined. On the other hand, we have
$\G_\pm[f[\chi_\pm\phi]] = \chi_\pm\phi$ from the uniqueness of
solutions to the inhomogeneous problem with retarded/advanced support.
The definition of the causal Green function then immediately implies
that $\G\circ f_\chi = \pm\id$ on $\S_{SC}(F)$. Also, a direct
calculation shows that $f\circ \G_\chi = \id$ on
$\Secs_{SC}(\tilde{F}^*)$:
\begin{equation}
	f\circ \G_\chi [\tilde{\alpha}^*]
	= \chi_+\tilde{\alpha}^* + \chi_-\tilde{\alpha}^*
	= \tilde{\alpha}^* .
\end{equation}
This concludes the proof. %\qed
\end{proof}

\subsection{Causal Green function (with constraints)}\label{sec:caus-green}
We will not discuss the most general kind of constraints and restrict
our attention only to \emph{parametrizable} ones. By the term
parametrizable, we mean that there exist an additional vector bundle
$E'\to M$ and additional differential operators $h'$, $c'$ and $q'$,
which fit into the following commutative diagram
\begin{equation}
\vcenter{\xymatrix{
	\Secs(E') \ar[d]^{h'} \ar[r]^{c'} &
		\Secs(F) \ar[d]^{f} \ar[r]^{c} &
		\Secs(E) \ar[d]^{h} \\
	\Secs(\tilde{E}^{\prime*}) \ar[r]^{q'} &
		\Secs(\tilde{F}^*) \ar[r]^{q} &
		\Secs(\tilde{E}^*)
}}
\end{equation}
such that $h'$ is hyperbolic, and that the horizontal rows form
complexes of differential operators ($c\circ c' = 0$ and $q\circ q' =
0$) that are \emph{formally exact} (\ref{sec:formal-exact}). Since
both $h$ and $h'$ are hyperbolic, we can define their retarded/advanced
Green functions, $\H_\pm$ and $\H'_\pm$, as well as their causal Green
functions, $\H = \H_+ - \H_-$ and $\H' = \H'_+ - \H'_-$. All these
operators then fit into the following commutative diagram:
\begin{equation}\label{eq:chyp-seq}
\vcenter{\xymatrix{
	0
		\ar[r] &
	\Secs_0(E')
		\ar[d]^{c'} \ar[r]^{h'} &
	\Secs_0(\tilde{E}^{\prime*})
		\ar[d]^{q'} \ar[r]^{\H'} &
	\Secs_{SC}(E')
		\ar[d]^{c'} \ar[r]^{h'} &
	\Secs_{SC}(\tilde{E}^{\prime*})
		\ar[d]^{q'} \ar[r] &
	0  \\
	0
		\ar[r] &
	\Secs_0(F)
		\ar[d]^{c} \ar[r]^{f} &
	\Secs_0(\tilde{F}^*)
		\ar[d]^{q} \ar[r]^{\G} &
	\Secs_{SC}(F)
		\ar[d]^{c} \ar[r]^{f} &
	\Secs_{SC}(\tilde{F}^*)
		\ar[d]^{q} \ar[r] &
	0  \\
	0
		\ar[r] &
	\Secs_0(E)
		\ar[r]^{h} &
	\Secs_0(\tilde{E}^*)
		\ar[r]^{\H} &
	\Secs_{SC}(E)
		\ar[r]^{h} &
	\Secs_{SC}(\tilde{E}^*)
		\ar[r] &
	0
}}
\end{equation}
Note that the causally restricted supports in the above diagram should
be defined with respect to a causal structure that is defined by the
total compound system with equation form $(h'\oplus f \oplus h,
\tilde{E}^{\prime*} \oplus_M \tilde{F}^* \oplus_M \tilde{E}^*)$.

\begin{lemma}\label{lem:inhom-constr}
The retarded/advanced inhomogeneous problem
\begin{align}
	f[\phi] &= \tilde{\beta}^* , \\
	c[\phi] &= \gamma ,
\end{align}
with $\tilde{\beta}^*\in \Secs_0(\tilde{F}^*)$ and $\gamma \in
\Secs_\pm(E)$, is solvable for $\phi\in \Secs_\pm(F)$ iff $h[\gamma] =
q[\tilde{\beta}^*]$.
\end{lemma}
\begin{proof}
In one direction, if $\phi$ is the desired solution, then $h[\gamma] =
h[c[\phi]] = q[f[\phi]] = q[\tilde{\beta}^*]$. In the other direction,
let $\phi = \G_\pm[\tilde{\beta^*}]$. We obviously have $f[\phi] =
\tilde{\beta}^*$. It remains to check
\begin{equation}
	c[\phi]
	= c[\G_+[\tilde{\beta}^*]]
	= \H_+[q[\tilde{\beta}^*]]
	= \H_+[h[\gamma]]
	= \gamma .
\end{equation}
This concludes the proof. %\qed
\end{proof}

It is convenient to state here a lemma concerning formally exact complexes of
differential operators, which shall be referred to in later sections.
\begin{lemma}\label{lem:fec-fact}
Suppose that linear differential operators $c'\colon \Secs(E') \to
\Secs(F)$ and $c\colon \Secs(F) \to \Secs(E)$ form a formally exact
complex. Then any linear differential operator $l\colon \Secs(F) \to
\Secs(L)$ such that $l\circ c' = 0$ factors as $l = l_c \circ c$, for
some linear differential operator $l_c\colon \Secs(E) \to \Secs(L)$.
Similarly, any linear differential operator $r\colon \Secs(R) \to
\Secs(F)$ such that $c\circ r = 0$ factors as $r = c' \circ r_c$, for
some linear differential operator $r_c\colon \Secs(R) \to \Secs(E')$.
\end{lemma}
\begin{proof}
We represent all differential operators as maps from appropriate jet
bundles (see Sec.~\ref{sec:jets-pdes}). The fact that the differential
operators $c'\colon J^\oo E' \to J^\oo F$ and $c\colon J^\oo F \to E$
form a formally exact complex shows that the prolongations $p^\oo c'\colon
J^\oo E' \to J^\oo F$ and $p^\oo c\colon J^\oo F \to J^\oo E$ compose
into an exact sequence of vector bundle maps. Hence, the image of $p^\oo
c'$ coincides with the kernel of $p^\oo c$. By hypothesis, the kernel of
the linear bundle map $l\colon J^\oo F \to L$ contains $\im p^\oo c'$,
while the image of the linear bundle map $p^\oo r\colon J^\oo R \to
J^\oo F$ is contained in the image of $p^\oo c'$. Therefore, desired
factorization formulas follow straight forwardly from linear algebra.
\end{proof}
Such arguments are common in the formal theory of PDEs and can even be
generalized to the nonlinear
setting.\cite{goldschmidt,goldschmidt-lin,pommaret}

\subsection{Pairings and adjoints}\label{sec:green-adj}
We conclude this section by remarking the identities
\begin{equation}\label{eq:adj-id}
	(\G_\pm)^* = \G^*_\mp ,
\end{equation}
where on the left hand side $(\G_\pm)^*$ denotes the adjoint of the
retarded/advanced Green function $\G_\pm$ of the equation $f[\phi]=0$,
and on the right hand side $\G^*_\mp$ denotes the advanced/retarded
Green function of the adjoint equation $f^*[\phi] = 0$. Note that taking
the adjoint flips the support between retarded and advanced.
\begin{definition}\label{def:green-form}
Given two differential operators $f,f^*\colon \Secs(F) \to
\Secs(\tilde{F}^*)$ are said to be mutually \emph{adjoint} if there
exists a bilinear differential operator $G\colon \Secs(F)\times \Secs(F)
\to \Forms^{n-1}(M)$ such that
\begin{equation}\label{eq:green-form}
	f[\phi]\cdot \psi - \phi\cdot f^*[\psi]
	= \d G[\phi,\psi]
\end{equation}
for any sections $\phi,\psi\colon M\to F$. The $(n-1)$-form valued
bilinear differential operator $G[\phi,\psi]$ is called a \emph{Green
form} associated to $f$ and $f^*$~\cite[\textsection IV.5]{palais},
\cite[\textsection V.1.3]{ch1}. Note that Eq.~\eqref{eq:green-form}
defines $G$ up to the addition of an exact
form,\cite{anderson-big,gms} $G\sim G+\d{H}$.
Denote by $[G]$ the uniquely defined equivalence class modulo exact
local bilinear forms $\d{H}[-,-]$.
\end{definition}

Note that we may introduce a natural pairing between elements $\phi
\in \Secs(F)$ and $\tilde{\alpha} \in \Secs(\tilde{F}^*)$ given by
\begin{equation}
	\langle \phi, \tilde{\alpha}^* \rangle =
	\langle \tilde{\alpha}^*, \phi \rangle =
	\int \phi\cdot \tilde{\alpha}^* .
\end{equation}
The pairing is only partially defined. That is, there exist arguments
for which the integral does not converge. For simplicity, we will
consider it only for those those pairs of sections for which the
integrand $\phi\cdot \tilde{\alpha}^*$ has compact support. This pairing
is non-degenerate in either argument, as follows from the standard
argument of the fundamental lemma of the calculus of
variations~\cite[\textsection IV.3.1]{ch1}. It will play an important
role in Secs.~\ref{sec:tt*-conf}, \ref{sec:tt*-sols}. It is easy to
show that the formal adjoint $f^*$ coincides with the adjoint of $f$
with respect to this natural pairing: $\langle f[\phi], \psi \rangle =
\langle \phi, f^*[\psi] \rangle$. This natural pairing allows us to
define adjoints for integral operators like Green functions, namely
$\langle \G_\pm[\tilde{\alpha}^*], \tilde{\beta}^* \rangle = \langle
\tilde{\alpha}^*, (G_\pm)^*[\tilde{\beta}^*] \rangle$.

It is straight forward that that the natural pairing $\langle -,
-\rangle$ is non-degenerate on the spaces $\Secs_\pm(F)\times
\Secs_\mp(\tilde{F}^*)$. And since, the identities $f\circ \G_\pm =
\G_\pm \circ f = \id$ hold on $\Secs_{\pm}(F)$, it is now easy to verify
the adjoint identities~\eqref{eq:adj-id} since
\begin{align}
	\int_M \phi_\mp \cdot f\circ \G_\pm[\tilde{\alpha}^*_\pm]
	&= \langle \phi_\mp, \tilde{\alpha}^*_\pm \rangle
		= \int_M (\G_\pm)^*\circ f^*[\phi_\mp] \cdot \tilde{\alpha}^*_\pm \\
	\int_M \G_\pm\circ f[\phi_\pm] \cdot \tilde{\alpha}^*_\mp ,
	&= \langle \phi_\pm, \tilde{\alpha}^*_\mp \rangle
		= \int_M \phi_\mp \cdot f^*\circ (\G_\pm)^*[\tilde{\alpha}^*_\mp] ,
\end{align}
for any $\phi_\pm \in \Secs_\pm(F)$ and $\tilde{\alpha}^*_\pm \in
\Secs_\pm(F)$. The causal Green functions then satisfy $(\G)^* = -\G^*$,
where $\G^*$ is the causal Green function for $f^*$.

\begin{definition}\label{def:green-pairing}
Let $\iota\colon \Sigma\sso M$ be a future oriented, Cauchy surface. The
\emph{Green pairing} $\langle -,- \rangle_G$ between a solution $\phi
\in \Secs_{SC}(F)$ of $f[\phi] = 0$ and a solution $\psi \in
\Secs_{SC}(F)$ of $f^*[\psi] = 0$ is given by
\begin{equation}
	\langle \phi, \psi \rangle_G
	= \int_\Sigma \iota^* G[\phi,\psi]
\end{equation}
\end{definition}
\begin{lemma}\label{lem:green-indep}
Following the notation of Def.~\ref{def:green-form}, the Green pairing
$\langle -,- \rangle_G$ depends only on the equivalence class $[G]$ of
$G$ and is independent of $\Sigma$.
\end{lemma}
\begin{proof}
Any two representatives $G_1$ and $G_2$ of $[G]$ will differ by an exact
term $\d{H}$, with $H[-,-]$ a bilinear bidifferential operator.
Therefore, the integrands $\iota^* G_i[\varphi_1,\varphi_2]$ will differ
by the exact term $\d\iota^* H[\varphi_1,\varphi_2]$, with necessarily
compact support. Therefore, since $\Sigma$ has no boundary, we can use
any representative of $[G]$ to evaluate the pairing.

Now, let $\iota'\colon \Sigma' \sso M$ be another Cauchy surface and let
$S\sse M$ be such that $\del S = \Sigma' - \Sigma$. Then an application
of Stokes' theorem shows the following:
\begin{equation}
	\int_{\Sigma'} \iota^{\prime*} G[\phi,\psi]
	- \int_{\Sigma} \iota^* G[\phi,\psi]
	= \int_S \d G[\phi,\psi]
	= \int_S (f[\phi]\cdot \psi - \phi\cdot f^*[\psi])
	= 0 ,
\end{equation}
where we used the solution properties $f[\phi] = 0$ and $f^*[\psi] = 0$.
This shows the independence of the Green pairing from the choice of a
Cauchy surface. %\qed
\end{proof}

\begin{lemma}\label{lem:green-nondegen}
The Green pairing has the following alternative forms:
\begin{equation}
	\langle \psi, \xi \rangle_G
	= \int_M \tilde{\alpha}^* \cdot \xi
	= \int_M \tilde{\alpha}^* \cdot \G^*[\tilde{\beta}^*]
	= -\int_M \G[\tilde{\alpha}^*] \cdot \tilde{\beta}^*
	= -\int_M \psi \cdot \tilde{\beta}^* ,
\end{equation}
where $\phi = \G[\tilde{\alpha}^*]$ and $\xi = \G^*[\tilde{\beta}^*]$,
$\tilde{\alpha}^*, \tilde{\beta}^* \in \Secs_0(\tilde{F}^*)$. Moreover,
it is non-degenerate.
\end{lemma}
Note that the following proof was partly inspired by Sec.~3.3 of
Ref.~\citen{fr-pois} and Lem.~3.2.1 of Ref.~\citen{wald-qft}.  It has
the same skeletal structure as one of the main technical lemmas
(Lem.~\ref{lem:bOmega}) of Sec.~\ref{sec:classquant}.
\begin{proof}
Consider a future oriented Cauchy surface $\iota\colon \Sigma \sso M$
(see \ref{sec:conal}) and a partition of unity $\{\chi_\pm\}$ adapted to
it. Then, recalling the notation for the splitting maps in
Lem.~\ref{lem:exsplit}, direct calculation gives
\begin{align}
	\langle \psi, \xi \rangle_G
	&= \int_\Sigma \iota^* G[\psi,\xi]
	= \sum_{\pm} \int_\Sigma \iota^* G[\psi, \chi_\pm\xi] \\
	&= \sum_{\pm} \pm \int_{I^{\mp}(\Sigma)} \d G[\psi,\chi_\pm\xi] \\
	&= \sum_{\pm} \pm \int_{I^{\mp}(\Sigma)}
		(f[\psi]\cdot (\chi_\pm\xi) - \psi\cdot f^*[\chi_\pm\xi]) \\
	&= -\int_{I^-(\Sigma)} \psi\cdot f^*_\chi[\xi]
		-\int_{I^+(\Sigma)} \psi\cdot f^*_\chi[\xi]
	= -\int_M \psi \cdot  f^*_\chi[\xi] \\
	&= -\int_M \G[\tilde{\alpha}^*] \cdot  f^*_\chi[\xi]
	= \int_M \tilde{\alpha}^* \cdot (\G^*\circ f^*_\chi)[\xi] \\
	&= \int_M \tilde{\alpha}^* \cdot \xi
	= \int_M \tilde{\alpha}^* \cdot \G^*[\tilde{\beta}^*] \\
	&= -\int_M \G[\tilde{\alpha}^*] \cdot \tilde{\beta}^*
	= -\int_M \psi \cdot \tilde{\beta}^* .
\end{align}
Appealing to the exact sequence of Prp.~\ref{prp:exact}, the dual
densities $\tilde{\alpha}^*$ and $\tilde{\beta}^*$ are only defined up
to an element of $\im f$ and $\im f^*$ respectively. However, the
formulas show that this freedom does not affect the result.

Now, based on the formula $\langle \psi, \xi \rangle_G = - \int_M \psi
\cdot \tilde{\beta}^*$, the fact that $\tilde{\beta}^*$ could be
arbitrary and the non-degeneracy of the natural pairing $\langle -,-
\rangle$ we can see that $\langle -,- \rangle_G$ must be non-degenerate
in its second argument. The same reasoning establishes non-degeneracy
in the first argument as well. %\qed
\end{proof}

\section{Covariant phase space formalism, Peierls formula}
\label{sec:classquant}
This section constitutes the main body of this review. At this point it
is helpful to at least skim the contents of \ref{sec:jets} and
\ref{sec:jets-pdes}, as they summarize relevant concepts and notation.
Its culmination is a precise set of conditions (Secs.~\ref{sec:constr},
\ref{sec:gauge} and~\ref{sec:constr-gf}) that are sufficient to
establish the validity of the covariant phase space formalism and the
Peierls formula constructions, respectively, of the symplectic
(Sec.~\ref{sec:formal-symp}) and Poisson (Sec.~\ref{sec:formal-pois})
structures on the phase space of a classical field theory, and their
equivalence via a generalized Forger-Romero\cite{fr-pois} argument
(Sec.~\ref{sec:symp-pois}).

Below, we study their argument in depth and generalize it to include
field theories with constraints and gauge invariance.
Sec.~\ref{sec:var-sys} defines variational PDE systems and shows how the
covariant symplectic formalism arises from their geometry.
Sec.~\ref{sec:formal-dg} discusses the differential geometry of the
space of solutions of the PDE system. Since the focus of this work is
more geometrical than analytical, we avoid a detailed discussion of the
subtleties of infinite dimensional manifolds. Instead, we define
so-called formal tangent and cotangent spaces to the solution space and
then fix a particular background solution, so that we need only consider
a single formal tangent and cotangent fiber at that point of the phase
space. This is sufficient for defining the formal symplectic form and
the Poisson bivector and proving their equivalence. An in-depth
discussion of the functional analytical details that go into defining
the necessary infinite dimensional differential geometry can be found in
Refs.~\citen{bsf,fr-bv,rejzner-thesis,bfr}.

\subsection{Variational systems}\label{sec:var-sys}
Consider a field vector bundle $F\to M$ over an $n$-dimensional manifold
$M$. A local action functional of order $k$ on $F\to M$ is a function
$S[\phi]$ of sections $\phi\colon M \to F$,
\begin{equation}
	S[\phi] = \int_M (j^k\phi)^*\L,
\end{equation}
where $\L$, the Lagrangian density, is a section of the bundle
$(\Lambda^n M)^k\to J^kF$ densities, which could depend on jet
coordinates of order up to $k$ (see \ref{sec:jets}). The Lagrangian
density is called local because, given a section $\phi$ and local
coordinates $(x^i,u^a_I)$ on $J^kF$, the pullback at $x\in M$ can be
written as 
\begin{equation}
	(j^k\phi)^*\L(x) = \L(x^i,\del_I\phi^a(x)) ,
\end{equation}
which depends only on $x$ and on the derivatives of $\phi$ at $x$ up to
order $k$. For the most part, the integral over $M$ can be considered
formal, since all the necessary properties will be derived from $\L$. On
the other hand, the finiteness of $S[\phi]$ or related quantities may be
important while discussing boundary conditions in spacetimes with
non-compact spatial extent. However, we will not discuss these issues
below.

Recall that \ref{sec:jets} introduces the variational
bicomplex $\Forms^{h,v}(F)$ of vertically and horizontally graded
differential forms on $J^\oo F$. Below, we use the notation introduced
there. A Lagrangian density is then an element
$\L\in\Forms^{n,0}(F)$ that can be projected to $J^kF$. Incidentally the
usual variational derivative of variational calculus can be put into
direct correspondence with the vertical differential $\dv$ on this
complex, which is how the name \emph{variational bicomplex} was
established.\cite{anderson-big,anderson-small}

Let $(x^i,u^a_I)$ be a set of adapted coordinates on the $\oo$-jet
bundle $J^\oo F$, where all the following calculations can be lifted. Any
result that depends only on jets of finite order can then be projected
onto the appropriate finite dimensional jet bundle. Using the
integration by parts identity~\eqref{eq:byparts} if necessary, we can
always write the first vertical variation of the Lagrangian density as
\begin{equation}\label{eq:dvL}
	\dv \L = \EL_a\wedge\dv{u^a} - \dh\theta.
\end{equation}
All terms proportional to $\dv{u^a_I}$, $|I|>0$, have been absorbed into
$\dh\theta$. In the course of the performing the integrations by parts,
$\EL_a$ can acquire dependence on jets up to order $2k$ (see
\ref{sec:jets-pdes}), and $\theta$ on jets up to order $2k-1$. Note that
$\EL_a=0$ are the \emph{Euler-Lagrange equations} associated with the
action functional $S[\phi]$ or the Lagrangian density $\L$. We can
identify the form $\EL_a\wedge\dv u^a$ with a possibly non-linear
differential operator $\EL\colon \Secs(F) \to \Secs(\tilde{F}^*)$, or
equivalently a bundle morphism $\EL\colon J^{2k}F \to \tilde{F}^*$.
Therefore, $(\EL,\tilde{F}^*)$ is an equation form of a PDE system
$\E_\EL\sso J^{2k}F$ on $F$ of order $2k$. A PDE system with an equation
form given by Euler-Lagrange equations of a Lagrangian density is said
to be \emph{variational}. Also, the form $\theta$ is an element of
$\Forms^{n-1,1}(F)$, projectable to $J^{2k-1}F$. It is referred to as
the \emph{presymplectic potential current density}. Applying the
vertical exterior differential to $\theta$ we obtain the
\emph{presymplectic current density} (or the \emph{presymplectic current
density defined by $\L$} if the extra precision is necessary):
\begin{equation}\label{eq:omega-def}
	\omega = \dv \theta,
\end{equation}
with $\omega\in \Forms^{n-1,2}(F)$. This terminology implies that
$\omega$ can be integrated over a codim-$1$ spacetime surface to
construct a presymplectic form (Sec.~\ref{sec:formal-symp}).  Such a
form on the solution space is referred to as \emph{local}. This method
of constructing a symplectic form on the phase space of classical field
theory is sometimes referred to as the \emph{covariant phase space
method}.\cite{lw,cw,abr,zuckerman}

The following lemma is an easy consequence of the definition of
$\omega$.
\begin{lemma}\label{lem:omega-cl}
The form $\omega\in\Forms^{n-1,2}(F)$ defined in
Eq.~\eqref{eq:omega-def} is both horizontally and vertically closed when
pulled back to $\iota_\oo\colon \E^\oo_\EL\sse J^\oo F$:
\begin{align}
	\dh \iota_\oo^* \omega &= 0 , \\
	\dv \iota_\oo^* \omega &= 0 .
\end{align}
\end{lemma}
\begin{proof}
The horizontal and vertical differentials on $\E^\oo_\EL$ are defined by
pullback along $\iota_\oo$, that is, $\dh \iota_\oo^* = \iota_\oo^* \dh$
and $\dv \iota_\oo^* = \iota_\oo^* \dv$. Since $\omega=\dv\theta$ is
already vertically closed on $J^\oo F$, it is a fortiori vertically
closed on $\E^\oo_\EL$. The rest is a consequence of the nilpotence and
anti-commutativity of $\dh$ and $\dv$:
\begin{align}
	0 = \dv^2\L
		&= \dv\EL_a\wedge\dv{u^a} - \dv\dh\theta, \\
	\dh\omega
		&= \dh\dv\theta
		= -\dv \EL_a\wedge\dv{u^a} , \\
	\dh\iota_\oo^*\omega
		&= \iota_\oo^*\dh\omega
		= -\iota_\oo^* \dv\EL_a \wedge \dv u^a
		= 0 ,
\end{align}
since $\EL_a$ and $\dv\EL_a$ generate the differential ideal in
$\Forms^*(J^\oo F)$ that is annihilated by the pullback $\iota_\oo^*$. %\qed
\end{proof}
In fact, we will promote the name \emph{presymplectic current density}
to any form satisfying these properties.
\begin{definition}
It is interesting to note that the existence of a presymplectic current
density compatible with a PDE $\E$ is almost equivalent to $\E$ being
variational.\cite{kh-invprob}

Given a PDE system $\iota\colon \E \sso J^kF$ we call a form $\omega$ a
\emph{presymplectic current density compatible with $\E$} if $\omega\in
\Forms^{n-1,2}(F)$ and it is both horizontally and
vertically closed on solutions:
\begin{align}
	\dh \iota_\oo^* \omega &= 0 , \\
	\dv \iota_\oo^* \omega &= 0 .
\end{align}
\end{definition}
The particular form $\omega$ defined by Eq.~\eqref{eq:omega-def} will be
referred to as the presymplectic current density associated to or
obtained from the Lagrangian density $\L$, if there is any potential
confusion.

\subsection{Formal differential geometry of solution spaces}
\label{sec:formal-dg}
Before describing the symplectic and Poisson structures on the space of
solutions, we should say something about the differential geometry of
the manifold of solutions of a PDE system as well as its tangent and
cotangent spaces. As usual for infinite dimensional manifolds, there are
some subtleties.

The main goal of this section is to describe the \emph{formal} tangent
and cotangent spaces of the manifold of arbitrary field sections and
the manifold of solution sections. The adjective formal, in the last
sentence, alludes to the fact that we avoid most technical issues of
infinite dimensional analysis and concentrate on what would be dense
subspaces of the true tangent and cotangent spaces with a reasonable choice for
their topologies. Results are algebraic and (finite dimensional)
geometric identities that would form the core of an earnest functional
analytical formulation of their non-formal versions. The formal tangent
and cotangent spaces have a natural dual pairing, which we prove to be
non-degenerate, as a substitute for the absence of true topological
duality between them. In the presence of constraints, the proof is
carried out under some additional sufficient conditions.

We start with Sec.~\ref{sec:linear}, which explains how linearizing the
linearized equations of motion are related to the formal cotangent
space of the space of solutions.
Secs.~\ref{sec:constr} and Sec.~\ref{sec:gauge} discuss
sufficient conditions on the constraints and gauge transformations
needed for later results. Secs.~\ref{sec:tt*-conf}
and~\ref{sec:tt*-sols} define the formal tangent and cotangent spaces in
the progressively more complicated cases of the space of field
configurations, the space of solutions (without constraints), and the
space of solutions (with constraints).

\subsubsection{Non-linear equations and linearization}\label{sec:linear}
In the preceding section (Sec.~\ref{sec:var-sys}) we have discussed
general variational systems, without regard for either linearity or
hyperbolicity. Note that the notion of Green hyperbolicity that we
discussed earlier in Sec.~\ref{sec:lin-inhom} is only applicable to
linear systems. The way that we shall restrict our discussion to linear
systems is by linearization, which is justified below.

Let us denote by $\S(F)$ the set of solutions of the equations of motion
of a given non-linear classical field theory on a field bundle $F\to M$.
For instance, for General Relativity $\S(F)$ would include metrics of
all possible signatures, not just Lorentzian ones. So, obviously, we
shall not be interested in all possible solutions, but those that have
good causal behavior. We shall not delve here into the question of what
constitutes good causal behavior in a non-linear field theory, but refer
the reader to Sec.~4.2 of Ref.~\citen{kh-caus}. We shall simply postulate that
there is a subset $\S_H(F) \sse \S(F)$ that consists of all solutions
with good causal behavior, with the subscript $H$ nominally standing for
\emph{globally hyperbolic}. For us, the most important property of any
background solution $\phi \in \S_H(F)$ is that the \emph{linearized}
equations of motion about $\phi$ are Green hyperbolic
(Sec.~\ref{sec:lin-inhom}). Sometimes, we shall also refer to a
background solution $\phi$ as a \emph{dynamical linearization point}.

We shall also assume the hypothesis that $\S_H(F)$ can be seen as a
possibly infinite dimensional manifold (see
Refs.~\citen{bsf,fr-bv,rejzner-thesis,bfr} for attempts to make that
precise). When dealing with either symplectic or Poisson structure, we
need also the notions of the tangent $T\S_H(F)$ and cotangent
$T^*\S_H(F)$ bundles, since these structures are needed to define 2-form
or bivector tensors on $\S_H(F)$. The de~Rham closedness and Jacobi
identities that respectively identify symplectic and Poisson structures
require a notion of differentiation, that is, a differential structure
on $T\S_H(F)$ and $T^*\S_H(F)$ as well. However, if we concentrate on
the mutual inverse relationship between a symplectic form $\Omega$ and a
Poisson bivector $\Pi$, we are allowed to work with a single tangent
space $T_\phi \S_H(F)$ and a single cotangent space $T_\phi^*\S_H(F)$ at
a time, with $\phi\in \S_H(F)$, verifying this property for each pair
$\Omega_\phi$ and $\Pi_\phi$ individually. This is precisely what we do
below.

From now on, we fix $\phi\in \S_H(F)$ to be a particular dynamical
linearization point. Given any, possibly non-linear, differential
operator, we denote its linearization by the same symbol but with a dot,
e.g.,\ $\dot{f}\colon \Secs(F) \to \Secs(\tilde{F}^*)$ is the
linearization of $f\colon \Secs(F) \to \Secs(\tilde{F}^*)$ about $\phi$.
Since solution space $\S_H(F)$ is embedded in the total field
configuration space $\Secs(F)$, the tangent space at $\phi$ is defined
by the space of linearized solutions, that is, solutions of the
linearized equations.  The following sections, Secs.~\ref{sec:constr}
and~\ref{sec:gauge}, introduce the linear differential operators that we
expect to obtain after linearizing the equations of motion with
constraints and gauge invariance. Later, in Sec.~\ref{sec:constr-gf}, we
consider the Euler-Lagrange equations of a possibly non-linear classical
field theory and linearize them, together with the corresponding
hyperbolic, constraint and gauge generator differential operators.

\subsubsection{Constraints}\label{sec:constr}
Earlier, in Sec.~\ref{sec:compat-constr}, we discussed linear
constrained hyperbolic systems. This notion can actually be extended to
non-linear systems, with a very similar structure of identities
satisfied by the differential operators involved. See
Refs.~\citen{geroch-pde} and~\citen{kh-caus} for details. At this point, w
will presume that we are dealing with a linearization of a possibly
non-linear constrained hyperbolic system, whose linearization is itself
a linear constrained hyperbolic system of the kind described in
Sec.~\ref{sec:compat-constr}. To keep the linearization in mind, we put
a dot on all the differential operators. Thus, we have a linear
constrained hyperbolic system defined by the operators $\dot{f}\colon
\Secs(F) \to \Secs(\tilde{F}^*)$, $\dot{c}\colon \Secs(F) \to \Secs(E)$,
$\dot{h}\colon \Secs(E) \to \Secs(\tilde{E}^*)$, $\dot{q}\colon
\Secs(\tilde{F}^*) \to \Secs(\tilde{E}^*)$, satisfying the identity
$\dot{h}\circ \dot{c} = \dot{q}\circ \dot{f}$. An important thing to
note is that their formal adjoints will satisfy the related identity
$\dot{c}^* \circ \dot{h}^* = \dot{f}^* \circ \dot{q}^*$, which is
exploited below.

When dealing with constrained systems, some results covered later will
require the further sufficient condition that the constraints be
\emph{parametrizable} (see Sec.~\ref{sec:caus-green}) so that we can
extend both the linearized system and its adjoint to the
following commutative diagrams:
\begin{equation}\label{eq:glpar}
\vcenter{\xymatrix{
	0
		\ar[r] &
	\Secs_0(E')
		\ar[d]^{\dot{c}'} \ar[r]^{\dot{h}'} &
	\Secs_0(\tilde{E}^{\prime*})
		\ar[d]^{\dot{q}'} \ar[r]^{\H'} &
	\Secs_{SC}(E')
		\ar[d]^{\dot{c}'} \ar[r]^{\dot{h}'} &
	\Secs_{SC}(\tilde{E}^{\prime*})
		\ar[d]^{\dot{q}'} \ar[r] &
	0  \\
	0
		\ar[r] &
	\Secs_0(F)
		\ar[d]^{\dot{c}} \ar[r]^{\dot{f}} &
	\Secs_0(\tilde{F}^*)
		\ar[d]^{\dot{q}} \ar[r]^{\G} &
	\Secs_{SC}(F)
		\ar[d]^{\dot{c}} \ar[r]^{\dot{f}} &
	\Secs_{SC}(\tilde{F}^*)
		\ar[d]^{\dot{q}} \ar[r] &
	0  \\
	0
		\ar[r] &
	\Secs_0(E)
		\ar[r]^{\dot{h}} &
	\Secs_0(\tilde{E}^*)
		\ar[r]^{\H} &
	\Secs_{SC}(E)
		\ar[r]^{\dot{h}} &
	\Secs_{SC}(\tilde{E}^*)
		\ar[r] &
	0
}}
\end{equation}
and
\begin{equation}\label{eq:glpar*}
\vcenter{\xymatrix{
	0
		\ar@{<-}[r] &
	\Secs_{SC}(\tilde{E}^{\prime*})
		\ar@{<-}[d]^{\dot{c}^{\prime*}} \ar@{<-}[r]^{\dot{h}^{\prime*}} &
	\Secs_{SC}(E')
		\ar@{<-}[d]^{\dot{q}^{\prime*}} \ar@{<-}[r]^{\H^{\prime*}} &
	\Secs_{0}(\tilde{E}^{\prime*})
		\ar@{<-}[d]^{\dot{c}^{\prime*}} \ar@{<-}[r]^{\dot{h}^{\prime*}} &
	\Secs_{0}(E')
		\ar@{<-}[d]^{\dot{q}^{\prime*}} \ar@{<-}[r] &
	0  \\
	0
		\ar@{<-}[r] &
	\Secs_{SC}(\tilde{F}^*)
		\ar@{<-}[d]^{\dot{c}^*} \ar@{<-}[r]^{\dot{f}^*} &
	\Secs_{SC}(F)
		\ar@{<-}[d]^{\dot{q}^*} \ar@{<-}[r]^{\G^*} &
	\Secs_{0}(\tilde{F}^*)
		\ar@{<-}[d]^{\dot{c}^*} \ar@{<-}[r]^{\dot{f}^*} &
	\Secs_{0}(F)
		\ar@{<-}[d]^{\dot{q}^*} \ar@{<-}[r] &
	0  \\
	0
		\ar@{<-}[r] &
	\Secs_{SC}(\tilde{E}^*)
		\ar@{<-}[r]^{\dot{h}^*} &
	\Secs_{SC}(E)
		\ar@{<-}[r]^{\H^*} &
	\Secs_{0}(\tilde{E}^*)
		\ar@{<-}[r]^{\dot{h}^*} &
	\Secs_{0}(E)
		\ar@{<-}[r] &
	0
}}
\end{equation}
The rows form exact sequences, while the columns form formally exact
complexes, as described in Sec.~\ref{sec:caus-green}. Note that the
adjoint diagram also describes a hyperbolic system with
hyperbolically integrable constraints, except that the role of the
constraint subsystem is now played by $(\dot{q}^{\prime*},E')$ and the
consistency subsystem is $(\dot{h}^{\prime*},\tilde{E}^{\prime*})$,
which satisfies the consistency identity $\dot{h}^{\prime*}\circ
\dot{q}^{\prime*} = \dot{c}^{\prime*}\circ \dot{f}^*$.

It is convenient to introduce here a cohomological condition, to be
applied in later sections, on the columns of the above commutative
diagrams. First, let us introduce some notation for the respective
cohomologies. Because the columns are so short, the cohomologies can
only be defined at the middle nodes. Each cohomology can be identified
by the vector bundle where it is defined, $F$ or $\tilde{F}^*$, the
support restriction, $0$ or $SC$, and the diagram used to define it,
\eqref{eq:glpar} or~\eqref{eq:glpar*}. Thus, we denote the cohomologies
defined by the columns of diagram~\eqref{eq:glpar} by $H^c_0(F)$,
$H^c_0(\tilde{F}^*)$, $H^c_{SC}(F)$ and $H^c_{SC}(\tilde{F}^*)$, while
those defined by the columns of diagram~\eqref{eq:glpar*} by
$H^{c^*}_0(F)$, $H^{c^*}_0(\tilde{F}^*)$, $H^{c^*}_{SC}(F)$ and
$H^{c^*}_{SC}(\tilde{F}^*)$. A cocycle section $\psi\in \Secs(F)$ (which
is annihilated by $\dot{c}$ or $\dot{q}^{\prime*}$), with appropriately
restricted support, represents a cohomology class denoted by $[\psi]_c$
or $[\psi]_{c^*}$, and similarly for cocycle sections in
$\Secs(\tilde{F}^*)$ (which are annihilated by $\dot{q}$ or
$\dot{c}^{\prime*}$).

Second, recall that, given $\psi\in \Secs(F)$ and
$\tilde{\alpha}^* \in \Secs(\tilde{F}^*)$, there is a natural paring
$\langle \psi, \tilde{\alpha}^* \rangle = \int_M \psi \cdot
\tilde{\alpha}^*$, provided the integral is finite. This pairing is
indeed well defined between the corresponding nodes of the
diagrams~\eqref{eq:glpar} and~\eqref{eq:glpar*}. Moreover, when
restricted to cocycle sections, this pairing descends to cohomology
classes, say $\langle [\psi]_c, [\tilde{\alpha}^*]_{c^*} \rangle =
\langle \psi, \tilde{\alpha}^* \rangle$. Thus, we have well defined
natural pairings on $H^c_0(F) \times H^{c^*}_{SC}(\tilde{F}^*)$,
$H^c_0(\tilde{F}^*) \times H^{c^*}_{SC}(F)$, $H^c_{SC}(F) \times
H^{c^*}_0(\tilde{F}^*)$, $H^c_{SC}(\tilde{F}^*) \times H^{c^*}_0(F)$.

Third, we must recall that the commutativity of the
diagrams~\eqref{eq:glrec} and~\eqref{eq:glrec*} allows us to consider
the respective cohomologies at $\Secs_0(\tilde{F}^*)$ modulo $\im
\dot{f}$ and at $\Secs_{SC}(F)$ restricted to $\im \G^*$. In both cases,
by the exactness of the rows of these diagrams, we are simply describing
the vertical cohomologies in the space of solutions of $\dot{f}$ and
$\dot{f}^*$, respectively. We shall denote them by $H^c_{SC}(F,\dot{f})$
and $H^{c^*}_{SC}(F,\dot{f}^*)$, respectively. The natural pairing also
descends to the cohomologies in solutions as follows, with say $\psi =
\G[\tilde{\alpha}^*]$ and $\xi = \G^*[\tilde{\beta}^*]$,
\begin{equation}
	\langle [\psi]_c , [\xi]_{c^*} \rangle_G
	= \langle \psi , \xi \rangle_G ,
\end{equation}
where on the right-hand-side $\langle - , - \rangle_G$ is the Green
pairing from Def.~\ref{def:green-pairing}.

\begin{definition} \label{def:glpar}
The constraints are said to be \emph{globally parametrizable} if the
natural pairing between the vertical cohomologies in solutions,
$H^c_{SC}(F,\dot{f})$ and $H^{c^*}_{SC}(F,\dot{f}^*)$, defined using the
commutative diagrams~\eqref{eq:glpar} and~\eqref{eq:glpar*}, is
non-degenerate.
\end{definition}

\begin{remark}
Note that, as also mentioned in Sec.~\ref{sec:caus-green}, when dealing
with parametrizable constraints, the causal structure that is in use is
that of the total compound system, whose equation form is $(\dot{h}'
\oplus \dot{f} \oplus \dot{h}, \tilde{E}^{\prime*} \oplus_M \tilde{F}^*
\oplus_M \tilde{E}^*)$. It is easy to
show that the adjoint system $(\dot{h}^*\oplus \dot{f} \oplus
\dot{h}^{\prime*}, \tilde{E}^* \oplus_M \tilde{F}^* \oplus_M
\tilde{E}^{\prime*})$ defines the same causal
structure.
\end{remark}

\subsubsection{Gauge transformations}\label{sec:gauge}
% XXX: The differential operator requirement on \dot{g} may be
% superfluous, since it preserves supports and that may already imply
% that it is a differential operator. cf. "Peetre theorem"
Many important classical field theories exhibit gauge invariance, like
Maxwell theory, Yang-Mills theory, and General Relativity. A \emph{gauge
transformation} is a family of maps $g_\eps\colon \Secs(F) \to
\Secs(F)$, parametrized by sections $\eps\in\Secs(P)$ of the \emph{gauge
parameter bundle} $P\to M$, that take solutions to solutions, while not
modifying a field section outside the support of $\delta \in \Secs(P)$,
$g_\delta[\phi](x) = \phi(x)$ if $x\not\in \supp \delta$, which may be
compact and arbitrarily small. If we linearize about some pair of
background section $\delta \to \delta + \eps$, we obtain a linearized
gauge transformation $g_\delta[\phi] \to g_\delta[\phi] +
\dot{g}[\eps]$. It is another requirement on gauge transformations that
the \emph{generator of linearized gauge transformations} $\dot{g}\colon
\Secs(P) \to \Secs(F)$ is a differential operator, which may depend on
the background sections $\delta$ and $\phi$.

Equivalence classes of sections under gauge transformations are
considered physically equivalent. Therefore, physical observables will
consist only of those functions on phase space that are gauge invariant
(constant on orbits of gauge transformations). Equivalently, observables
are annihilated by the action of linearized gauge transformations.
Another way to look at it, is to consider observables as functions on
the space of gauge orbits. We denote the space of solutions of the
possibly non-linear equations of motion (with good causal behavior,
cf.~Sec.~\ref{sec:linear}) modulo gauge transformations, $\bar{S}_H(F) =
\S_H(F)/{\sim}$ and call it the \emph{physical phase space}.

Often it is convenient to impose subsidiary conditions on field
sections, called \emph{gauge fixing}, that restrict the choice of
representatives of gauge equivalence classes. The gauge fixing is called
\emph{full} if they only allow a unique representative from each
equivalence class, and otherwise called \emph{partial}. The gauge
transformations that are compatible with a partial gauge fixing are
called \emph{residual}.

Unfortunately, PDE systems with gauge invariance cannot have a
well-posed initial value problem, and hence cannot be hyperbolic.  In
particular, their linearizations cannot be Green hyperbolic.  However,
the addition of subsidiary conditions on field sections can make the new
PDE system equivalent to a hyperbolic one, usually with constraints.  In
practice, many hyperbolic systems with constraints arise after adding
such gauge fixing conditions to a non-hyperbolic system with gauge
invariance. Interestingly, after many convenient gauge fixings, there
may remain non-trivial residual gauge freedom. For later convenience, as
we did with constraints, we restrict our attention to what we call
\emph{recognizable gauge transformations}.  That is, given linearized
gauge transformations of the form $\dot{g}[\eps]$ and a partially gauge
fixed hyperbolic system with equation form $(\dot{f},\tilde{F}^*)$, we
can fit them into the following commutative diagram, whose columns form
formally exact complexes:
\begin{equation}\label{eq:glrec}
\vcenter{\xymatrix{
	0
		\ar[r] &
	\Secs_0(P)
		\ar[d]^{\dot{g}} \ar[r]^{\dot{k}} &
	\Secs_0(\tilde{P}^*)
		\ar[d]^{\dot{s}} \ar[r]^{\K} &
	\Secs_{SC}(P)
		\ar[d]^{\dot{g}} \ar[r]^{\dot{k}} &
	\Secs_{SC}(\tilde{P}^*)
		\ar[d]^{\dot{s}} \ar[r] &
	0 \\
	0
		\ar[r] &
	\Secs_0(F)
		\ar[d]^{\dot{g}'} \ar[r]^{\dot{f}} &
	\Secs_0(\tilde{F}^*)
		\ar[d]^{\dot{s}'} \ar[r]^{\G} &
	\Secs_{SC}(F)
		\ar[d]^{\dot{g}'} \ar[r]^{\dot{f}} &
	\Secs_{SC}(\tilde{F}^*)
		\ar[d]^{\dot{s}'} \ar[r] &
	0  \\
	0
		\ar[r] &
	\Secs_0(P')
		\ar[r]^{\dot{k}'} &
	\Secs_0(\tilde{P}^{\prime*})
		\ar[r]^{\K'} &
	\Secs_{SC}(P')
		\ar[r]^{\dot{k}'} &
	\Secs_{SC}(\tilde{P}^{\prime*})
		\ar[r] &
	0
}}
\end{equation}
Their adjoints fit into the adjoint diagram whose columns are also
formally exact complexes:
\begin{equation}\label{eq:glrec*}
\vcenter{\xymatrix{
	0
		\ar@{<-}[r] &
	\Secs_{SC}(\tilde{P}^*)
		\ar@{<-}[d]^{\dot{g}^*} \ar@{<-}[r]^{\dot{k}^*} &
	\Secs_{SC}(P)
		\ar@{<-}[d]^{\dot{s}^*} \ar@{<-}[r]^{\K^*} &
	\Secs_{0}(\tilde{P}^*)
		\ar@{<-}[d]^{\dot{g}^*} \ar@{<-}[r]^{\dot{k}^*} &
	\Secs_{0}(P)
		\ar@{<-}[d]^{\dot{s}^*} \ar@{<-}[r] &
	0 \\
	0
		\ar@{<-}[r] &
	\Secs_{SC}(\tilde{F}^*)
		\ar@{<-}[d]^{\dot{g}^{\prime*}} \ar@{<-}[r]^{\dot{f}^*} &
	\Secs_{SC}(F)
		\ar@{<-}[d]^{\dot{s}^{\prime*}} \ar@{<-}[r]^{\G^*} &
	\Secs_{0}(\tilde{F}^*)
		\ar@{<-}[d]^{\dot{g}^{\prime*}} \ar@{<-}[r]^{\dot{f}^*} &
	\Secs_{0}(F)
		\ar@{<-}[d]^{\dot{s}^{\prime*}} \ar@{<-}[r] &
	0  \\
	0
		\ar@{<-}[r] &
	\Secs_{SC}(\tilde{P}^{\prime*})
		\ar@{<-}[r]^{\dot{k}^{\prime*}} &
	\Secs_{SC}(P')
		\ar@{<-}[r]^{\K^{\prime*}} &
	\Secs_{0}(\tilde{P}^{\prime*})
		\ar@{<-}[r]^{\dot{k}^{\prime*}} &
	\Secs_{0}(P')
		\ar@{<-}[r] &
	0
}}
\end{equation}
The systems with equation forms $(\dot{k},\tilde{P}^*)$ and
$(\dot{k}^{\prime*}, \tilde{P}^{\prime*})$ are required to be hyperbolic
and $P'\to M$ is called the \emph{gauge invariant field bundle}, while
$\dot{g}'$ is called the \emph{operator of gauge invariant field
combinations}.

The above commutative diagrams are very similar to those used to define
parametrizable constraints in Sec.~\ref{sec:constr}. Thus, we can define
all the same cohomologies: $H^g_0(F)$, $H^g_0(\tilde{F}^*)$,
$H^g_{SC}(F)$, $H^g_{SC}(\tilde{F}^*)$, $H^g_{SC}(F,\dot{f})$ defined by
diagram~\eqref{eq:glrec}, and $H^{g^*}_{SC}(\tilde{F}^*)$,
$H^{g^*}_{SC}(F)$, $H^{g^*}_0(\tilde{F}^*)$, $H^{g^*}_0(F)$,
$H^{g^*}_{SC}(F,\dot{f}^*)$ defined by diagram~\eqref{eq:glrec*}. Also,
in exactly the same way, there are bilinear pairings $\langle - , -
\rangle$ and $\langle - , - \rangle_G$ defined on respective pairs of
these cohomologies. On the other hand, the following definition is not
quite analogous, reflecting how this hypothesis is used in later
sections.

\begin{definition} \label{def:glrec}
The gauge transformations are said to be \emph{globally recognizable} if
the natural pairing between the cohomologies $H^g_{SC}(F)$ and
$H^{g^*}_0(\tilde{F}^*)$, defined using the commutative
diagrams~\eqref{eq:glrec} and~\eqref{eq:glrec*}, is non-degenerate.
\end{definition}

\subsubsection{Formal $T$ and $T^*$ for configurations}
\label{sec:tt*-conf}
Here we consider a section $\phi\in \S_H(F) \sso \Secs(F)$ and examine
the formal tangent and cotangent spaces at $T_\phi\Secs =
T_\phi\Secs(F)$ and $T^*_\phi\Secs = T^*_\phi \Secs(F)$ at $\phi$.
\begin{definition}
We define the \emph{formal full tangent space at $\phi$} as the set of
spacelike compactly supported sections and we define the \emph{formal full
cotangent space at $\phi$} as the set
\begin{equation}
	T_\phi\Secs \cong \Secs_{SC}(F)
	\quad\text{and}\quad
	T^*_\phi\Secs \cong \Secs_0(\tilde{F}^*) .
\end{equation}
The natural pairing $\langle-,-\rangle \colon T_\phi\Secs \times
T^*_\phi\Secs\to \R$ is
\begin{equation}
	\langle \psi, \tilde{\alpha}^* \rangle
	= \int_M \psi\cdot \tilde{\alpha}^* .
\end{equation}
\end{definition}
\begin{lemma}
The natural pairing between $T_\phi\Secs$ and $T^*_\phi\Secs$ is
non-degenerate.
\end{lemma}
This is essentially the fundamental lemma of the calculus of variations
and the proof is standard~\cite[\textsection IV.3.1]{ch1}.

Since the physical phase space will be identified with the space of
gauge orbits $\bar{\S}_H(F)$ in the solution space $\S_H(F)$, given a
solution section $\phi\in\S_H(F)$, the formal tangent space
$T_\phi\bar{\S} = T^*_\phi\bar{\S}_H(F)$ at the corresponding
equivalence class $[\phi]\in\bar{\S}_H(F)$ in the space of gauge orbits
consists of equivalence classes of linearized solutions up to linearized
gauge transformations. Dually, the formal cotangent space
$T^*_\phi\bar{S}=T^*_\phi\bar{\S}_H(F)$ will consist of dual densities
annihilated by the adjoint of infinitesimal gauge transformation
generator.

Since gauge transformations act on field configurations and not just
solutions, it makes sense to consider all field configurations related
by gauge transformations as physically equivalent. The tangent space
$T_\phi\Secs$ will be reduced to the quotient (or physical) tangent
space $T_\phi\bar{\Secs}$ and the cotangent space $T^*\Secs$ to the
subset $T^*\bar{\Secs}$ of gauge invariant elements. The natural
pairing between them is shown to be non-degenerate under the condition
of global recognizability, that is, the vertical formally exact complexes in
diagrams~\eqref{eq:glrec} and~\eqref{eq:glrec*} are exact.  We deal with
field configurations first and delay the discussion of solutions to the
next section.

The exactness of the composition $\dot{g}'\circ\dot{g}=0$ ensures that
we can recognize pure gauge field configurations, which are of the form
$\psi=\dot{g}[\eps]$ for some spacelike compactly supported section
$\eps\colon M\to P$, precisely as those spacelike compact field
sections $\psi\colon M\to F$ that give vanishing gauge invariant field
combinations $\dot{g}'[\psi] = 0$.  On the other hand, the exactness of
the dual composition $\dot{g}^*\circ \dot{g}^{\prime*} = 0$ ensures that
we can parametrize gauge invariant, compactly supported dual densities
$\tilde{\alpha}^*\colon M\to \tilde{F}^*$, those satisfying
$\dot{g}^*[\alpha] = 0$, precisely as the image of the differential
operator $\dot{g}^{\prime*}$ acting on compactly supported sections of
$\tilde{P}^{\prime*}\to M$.

\begin{definition}
The \emph{formal gauge invariant full tangent space} at $\phi$ is the
set of gauge equivalence classes of $\phi$-spacelike compactly supported
sections,
\begin{align}
	T_\phi\bar{\Secs}
		&= \{ [\psi] \mid \psi\in \Secs_{SC}(F) \} , \\
	[\psi] &\sim \psi + \dot{g}[\eps], ~
			\text{with}~ \eps\in \Secs_{SC}(P) .
\end{align}
The \emph{formal gauge invariant full cotangent space} at $\phi$ is the
set of compactly supported gauge invariant dual densities,
\begin{equation}
	T_\phi^*\bar{\Secs}
		= \{ \tilde{\alpha}^* \in \Secs_0(\tilde{F}^*) \mid
			\dot{g}^*[\tilde{\alpha}^*]=0 \} .
\end{equation}
The natural pairing $\langle-,-\rangle\colon T_\phi\bar{\Secs} \times
T^*_\phi\bar{\Secs} \to \R$ is 
\begin{equation}
	\langle [\psi], \tilde{\alpha}^* \rangle
	= \int_M \psi\cdot\tilde{\alpha}^* .
\end{equation}
\end{definition}
\begin{lemma}\label{lem:ginv-full-nondegen}
If the gauge transformations are globally recognizable
(Def.~\ref{def:glrec}), the natural pairing between the gauge invariant
spaces $T_\phi\bar{\Secs}$ and $T^*_\phi\bar{\Secs}$ is non-degenerate.
\end{lemma}
\begin{proof}
Non-degeneracy in the second argument follows once again from the
fundamental lemma of the calculus of variations: $\langle [\psi],
\tilde{\alpha}^* \rangle = \langle \psi, \tilde{\alpha}^* \rangle = 0$
for all $\psi\in T_\phi\Secs$ implies that $\tilde{\alpha}^* = 0$.

Non-degeneracy in the first argument is more complicated, since we can
no longer use arbitrary $\tilde{\alpha}^*$ in the second argument. It
now requires an appeal to the global recognizability of the gauge
transformations. Suppose that $\langle [\psi], \tilde{\alpha}^* \rangle
= 0$ for all $\tilde{\alpha}^*\in T^*_\phi\bar{\Secs}$.  We need to show
that this implies $\psi = \dot{g}[\eps]$ is pure gauge, for some
spacelike compactly supported $\eps\in \Secs_{SC}(P)$.

By definition, a gauge invariant dual density represents a cohomology
class $[\tilde{\alpha}^*]_{g^*} \in H^{g^*}_0(\tilde{F}^*)$, defined in
Sec.~\ref{sec:gauge}. In fact, any such class could be represented.
Considering $\tilde{\alpha}^* =
\dot{g}^{\prime*}[\tilde{\eps}^{\prime*}]$ with arbitrary
$\tilde{\eps}^{\prime*}\in \Secs_0(\tilde{P}^{\prime*})$, we find
\begin{equation}
	\langle [\psi] , \tilde{\alpha}^* \rangle
	= \langle \psi, \tilde{\alpha}^* \rangle
	= \langle \psi, \dot{g}^{\prime*}[\tilde{\eps}^{\prime*}] \rangle
	= \langle \dot{g}'[\psi], \tilde{\eps}^{\prime*} \rangle .
\end{equation}
Since $\tilde{\eps}^{\prime*}$ could be arbitrary, the vanishing of
$\langle \dot{g}'[\psi], \tilde{\eps}^{\prime*} \rangle$ implies that
$\dot{g}'[\psi] = 0$. That is, $\psi$ necessarily represents a
cohomology class $[\psi]_g\in H^g_{SC}(F)$, also defined in
Sec.~\ref{sec:gauge}. Therefore we find that non-degeneracy in the
first argument implies that
\begin{equation}
	\langle \psi, \tilde{\alpha}^* \rangle
	= \langle [\psi]_g , [\tilde{\alpha}^*]_{g^*} \rangle = 0 ,
\end{equation}
where the last pairing is in the respective cohomologies and the class
$[\tilde{\alpha}^*]_{g^*}$ allowed to be arbitrary. But the global
recognizability hypothesis specifies precisely that the above pairing in
cohomology is non-degenerate and implies that $[\psi]_g = [0]$ and hence
that $\psi = \dot{g}[\eps]$, with $\eps \in \Secs_{SC}(P)$, is pure
gauge. %\qed
% XXX: In fact we only need left non-degeneracy for
%   H^g_SC(F) x H^{g^*}_0(F^*)
% moreover, any representative of H^g_SC(F) is allowed to be an element
% of T\Secs
\end{proof}

\subsubsection{Formal $T$ and $T^*$ for solutions}%
\label{sec:tt*-sols}
The formal tangent space $T_\phi\S$ will consist of linearized
solutions, that is solutions of the linearized constrained hyperbolic
system $\dot{f}[\psi] = 0$ and $\dot{c}[\psi] = 0$. The formal cotangent
space will naturally consist of equivalence classes of dual densities up
to the images of the adjoints of $\dot{f}$ and $\dot{c}$. After giving
the precise definitions below, we prove that that the natural pairing
between these formal tangent and cotangent spaces is non-degenerate.

\begin{definition}\label{def:tt*-sols}
We define the \emph{formal solutions tangent space} at $\phi$ as the set of
spacelike compactly supported linearized solution sections,
\begin{equation}
	T_\phi\S = T_\phi\S_H(F) = \{ \psi \in \Secs_{SC}(F) \mid 
		\dot{f}[\psi] = \dot{c}[\psi] = 0 \} .
\end{equation}
We define the \emph{formal solutions cotangent space} at $\phi$ as the set
of equivalence classes of compactly supported dual densities,
\begin{align}
	T^*_\phi\S = T^*_\phi\S_H(F) &= \{ [\tilde{\alpha}^*] \mid
		\tilde{\alpha}^* \in \Secs_0(\tilde{F}^*) \} , \\
	[\tilde{\alpha}^*] &\sim \tilde{\alpha}^*
		+ \dot{f}^*[\xi] + \dot{c}^*[\tilde{\eps}^*],  ~~\text{with}~~
		\xi \in \Secs_0(F), ~ \tilde{\eps}^*\in \Secs_0(\tilde{E}^*) .
\end{align}
The natural pairing $\langle-,-\rangle \colon T_\phi\S \times
T^*_\phi\S\to \R$ is
\begin{equation}
	\langle \psi, [\tilde{\alpha}^*] \rangle
		= \int_M \psi\cdot \tilde{\alpha}^* .
\end{equation}
\end{definition}

As a warm-up before the main result of this section, we fist handle the
case where the constraints and gauge transformations are trivial.

\begin{lemma}\label{lem:sols-nondegen}
If the constraints $\dot{c}[\phi] = 0$ and the gauge transformations are
trivial, then the natural pairing between $T_\phi\S$ and $T^*_\phi\S$ is
non-degenerate.
\end{lemma}
\begin{proof}
Non-degeneracy in the first argument follows again from the fundamental
lemma of the calculus of variations: $\langle \psi, [\tilde{\alpha}^*]
\rangle = 0$ for all $\tilde{\alpha}^*\in \Secs_0(\tilde{F}^*)$ implies
that $\psi = 0$.

Non-degeneracy in the second argument is more tricky, since now $\psi$
can no longer be arbitrary. Suppose we have $\langle \psi,
[\tilde{\alpha}^*] \rangle = 0$ for all spacelike compactly supported
linearized solutions $\psi\in T_\phi\S$. From this, we need to deduce
that $[\tilde{\alpha}^*] = [0]$, which means $\tilde{\alpha}^* =
\dot{f}^*[\xi]$ for some compactly supported $\xi\in \Secs_0(F)$.

By Prp.~\ref{prp:exact}, we can parametrize all solutions as $\psi =
\G[\tilde{\beta}^*]$, using unrestricted $\tilde{\beta}^* \in
\Secs_0(\tilde{F}^*)$. The following simple calculation
\begin{equation}
	\langle \G[\tilde{\beta}^*], [\tilde{\alpha}^*] \rangle
	= \langle \G[\tilde{\beta}^*], \tilde{\alpha}^* \rangle
	= -\langle \tilde{\beta}^*, \G^*[\tilde{\alpha}^*] \rangle
\end{equation}
and an application of the fundamental lemma of the calculus of
variations shows that $\G^*[\tilde{\alpha}^*] = 0$. But, once again
appealing to Prp.~\ref{prp:exact}, this implies that $\tilde{\alpha}^*
= \dot{f}^*[\xi]$ with $\xi \in \Secs_0(F)$, which concludes the proof.
%\qed
\end{proof}

\begin{remark}
At this point, it is worth mentioning that the natural pairing between
$T_\phi\S$ and $T_\phi^*\S$ (again, in the absence of constraints or
gauge transformations) is essentially equivalent, via
Lem.~\ref{lem:green-nondegen}, to the Green pairing
(Def.~\ref{def:green-pairing}), which is also non-degenerate.
\end{remark}

In the presence of gauge symmetries, the formal tangent space consists
of equivalence classes of linearized solutions up to gauge
transformations. On the other hand, the formal cotangent space is
restricted to equivalence represented by gauge invariant dual densities.
After giving the precise definitions below, we prove that the natural
pairing between these formal tangent and cotangent spaces is
non-degenerate, provided the constraints $\dot{c}[\phi] = 0$ are
globally parametrizable and the gauge transformation are globally
recognizable.

The following technical definition is motivated by following steps: we
first construct the solution space $T_\phi\S$ and then quotient by its
purge gauge subspace.
\begin{definition}\label{def:tt*-sols-gauge}
We define the \emph{formal gauge invariant solutions tangent space} at
$\phi$ as the set of gauge equivalence classes of $\phi$-spacelike
compactly supported linearized solution sections,
\begin{align}
\notag
	T_\phi\bar{\S} = T_\phi\bar{\S}_H(F)
		&= \{ [\psi] \mid \psi\in \Secs_{SC}(F), \dot{f}[\psi] = 0,
			\dot{c}[\psi] = 0 \} , \\
\notag
	[\psi] &\sim \psi + \dot{g}[\eps], \\
\notag
	&{} \quad \text{with}~ 
		\eps\in \Secs_{SC}(P) \\
	&{} \quad \text{and}~
		\dot{f}[\dot{g}[\eps]] = 0, \dot{c}[\dot{g}[\eps]] = 0 .
\end{align}
The \emph{formal gauge invariant solutions cotangent space} at $\phi$ is
the set of equivalence classes of compactly supported gauge invariant
dual densities,
\begin{align}
\notag
	T^*_\phi\bar{\S} = T^*_\phi\bar{\S}_H(F)
		&= \{ [\tilde{\alpha}^*] \mid \tilde{\alpha}^* \in \Secs_0(\tilde{F}^*) ,
			\dot{g}^*[\tilde{\alpha}^*]
				= \dot{g}^*[\dot{f}^*[\xi] + \dot{c}^*[\tilde{\eps}^*]], \\
\notag
		&{} \qquad \text{with}~
			\xi \in \Secs_0(F), \tilde{\eps}^* \in \Secs_0(\tilde{E}^*) \} , \\
\notag
	[\tilde{\alpha}^*] &\sim \tilde{\alpha}^*
		+ \dot{f}^*[\xi] + \dot{c}^*[\tilde{\eps}^*], \\
		&{} \quad \text{with}~
			\xi\in \Secs_0(F), \tilde{\eps}^* \in \Secs_0(\tilde{E}^*) .
\end{align}
The natural pairing $\langle-,-\rangle\colon T_\phi\bar{\S}\times
T^*_\phi\bar{\S} \to \R$ is
\begin{equation}
	\langle [\psi], [\tilde{\alpha}^*] \rangle
		= \int_M \psi\cdot\tilde{\alpha}^* .
\end{equation}
\end{definition}

We now prove the main result of this section.

\begin{lemma}\label{lem:ginv-sols-nondegen}
If the constraints $\dot{c}[\phi] = 0$ are globally parametrizable
(Def.~\ref{def:glpar}) and the gauge transformations are globally
recognizable (Def.~\ref{def:glrec}), then the natural pairing between
$T_\phi\bar\S$ and $T^*_\phi\bar\S$ is non-degenerate.
\end{lemma}
\begin{proof}
Unfortunately, we now cannot directly rely on the fundamental lemma of the
calculus of variations to prove non-degeneracy in either argument.
Instead, we proceed roughly as in the proof of
Lem.~\ref{lem:ginv-full-nondegen}.

To prove non-degeneracy in the first argument, suppose we have $\langle
[\psi], [\tilde{\alpha}^*] \rangle = 0$ for arbitrary
$[\tilde{\alpha}^*] \in T_\phi^*\bar{\S}$. It is easy to see from the
definition that we can restrict ourselves to representatives that
satisfy $\dot{g}^*[\tilde{\alpha}^*] = 0$. Then still,
$\tilde{\alpha}^*$ may represent an arbitrary cohomology class
$[\tilde{\alpha}^*]_{g^*} \in H^{g^*}_0(\tilde{F}^*)$, defined in
Sec.~\ref{sec:gauge}. We now need to show that $[\psi] = [0]$, or
equivalently that $\psi = \dot{g}[\eps]$ for some $\eps\in
\Secs_{SC}(P)$. Considering $\tilde{\alpha}^* =
\dot{g}^{\prime*}[\tilde{\eps}^{\prime*}]$ with arbitrary
$\tilde{\eps}^{\prime*} \in \Secs_0(\tilde{P}^*)$, we have
\begin{equation}
	\langle [\psi] , [\tilde{\alpha}^*] \rangle
	= \langle \psi , \dot{g}^{\prime*}[\tilde{\eps}^{\prime*}] \rangle
	= \langle \dot{g}'[\psi], \tilde{\eps}^{\prime*} \rangle .
\end{equation}
Since $\tilde{\eps}^{\prime*}$ could be arbitrary, the vanishing of
$\langle \dot{g}'[\psi], \tilde{\eps}^{\prime*} \rangle$ implies that
$\dot{g}'[\psi] = 0$. That is, $\psi$ necessarily represents a
cohomology class $[\psi]_g \in H^g_{SC}(F)$, also defined in
Sec.~\ref{sec:gauge}. Therefore, for any $[\tilde{\alpha}^*] \in
T_\phi^*\bar{\S}$, we find
\begin{equation}
	0 = \langle \psi, \tilde{\alpha}^* \rangle
	= \langle [\psi]_g , [\tilde{\alpha}^*]_{g^*} \rangle ,
\end{equation}
where the last pairing is in the respective cohomologies and the class
$[\tilde{\alpha}^*]_{g^*}$ is allowed to be arbitrary. But the global
recognizability hypothesis specifies precisely that the above pairing in
cohomology is non-degenerate and implies that $[\psi]_g = [0]$ and hence
that $\psi = \dot{g}[\eps]$, with $\eps \in \Secs_{SC}(P)$, is pure
gauge.
% XXX: Could [psi]_g be arbitrary in H^g_SC(F)? Definitely not, because
% psi is restricted to constrained solutions. So the left degeneracy in
%   H^g_SC(F) x H^{g*}_0(F^*)  or even  H^g_SC(F,f) x H^{g^*}_0(F^*)
% could be greater.

To prove non-degeneracy in the second argument, suppose we have $\langle
[\psi] , [\tilde{\alpha}^*] \rangle = 0$ for arbitrary $[\psi] \in
T_\phi\bar{S}$, which is represented by a solution of $\dot{f}[\psi] =
0$, $\dot{c}[\psi] = 0$. Then, from the definition, it is clear that a
solution $\psi$ may also represent an arbitrary cohomology class
$[\psi]_g \in H^g_{SC}(F,\dot{f})$, defined in Sec.~\ref{sec:constr}. We
now need to show that $[\tilde{\alpha}^*] = [0]$, or equivalently that
$\tilde{\alpha}^* = \dot{f}^*[\xi] + \dot{c}^*[\tilde{\eps}^*]$ for some
$\xi\in \Secs_0(F)$ and $\dot{\eps}^* \in \Secs_0(\tilde{E}^*)$. We can
always choose $\psi = \dot{c}'[\zeta']$ with $\zeta' \in \Secs_{SC}(E')$
such that $\dot{h'}[\zeta'] = 0$, or equivalently $\zeta' =
\H'[\tilde{\eps}^{\prime*}]$ and $\psi = \dot{c}'\circ
\H'[\tilde{\eps}^{\prime*}] = \G\circ\dot{q}'[\tilde{\eps}^{\prime*}]$,
with $\tilde{\eps}^{\prime*} \in \Secs_0(\tilde{E}^{\prime*})$
arbitrary. We then have
\begin{equation}
	\langle [\psi] , [\tilde{\alpha}^*] \rangle
	= \langle \G\circ\dot{q}'[\tilde{\eps}^{\prime*}] ,
			\tilde{\alpha}^* \rangle
	= -\langle \tilde{\eps}^{\prime*} ,
			\dot{q}^{\prime*}[\G^*[\tilde{\alpha}^*]] \rangle .
\end{equation}
Since $\tilde{\eps}^{\prime*}$ could be arbitrary, the vanishing of
$\langle \tilde{\eps}^{\prime*} ,
\dot{q}^{\prime*}[\G^*[\tilde{\alpha}^*]] \rangle$ implies that
$\dot{q}^{\prime*}[\G^*[\tilde{\alpha}^*]]$. That is,
$\G^*[\tilde{\alpha}^*]$ represents a cohomology class
$[\G^*[\tilde{\alpha}^*]]_{c^*} \in H^{c^*}_0(\tilde{F}^*, \dot{f}^*)$, also
defined in Sec.~\ref{sec:constr}. Therefore, for any $[\psi] \in
T_\phi\bar{\S}$, we find
\begin{equation}
	0 = \langle [\psi] , [\tilde{\alpha}^*] \rangle
	= \langle \psi , \tilde{\alpha}^* \rangle
	= -\langle \psi , \G^*[\tilde{\alpha}^*] \rangle_G
	= -\langle [\psi]_c , [\G^*[\tilde{\alpha}^*]]_{c^*} \rangle_G ,
\end{equation}
where the last two pairings are the Green pairing
(Def.~\ref{def:green-pairing}, and Lem.~\ref{lem:green-nondegen}) and
its descent to the respective cohomologies. But the global
parametrizability hypothesis specifies precisely that the above pairing
is non-degenerate and implies that $[\G^*[\tilde{\alpha}^*]]_{c^*} =
[0]$, or equivalently that there exists a $\tilde{\eps}^* \in
\Secs_0(\tilde{E}^*)$ such that
\begin{equation}
	\G^*[\tilde{\alpha}^*]
	= \dot{q}^*\circ\H^*[\tilde{\eps}^*]
	= \G^*\circ\dot{c}^*[\tilde{\eps}^*] .
\end{equation}
This, in turn, implies that $\G^*[\tilde{\alpha}^* -
\dot{c}^*[\tilde{\eps}^*]] = 0$. Hence, by the exact sequence of
Lem.~\ref{lem:exsplit}, we have that $\tilde{\alpha}^* = \dot{f}^*[\xi]
+ \dot{c}^*[\tilde{\eps}^*]$ for some $\xi\in \Secs_0(F)$. %\qed
\end{proof}

\subsection{Symplectic and Poisson structure}\label{sec:symp-pois}
In this section, we endow the space of solutions $\S_H(F)$ of a
variational PDE system with the structures of both a symplectic and a
Poisson manifold (or rather formal versions of these structures), really
turning it into the phase space of classical field theory.

In general, a variational system may have gauge symmetries. These must
be gauge fixed. The resulting system should then be put into the form of
a constrained hyperbolic system. Or, rather, what is most important for
us is that these steps can be carried out for the linearization of our
variational system (Sec.~\ref{sec:constr-gf}). If the constraints are
(a) globally parametrizable, (b) the gauge transformations globally
recognizable and (c) the gauge fixing satisfies an extra compatibility
condition, then we can apply a generalized Forger-Romero argument
(Sec.~\ref{sec:Jgreen}).  Next, we use the covariant phase space
formalism to build the formal symplectic form
(Sec.~\ref{sec:formal-symp}) and the Peierls formula to build the formal
Poisson bivector (Sec.~\ref{sec:formal-pois}). Finally, we prove that
the two structures are equivalent (Sec.~\ref{sec:equiv}).

For the remainder of this section, let us fix a Lagrangian density $\L
\in \Forms^{n,0}(F)$. Following Sec.~\ref{sec:var-sys}, it defines a
presymplectic current density $\omega \in \Forms^{n-1,2}(F)$ and its
Euler-Lagrange equations define a PDE system with equation form
$(\EL,\tilde{F}^*)$. 

\subsubsection{Variational systems, gauge fixing and constraints}\label{sec:constr-gf}
There are many reasons why the equation form $(\EL,\tilde{F}^*)$ of the
equations of motion of the classical field theory is not optimal for our
analysis. As we shall see later on, the Peierls formula calls for a
Green function of the linearized equations of motion. However, the
particular form of the differential operator $\EL$ may not be one that
directly falls into one of the classes of differential operators that
are easily recognized as hyperbolic, so that its linearizations possess
Green functions. For instance, in the presence of gauge invariance, we
must first adjoin a gauge fixing condition, say $c_g[\phi] = 0$ valued
in a bundle $E_g \to M$. Also, in many cases, either due to the extra
gauge fixing equations or due to internal integrability conditions
(Sec.~\ref{sec:integrability}), the equations can only be cast in
hyperbolic form with constraints. Of course, let us not forget that, a
priori, we haven't yet restricted the choice of $\L$ in any way that
would guarantee that its Euler-Lagrange system is not elliptic or of
some other hyperbolic type. So, we call the Euler-Lagrange system
\emph{hyperbolizable} if, after a possible gauge fixing, it can be shown
to be equivalent to a constrained hyperbolic system in a way that we
make precise below. From now on, we require that for a classical field
theory the Lagrangian $\L$ is chosen such that its Euler-Lagrange
equations are hyperbolizable. We shall see later in
Sec.~\ref{sec:examples}, that many relativistic field theories of
physical interest are in fact hyperbolizable.

Consider the gauge fixed Euler-Lagrange system, whose equation form is
$(\EL\oplus c_g, \tilde{F}^* \oplus_M E_g)$. It is hyperbolizable if it
equivalent (in the sense of Sec.~\ref{sec:integrability}) to a
constrained hyperbolic system $(f\oplus c, \tilde{F}^* \oplus_M E)$.
Again, we are not going into the details of what constitutes a
non-linear constrained hyperbolic system but defer instead to
Refs.~\citen{geroch-pde} and~\citen{kh-caus}. The equivalence must have the
following form:
\begin{equation}\label{eq:ELgf-equiv}
	\left\{
		\begin{aligned}
			\EL &= R\circ(f\oplus c)  \\
			c_g &= R_g\circ c  \\
		\end{aligned}
	\right.
	\quad \iff \quad
	\left\{
		\begin{aligned}
			f &= \bar{R}\circ(\EL\oplus c_g)  \\
			c &= \bar{R}_g \circ (\EL\oplus c_g) 
		\end{aligned}
	\right. ,
\end{equation}
where the $R$, $\bar{R}$, $R_g$ and $\bar{R}_g$ are possibly
non-linear differential operators.

As discussed earlier, in Sec.~\ref{sec:linear}, for the purposes of our
discussion, it is sufficient to pick a single dynamical linearization
point $\phi\in \S_H(F)$ and linearize the above PDE systems about it. In
particular, the linearized equations will be sufficient to define the
formal tangent and cotangent spaces $T_\phi\S$, $T_\phi^*\S$ and their
gauge invariant analogs $T_\phi\bar{\S}$, $T_\phi^*\bar{\S}$. In other
words, we need to work with the linearized versions of each of the
hyperbolic system, the constraints, the Euler-Lagrange system, the gauge
fixing conditions, the gauge transformations, as well as the
hyperbolization. As before, we denote the equation form of the
linearized hyperbolic system $(\dot{f},\tilde{F}^*)$. The linearized
constraints are presumed to be globally parametrizable and fit into the
commutative diagrams~\eqref{eq:glpar} and~\eqref{eq:glpar*}. The
linearized gauge transformations are presumed to be globally
recognizable and fit into the commutative diagrams~\eqref{eq:glrec}
and~\eqref{eq:glrec*}. The linearized $\EL$ equations are denoted
$(\J,\tilde{F}^*)$ and are also called the \emph{Jacobi system},
\begin{equation}
	\J[\psi]_a(x) = \J^I_{ab} \del_I \phi^b(x) = 0,
\end{equation}
with $\J$ the \emph{Jacobi operator},\cite{dewitt-qft} while the
linearized gauge fixing conditions are denoted by the equation form
$(\dot{c}_g,E_g)$. In local coordinates $(x^i,u^a)$ on $F$, the
components of the Jacobi operator satisfy the identity
\begin{equation}
	\J^I_{ab} \wedge \dv u^b_I = \dv \EL_a .
\end{equation}
The equivalence of the linearized systems takes the following form:
\begin{equation}\label{eq:Jgf-equiv}
	\left\{
		\begin{aligned}
			\J &= r\circ \dot{f} + r_c\circ \dot{c}  \\
			\dot{c}_g &= r_g\circ \dot{c}  \\
		\end{aligned}
	\right.
	\quad \iff \quad
	\left\{
		\begin{aligned}
			\dot{f} &= \bar{r}\circ \J + \bar{r}_c\circ \dot{c}_g \\
			\dot{c} &= \bar{r}_\J\circ \J + \bar{r}_g\circ \dot{c}_g
		\end{aligned}
	\right.
\end{equation}
If the operator $\bar{r}_\J$ is non-vanishing, it means that part of the
constraints consist of integrability conditions of the Jacobi system.

Note that, strictly speaking, the $r$- and $\bar{r}$- differential
operators effecting the equivalence are not inverses of each other.
Their compositions may differ from the identity by some differential
operator that factors through a differential identity, that is,
$\dot{q}\circ \dot{f} - \dot{h}\circ \dot{c} = 0$ or $\dot{g}^*\circ \J
= 0$. In other words, we must have
\begin{align}
\label{eq:rinv1}
	r\circ \bar{r} + r_c\circ\bar{r}_\J &= \id + p_\J \circ \dot{g}^* ,
		& \bar{r} \circ r &= \id + p_f \circ \dot{q} , \\
\label{eq:rinv2}
	r\circ \bar{r}_c + r_c\circ \bar{r}_g &= 0 ,
		& \bar{r}\circ r_c + \bar{r}_c\circ r_g &= -p_f\circ \dot{h} , \\
\label{eq:rinv3}
	r_g\circ \bar{r}_\J &= p_g \circ \dot{g}^* ,
		& \bar{r}_\J \circ r &= p_c\circ \dot{q} , \\
\label{eq:rinv4}
	r_g\circ \bar{r}_g &= \id ,
		& \bar{r}_g\circ r_g + \bar{r}_\J \circ r_c &= \id - p_c\circ \dot{h} ,
\end{align}
for some differential operators $p_\J$, $p_f$, $p_g$ and $p_c$. Also,
the identity $\dot{q}\circ \dot{f} - \dot{h}\circ \dot{c} = 0$, when
expressed in terms of the $\J$ and $\dot{c}_g$ operators, is identically
satisfied when
\begin{align}
\label{eq:consist1}
	\dot{q}\circ \bar{r} - \dot{h}\circ \bar{r}_\J &= q_\J\circ \dot{g}^*, \\
\label{eq:consist2}
	\dot{q}\circ \bar{r}_c - \dot{h}\circ \bar{r}_g &= 0 .
\end{align}
It is worth noting that the above relations involving the $r$- and
$\bar{r}$-operators follow from the equivalence~\eqref{eq:Jgf-equiv}
only when $\J$ and $\dot{c}_g$ satisfy no additional differential
identities. However, we will simply presume that they hold as needed
sufficient conditions for the derivation of the Peierls formulas later
in Sec.~\ref{sec:formal-pois}.

Finally, to make sure that the condition $\dot{c}_g[\psi] = 0$ in fact
constitutes a gauge fixing condition, we require the following
compatibility between the gauge transformation operator and the
constraints that we shall refer to as the \emph{gauge fixing
compatibility} condition:
\begin{equation}\label{eq:gf}
	\dot{s}' \circ \bar{r}_c\circ\dot{c}_g = 0 .
\end{equation}
This condition connects constraints (represented by $\dot{c}_g$) and
gauge transformations (represented by $\dot{s}'$). Roughly speaking,
this condition says that the part of the constraints $\dot{c}[\psi] = 0$
that comes from $\dot{c}_g[\psi] = 0$ is sufficient, when adjoined to
$\J[\psi] = 0$ to make the gauge fixed system hyperbolizable. This
compatibility condition becomes important later on, in
Lem.~\ref{lem:t*-J}, to show that the gauge invariant formal cotangent
space $T^*_\phi\bar{\S}$ can be equivalently defined in two ways,
involving either the $\J$ operator or the $\dot{f}$, $\dot{c}$
operators. Also, it helps prove that the hyperbolic differential
operator $\dot{k}$, that acts on gauge invariant field combinations in
the presence of recognizable gauge transformations, is actually
independent of the choice of gauge fixing operator $\dot{c}_g$ as long
as the compatibility condition is satisfied (see Cor.~\ref{cor:ginv}).

For future reference, it is convenient to state here the following
\begin{lemma}\label{lem:gf-equiv}
The gauge fixing compatibility condition Eq.~\eqref{eq:gf} is equivalent
to the existence of a differential operator $\bar{r}_s\colon \Secs(F)
\to \Secs(\tilde{P}^*)$ such that
\begin{equation}\label{eq:gf-equiv}
	\bar{r}_c\circ \dot{c}_g = \dot{s}\circ \bar{r}_s .
\end{equation}
That is, $\bar{r}_c\circ \dot{c}_g$ factors through $\dot{s}$.
\end{lemma}
\begin{proof}
This follows directly from the gauge compatibility
condition~\eqref{eq:gf} and Lem.~\ref{lem:fec-fact}.
\end{proof}

We summarize the conditions listed in this section in the following
\begin{definition}\label{def:hyper}
The Euler-Lagrange system $(\EL,\tilde{F}^*)$ or just the Jacobi system
$(\J,\tilde{F}^*)$ is said to be \emph{hyperbolizable} if the following
conditions are met: (a) there exists a gauge fixing and an equivalence
with a constrained hyperbolic system of the form~\eqref{eq:ELgf-equiv}
or~\eqref{eq:Jgf-equiv}, (b) the constraints are parametrizable and the
gauge transformations are recognizable with respect to the resulting
hyperbolic subsystem, and (c) the gauge fixing compatibility
condition~\eqref{eq:gf} is satisfied.
\end{definition}
The consequences of hyperbolizability are explored in the following
section. We stress that these conditions are sufficient for our purposes
and can in fact be satisfied by many relativistic field theories of
physical interest, but some of the same results could also hold under
weaker conditions.

\subsubsection{Causal Green functions}\label{sec:Jgreen}
The goal of this section is to use the gauge fixed
equivalence~\eqref{eq:Jgf-equiv} with a constrained hyperbolic system to
construct a causal Green function for the Jacobi system.

First, we show that the residual gauge transformations (those that are
still allowed by the gauge fixing condition $\dot{c}_g[\psi] = 0$)
essentially come from gauge parameters that satisfy the symmetric
hyperbolic equation $\dot{k}[\eps] = 0$. The main purpose of this lemma
is to serve as a reference argument for one of the sub-results of
Thm.~\ref{thm:Jexsplit}.
\begin{lemma}\label{lem:gauge-sol}
Given a $\psi\in \Secs_{SC}(F)$ such that $\dot{c}_g[\psi] = 0$
and $\psi\in \im \dot{g}$, there exists $\eps\in \Secs_{SC}(P)$ such
that $\dot{k}[\eps] = 0$ and $\psi = \dot{g}[\eps]$ precisely when the
image of a map (to be defined in the proof) $\gK\colon \ker\dot{s} \sse
\Secs_{SC}(\tilde{P}^*) \to H^g_{SC}(F,\dot{f})$ is trivial (see
Sec.~\ref{sec:gauge}).
\end{lemma}
\begin{proof}
Recall that the Jacobi system, due to its variational character is
easily shown to be self-adjoint:
\begin{equation}
	\J^* = \J .
\end{equation}
Also, gauge invariance and Noether's second theorem imply the identities
\begin{equation}\label{eq:J-ginv}
	\J \circ \dot{g} = 0
	\quad\text{and}\quad
	\dot{g}^* \circ \J = 0 .
\end{equation}
The equivalence of the gauge fixed Jacobi system with the constrained
hyperbolic system postulated in~\eqref{eq:Jgf-equiv} then gives
\begin{equation}
	\dot{s}\circ \dot{k}
	= \dot{f}\circ \dot{g}
	= \bar{r}_c \circ \dot{c}_g \circ \dot{g} .
\end{equation}
Suppose that $\psi \in \Secs_{SC}(F)$ such that $\dot{c}_g[\psi] = 0$
and $\psi = \dot{g}[\eps']$ for some $\eps'\in\Secs_{SC}(P)$. Let
$\tilde{\beta}^* = \dot{k}[\eps'] \in \Secs_{SC}(\tilde{P}^*)$ and note
that
\begin{equation}
	\dot{s}[\tilde{\beta}^*] = \dot{s}\circ \dot{k}[\eps']
	= \bar{r}_c\circ \dot{c}_g\circ \dot{g}[\eps']
	= \bar{r}_c\circ \dot{c}_g[\psi] = 0 .
\end{equation}
Then also, of course, $\dot{f}[\psi] = \dot{f}\circ \dot{g}[\eps'] =
\dot{s}\circ \dot{k}[\eps'] = \dot{s}[\tilde{\beta}^*] = 0$. Hence,
$\psi$ represents a cohomology class $[\psi]_g \in H^g_{SC}(F,\dot{f})$
(Sec.~\ref{sec:gauge}). The conclusion of this lemma holds precisely
when this cohomology class is trivial, $[\psi]_g$.

Let $\{\chi_\pm\}$ be a partition of unity adapted to a Cauchy surface
(Def.~\ref{def:adapt-pu}) and recall the associated splitting map
(Lem.~\ref{lem:exsplit}) $\K_\chi \colon \Secs_{SC}(\tilde{P}^*) \to
\Secs_{SC}(P)$ that inverts $\dot{k}$ from the right. We can make two
observations: (a) the difference $\eps'-\K_\chi[\tilde{\beta}^*]$ is in
$\S_{SC}(P)$, and (b) $\eta_\chi = \dot{g}[\K_\chi[\tilde{\beta}^*]] =
\G[\dot{s}_\chi[\tilde{\beta}^*]]$ is in $\S_{SC}(F)$, where
$\dot{s}_\chi[\tilde{\beta}^*] = \pm \dot{s}[\chi_{\pm}
\tilde{\beta}^*]$ has compact support since $\dot{s}[\tilde{\beta}^*] =
0$. Hence $\eta$ and $\psi$ define the same cohomology class,
$[\eta_\chi]_g = [\psi]_g$. Moreover, the choice of the adapted
partition of unity $\{\chi_\pm\}$ doesn't matter, since for any other
choice $\{\chi'_\pm\}$ the differences
$(\chi'_\pm-\chi_\pm)\tilde{\beta}^*$ have compact support, so that
$[\eta_{\chi'}] = [\eta_{\chi}]$.

Thus, the composition of maps
\begin{equation}
	\gK \colon \ker\dot{s} \sso \Secs_{SC}(\tilde{P}^*)
	\stackrel{\dot{g}\circ \K_\chi}{\longrightarrow}
	\S_{SC}(F) \to H^g_{SC}(F,\dot{f})
\end{equation}
is independent of the choice of the adapted partition of unity and, as
desired, its image coincides with the image of the subset of
$\S_{SC}(F)$ consisting of elements of the form $\psi = \dot{g}[\eps']$
with $\eps' \in \Secs_{SC}(P)$. Thus, each such $\psi = \dot{g}[\eps]$
where $\dot{k}[\eps] = 0$ precisely when the image of $\gK$ is trivial.
\end{proof}

It was remarked in the above proof that $\J \circ \dot{g} = 0$ and
$\dot{g}^*\circ \J = 0$. This is actually enough information to prove
that the Jacobi operator must factor through $\dot{g}'$ on the right and
through $\dot{g}^{\prime*}$ on the left.
\begin{lemma}\label{lem:J-fact}
There exists a differential operator $\J_g$ such that $\J =
\J_g \circ \dot{g}' = \dot{g}^{\prime*} \circ \J_g^*$.
\end{lemma}
\begin{proof}
This is a simple consequence of Eq.~\eqref{eq:J-ginv},
Lem.~\ref{lem:fec-fact}, and the self-adjointness of $\J$.
\end{proof}

Next, we prove a lemma that will be used to establish an alternative
characterization of the formal gauge invariant solutions cotangent space
(Def.~\ref{def:tt*-sols-gauge}) in the main theorem of this section.
\begin{lemma}\label{lem:t*-J}
Provided the sufficient conditions listed in Sec.~\ref{sec:constr-gf}
hold, any compactly supported dual density $\tilde{\alpha}^* \in
\Secs_0(\tilde{F}^*)$ that satisfies
\begin{equation}
	\dot{g}^*[\tilde{\alpha}^*]
	= \dot{g}^*[\dot{f}^*[\psi] + \dot{c}^*[\tilde{\beta}]] ,
\end{equation}
for some $\psi \in \Secs_0(F)$ and $\tilde{\beta}^* \in
\Secs_0(\tilde{E}^*)$, can be written as $\tilde{\alpha}^* = \J[\xi]$,
for some $\xi\in \Secs_0(F)$.
\end{lemma}
\begin{proof}
We start by defining a field section $\psi_+$ with retarded support such
that $\dot{f}^*[\psi_+] = \dot{f}^*[\psi] + \dot{c}^*[\tilde{\beta}^*]$.
If it is simple to check that we can define $\psi_+ = \psi +
\dot{q}^*\circ \H^*_+[\tilde{\beta}^*]$. (Using advanced support would
have also been possible.) Then,
\begin{equation}
	\dot{k}^*\circ\dot{s}^*[\psi_+]
	= \dot{g}^*\circ\dot{f}^*[\psi_+]
	= \dot{g}^*[\tilde{\alpha}^*]
	= 0 .
\end{equation}
Since $\dot{k}^*$ is invertible on sections of retarded support, we must
have $\dot{s}^*[\psi_+] = 0$. On the other hand,
\begin{align}
	\tilde{\alpha}^*
	&= \dot{f}^*[\psi_+]
	= \J^*\circ\bar{r}^*[\psi_+] + \dot{c}_g^*\circ\bar{r}_c[\psi_+] \\
	&= \J[\bar{r}^*[\psi] + \bar{r}^*\circ\dot{q}^*\circ
			\H^*_+[\tilde{\beta}^*]]
		+ \bar{r}_s^*[\dot{s}^*[\psi_+]] \\
	&= \J[\bar{r}^*[\psi] + \dot{g}\circ q_\J^*\circ
			\H^*_+[\tilde{\beta}^*] + \bar{r}_\J^* \circ (\dot{h}^* \circ
			\H^*_+)[\tilde{\beta}^*]] \\
	&= \J[\bar{r}^*[\psi] + \bar{r}_\J^*[\tilde{\beta}^*]] ,
\end{align}
where we have used the equivalence~\eqref{eq:Jgf-equiv}, the formal
self-adjointness of $\J$, Lem.~\ref{lem:gf-equiv}, and the
identities~\eqref{eq:consist1} and $\J\circ\dot{g} = 0$. Therefore, the
desired conclusion holds with $\xi = \bar{r}^*[\psi] +
\bar{r}_\J^*[\tilde{\beta}^*]$.
\end{proof}

Finally, we motivate the Peierls formula and then state and prove the
main theorem of this section. Equivalence with a constrained hyperbolic
system now allows us to solve the inhomogeneous problem
\begin{equation}
	\J[\psi] = \tilde{\alpha}^*, \quad \dot{c}_g[\psi] = 0 ,
\end{equation}
where the source must necessarily satisfy the gauge invariance condition
$\dot{g}^*[\tilde{\alpha}^*] = 0$. The equivalent inhomogeneous problem in
symmetric hyperbolic form is
\begin{equation}
	\dot{f}[\psi] = \bar{r}[\tilde{\alpha}^*], \quad
	\dot{c}[\psi] = \bar{r}_\J[\tilde{\alpha}^*] .
\end{equation}
Recall from Lem.~\ref{lem:inhom-constr} that this system is solvable iff
the sources satisfy the consistency identity:
\begin{equation}
	\dot{q}[\bar{r}[\tilde{\alpha}^*]]
		- \dot{h}[\bar{r}_\J[\tilde{\alpha}^*]]
	= p_\J\circ \dot{g}^*[\tilde{\alpha}^*]
	= 0 ,
\end{equation}
which is obviously satisfied, after using identity~\eqref{eq:consist1}, for any gauge
invariant source. The retarded and advanced solutions to this
inhomogeneous problem are then $\psi_\pm = \G_\pm
[\bar{r}[\tilde{\alpha}^*]]$. This means that
$\bar{r}_\J[\tilde{\alpha}^*] =
\dot{c}[\G_\pm[\bar{r}[\tilde{\alpha}^*]]$ and, in particular,
$\dot{c}_g[\psi] = 0$. Motivated by this formula, we introduce the
following retarded, advanced and causal Green functions for the gauge
fixed Jacobi system.
\begin{definition}\label{def:E}
Let $\EE_\pm = \G_\pm \circ \bar{r}$. The \emph{Peierls formula} is
\begin{equation}\label{eq:E-def}
	\EE = \EE_+ - \EE_- = \G\circ \bar{r} .
\end{equation}
We also call $\EE$ the \emph{Peierls} or \emph{Jacobi causal Green
function}.
\end{definition}
One can immediately check that $\psi = \EE[\tilde{\alpha}^*]$ satisfies
both $\dot{f}[\psi] = 0$ and $\dot{c}[\psi] = 0$, whenever
$\tilde{\alpha}^*$ is a gauge invariant dual density. By the
equivalence~\eqref{eq:Jgf-equiv}, the same solution also satisfies
$\J[\psi] = 0$ and $\dot{c}_g[\psi] = 0$.

\begin{theorem}\label{thm:Jexsplit}
Provided the gauge fixed Jacobi system $\J[\psi] = 0$, $\dot{c}_g[\psi]
= 0$ is hyperbolizable (Def.~\ref{def:hyper}), the Jacobi causal Green
function $\EE$ defined in Eq.~\eqref{eq:E-def} fits into the following
commutative diagram
\begin{equation}\label{eq:JE-seq}
\vcenter{\xymatrix{
	%0
		                  %\ar[r] &
	& \Secs_0(P)
		\ar[d]^{\dot{g}}  \ar[r] %&
	& 0
		\ar[d]            \ar[r] %&
	& \Secs_{SC}(P)
		\ar[d]^{\dot{g}}  \ar[r] %&
	& 0
		\ar[d]            %\ar[r] &
	%0
	\\
	0
		                   \ar[r] &
	\Secs_0(F)
		\ar[d]             \ar[r]^{\J} &
	\Secs_0(\tilde{F}^*)
		\ar[d]^{\dot{g}^*} \ar[r]^{\EE} &
	\Secs_{SC}(F)
		\ar[d]             \ar[r]^{\J} &
	\Secs_{SC}(\tilde{F}^*)
		\ar[d]^{\dot{g}^*} \ar[r] &
	0 ,
	\\
	%0
		                  %\ar[r] &
	& 0
		                  \ar[r] %&
	& \Secs_0(\tilde{P}^*)
		                  \ar[r] %&
	& 0
		                  \ar[r] %&
	& \Secs_{SC}(\tilde{P}^*)
		                  %\ar[r] &
	%0
}}
\end{equation}
which becomes a complex (successive arrows compose to $0$) after taking
the vertical cohomologies. That complex is exact at
$\Secs_0(\tilde{F}^*)$ and $\Secs_{SC}(F)$, while at $\Secs_0(F)$ and
$\Secs_{SC}(\tilde{F}^*)$ its cohomologies coincide respectively with
$H^g_0(F)$ and $H^{g^*}_{SC}(\tilde{F}^*)$.

Moreover, with reference to Def.~\ref{def:tt*-sols-gauge}, we have the
isomorphisms $T_\phi\bar{\S} \cong \ker \J / \im \dot{g}$, at
$\Secs_{SC}(F)$, and $T_\phi^*\bar{\S} \cong \ker\dot{g}^* / \im \J$, at
$\Secs_0(\tilde{F}^*)$.

Finally, a Cauchy surface $\Sigma\sso M$ and a partition of unity
$\{\chi_\pm\}$ adapted to it define the following splittings at
$\Secs_0(\tilde{F}^*)$ and $\Secs_{SC}(F)$:
\begin{equation}\label{eq:Jexsplit}
	\ker \dot{g}^* \cong \im \J \oplus T_\phi\S
	\quad\text{and}\quad
	\coker \dot{g} \cong T_\phi\bar{\S} \oplus \im\dot{g}^{\prime*} ,
\end{equation}
where (cf.~Def.~\ref{def:adapt-pu}, Lem.~\ref{lem:exsplit}, and
Def.~\ref{def:tt*-sols})
\begin{align}
\label{eq:Jsplit}
	\J_\chi\colon& T_\phi\S \to \Secs_0(\tilde{F}^*) ,
		& \J_\chi[\psi] &= \pm \J[\chi_\pm \psi] , \\
\label{eq:Esplit}
	\EE_\chi\colon& \im\dot{g}^{\prime*} \to \Secs_{SC}(F) ,
		& \EE_\chi[\tilde{\alpha}^*]
			&= \bar{r}^*\circ \G^*_\chi[\tilde{\alpha}^*] .
\end{align}
\end{theorem}
% XXX: Note, the isomorphism $\S_{SC}(F)/\im\dot{g}$ with $\ker\J /
% \im\dot{g}$ shows that every solution class can be represented by one
% that satisfies the gauge fixing condition. Therefore, it is equally
% acceptable to use gauge equivalence classes of gauge fixed
% solutions(!) to represent arbitrary equivalence classes of solutions.
The conclusion of the theorem is rather dense with information, so its
proof is somewhat lengthy. However it simply consists of checking the
properties of the horizontal sequence in the above diagram at each of
its objects.
\begin{proof}
The fact that successive maps compose to zero, after taking the vertical
cohomologies, is established in items~(2) and~(3) below, which also
prove exactness of the resulting complex at $\Secs_0(\tilde{F}^*)$ and
$\Secs_{SC}(F)$. On the other hand, cohomologies at $\Secs_0(F)$ and
$\Secs_{SC}(\tilde{F}^*)$ are computed in items~(1) and~(4).

The isomorphism $T_\phi\bar{\S} \cong \ker\J/\im\dot{g}$ is established
as follows: on the one hand, it is obvious that $\dot{f}[\psi] = 0$,
$\dot{c}[\psi] = 0$ implies $\J[\psi] = 0$; on the other hand, item~(3)
shows that both $\dot{f}\circ \EE = 0$ and $\dot{c}\circ \EE = 0$, while
$\im \EE = \ker \J \pmod{\im \dot{g}}$. The isomorphism
$T_\phi^*\bar{\S} \cong \ker \dot{g}^* / \im \J$ is established as
follows: on the one hand, by the equivalence~\eqref{eq:Jgf-equiv}, it is
obvious that $\tilde{\alpha}^* + \J[\xi]$, with
$\dot{g}^*[\tilde{\alpha}^*] = 0$ and $\xi$ arbitrary, represents a
unique element of $T_\phi^*\bar{S}$; on the other hand, if
$\dot{g}^*[\tilde{\alpha}^*] = \dot{g}^*[\dot{f}^*[\psi] +
\dot{c}^*[\tilde{\beta}^*]]$, then $[\tilde{\alpha}^*] =
[\tilde{\alpha}^* - \dot{f}^*[\psi] - \dot{c}^*[\tilde{\beta}^*]]$ in
$T_\phi^*\bar{\S}$ and Lem.~\ref{lem:t*-J} shows that any representative
of $[0] \in T_\phi^*\bar{\S}$ represents only $[0] \in \ker\dot{g}^* /
\im \J$.

Finally, the splittings~\eqref{eq:Jexsplit}, with the corresponding
splitting map identities, are established in items~(3) and~(4).

Note that below we make liberal use of various maps defined using the
adapted partition of unity $\{\chi_\pm\}$ introduced in the hypothesis
of the theorem (cf.~Lem.~\ref{lem:exsplit}).
\begin{enumerate}
\item
	\emph{If $\psi\in \Secs_0(F)$, then $\J[\psi] = 0$ is equivalent to
	$\dot{g}'[\psi] = 0$.}

	If $\psi\in \ker \dot{g}'$, then by Lem.~\ref{lem:J-fact} ($\J = \J_g
	\circ \dot{g}'$) we certainly have $\psi \in \ker \J$. On the other
	hand, if $\psi \in \ker \J$, then
	\begin{equation}
		\dot{k}'[\dot{g}'[\psi]]
		= \dot{s}'[\dot{f}[\psi]]
		= \dot{s}'[\bar{r}\circ \J[\psi] + \bar{r}_c\circ \dot{c}_g[\psi]]
		= \dot{s}'\circ \bar{r}_c\circ \dot{c}_g [\psi]
		= 0 ,
	\end{equation}
	where the last equality holds due to the gauge fixing compatibility
	condition~\eqref{eq:gf} and we have also used the
	equivalence~\eqref{eq:Jgf-equiv}. But, since $\dot{k}'$ is injective
	on $\Secs_0(P')$, this can only be if $\psi \in \ker \dot{g}'$.
	Therefore, after taking vertical cohomologies, the cohomology
	of~\eqref{eq:JE-seq} at $\Secs_0(F)$ is isomorphic to
	$H^g_0(F)$.

\item
	\emph{At $\Secs_0(\tilde{F}^*)$, we have $\EE\circ \J = 0
	\pmod{\im\dot{g}}$.  Any $\tilde{\alpha}^* \in \Secs_0(\tilde{F}^*)$
	such that $\dot{g}^*[\tilde{\alpha}^*] = 0$ and $\EE[\tilde{\alpha}^*]
	= \dot{g}[\eps]$, with $\eps\in \Secs_{SC}(P)$, can be written as
	$\tilde{\alpha}^* = \J[\xi]$, with $\xi\in \Secs_0(F)$.}

	For the first part, recall the equivalent form~\eqref{eq:gf-equiv} of
	the gauge fixing compatibility condition. Direct calculation, with
	$\xi\in \Secs_0(F)$, then gives
	\begin{align}
		\EE[\J[\xi]]
		&= \G[\bar{r}\circ\J[\xi]] \\
		&= \G[\dot{f}[\xi] - \bar{r}_c\circ\dot{c}_g[\xi]] \\
		&= -\G[\dot{s}\circ \bar{r}_s[\xi]] \\
		&= \dot{g}[\K[\bar{r}_s[\xi]]] .
	\end{align}

	For the second part, let $\psi = \EE[\tilde{\alpha}] = \dot{g}[\eps]
	\in \Secs_{SC}(F)$ and $\tilde{\beta}^* = \dot{k}[\eps] \in
	\Secs_{SC}(\tilde{P}^*)$. Note that $\dot{f}[\psi] = \dot{f}\circ
	\G[\bar{r}[\tilde{\alpha}^*]] = 0$ and also $\dot{s}[\tilde{\beta}^*]
	= \dot{s}\circ \dot{k}[\eps] = \dot{f}\circ \dot{g}[\eps] = 0$. Using
	the same logic and notation as in the proof of
	Lem.~\ref{lem:gauge-sol}, we can write $\psi =
	\G[\dot{s}_\chi[\tilde{\beta}^*]+\dot{s}[\tilde{\gamma}^*]] = \G\circ
	\dot{s}[\chi_+\tilde{\beta}^* + \tilde{\gamma}^*]$, for some
	$\tilde{\gamma}^* \in \Secs_0(\tilde{P}^*)$. Note that the argument of
	$\G$ has compact support. Recalling the definition $\psi =
	\G[\bar{r}[\tilde{\alpha}]]$ and the fact that $\ker\G = \im\dot{f}$,
	there must exist a $\xi\in \Secs_0(F)$ such that
	\begin{equation}
		\dot{f}[\xi] = \bar{r}[\tilde{\alpha}^*]
			- \dot{s}[\chi_+\tilde{\beta}^* + \tilde{\gamma}^*] .
	\end{equation}
	By uniqueness of solutions with retarded support, we must have $\xi =
	\G_+[\bar{r}[\tilde{\alpha}^*] - \dot{s}[\chi_+\tilde{\beta}^* +
	\tilde{\gamma}^*]] = \G_+[\bar{r}[\tilde{\alpha}^*]] -
	\dot{g}\circ\K_+[\chi_+\tilde{\beta}^* + \tilde{\gamma}^*]$ (we could
	have also used $\G_-$ with the substitution $\chi_+ \to -\chi_-$).
	Then, direct calculation shows
	\begin{align}
		\J[\xi]
		&= \J[\G_+\circ\bar{r}[\tilde{\alpha}^*]
			- \dot{g}\circ \K_+[\chi_+\tilde{\beta}^*+\tilde{\gamma}^*]] \\
		&= r\circ (\dot{f}\circ \G_+)\circ\bar{r}[\tilde{\alpha}^*]
			+ r_c\circ (\dot{c}\circ \G_+)\circ \bar{r}[\tilde{\alpha}^*] \\
		&= (r\circ\bar{r})[\tilde{\alpha}^*]
			+ r_c\circ\H_+\circ(\dot{q}\circ\bar{r})[\tilde{\alpha}^*] \\
		&= (\id + p_\J\circ \dot{g}^* - r_c\circ \bar{r}_\J)[\tilde{\alpha}^*]
			+ r_c\circ\H_+\circ (\dot{h}\circ\bar{r}_\J + q_\J\circ\dot{g}^*)
					[\tilde{\alpha}^*] \\
		&= \tilde{\alpha}^*
			- r_c\circ \bar{r}_\J[\tilde{\alpha}^*]
			+ r_c\circ (\H_+\circ\dot{h})\circ \bar{r}_\J[\tilde{\alpha}^*] \\
		&= \tilde{\alpha}^* ,
	\end{align}
	where we have used identities $\J\circ \dot{g} = 0$, \eqref{eq:rinv1}
	\eqref{eq:consist1} and the commutative diagram~\eqref{eq:glpar}.

\item
	\emph{At $\Secs_{SC}(F)$, we have $\J\circ \EE = 0$ when restricted to
	$\ker\dot{g}^*$. Any $\psi\in \Secs_{SC}(F)$ such that $\J[\psi] = 0$
	can be written as $\psi = \EE[\tilde{\alpha}^*] + \dot{g}[\eps]$, with
	$\eps\in \Secs_{SC}(P)$.}

	For the first part, with $\tilde{\alpha}^* \in \Secs_0(\tilde{F}^*)$
	and $\dot{g}^*[\tilde{\alpha}^*] = 0$, direct calculation gives
	\begin{align}
		\dot{f}\circ \EE[\tilde{\alpha}^*]
		&= (\dot{f}\circ \G) \circ\bar{r}[\tilde{\alpha}^*] = 0 , \\
		\dot{c}\circ \EE[\tilde{\alpha}^*]
		&= (\dot{c}\circ \G)\circ\bar{r}[\tilde{\alpha}^*]
		= \H\circ (\dot{q}\circ\bar{r})[\tilde{\alpha}^*] \\
		&= (\H\circ\dot{h})\circ\bar{r}_\J[\tilde{\alpha}^*]
				+ \H\circ q_\J[\dot{g}^*[\tilde{\alpha}^*]]
		= 0 , \\
		\J\circ \EE[\tilde{\alpha}^*]
		&= (r\circ (\dot{f}\circ\EE) + r_c\circ (\dot{c}\circ\EE))
				[\tilde{\alpha}^*] ,
	\end{align}
	where we have used the identity~\eqref{eq:consist1} and exactness of
	the sequence in Prp.~\ref{prp:exact}.

	For the second part, let $\tilde{\alpha}^* = \J_\chi[\psi] = \pm
	\J[\chi_\pm \psi]$, so that $\tilde{\alpha}^* \in
	\Secs_0(\tilde{F}^*)$. We claim that $\psi' = \EE[\tilde{\alpha}^*]$
	differs from $\psi$ only by a pure gauge term $\dot{g}[\eps]$, with
	$\eps \in \Secs_{SC}(P)$. We examine closely the following expression,
	which appears in the definition of $\psi'$,
	\begin{align}
		\bar{r}\circ \J_\chi[\psi]
		&= \pm [(\bar{r}\circ r)\circ \dot{f}[\chi_\pm\psi]
				+ (\bar{r}\circ r_c)\circ \dot{c}[\chi_\pm\psi]] \\
		&= \pm [(\id + p_f\circ \dot{q})\circ \dot{f}[\chi_\pm\psi]
				- (\bar{r}_c\circ r_g + p_f\circ \dot{h})\circ \dot{c}[\chi_\pm\psi]] \\
		&= \pm [\dot{f}[\chi_\pm\psi] - (\bar{r}_c\circ \dot{c}_g)[\chi_\pm\psi]
				+ p_f\circ(\dot{q}\circ\dot{f} - \dot{h}\circ\dot{c})[\chi_\pm\psi]] \\
		&= \pm (\dot{f}[\chi_\pm\psi] - \dot{s}\circ \bar{r}_s[\chi_\pm\psi]) ,
	\end{align}
	where we have used the identities~\eqref{eq:rinv1}
	and~\eqref{eq:rinv2}. Note that, while the above expression has
	compact support, the two individual terms do not. It is important to
	be able to decompose this expression in two different ways: into terms
	having retarded support ($+$) or advanced support ($-$). Direct
	calculation then shows
	\begin{align}
		\psi'
		&= \G[\bar{r}\circ \J_\chi[\psi]] \\
		&= \G_+[\dot{f}[\chi_+\psi] - \dot{s}\circ\bar{r}_s[\chi_+\psi]]
			- \G_-[-\dot{f}[\chi_-\psi] + \dot{s}\circ\bar{r}_s[\chi_-\psi]] \\
		&= (\chi_+ + \chi_-)\psi
			- (\G_+\circ\dot{s})\circ\bar{r}_s[\chi_+\psi]
			- (\G_-\circ\dot{s})\circ\bar{r}_s[\chi_-\psi] \\
		&= \psi
			- \dot{g}[\K_+\circ \bar{r}_s[\chi_+\psi]
				+ \K_-\circ \bar{r}_s[\chi_-\psi]] .
	\end{align}
	Therefore, the desired conclusion holds, $\psi' =
	\EE\circ\J_\chi[\psi] = \dot{g}[\eps]$, with $\eps = \K_+\circ
	\bar{r}_s[\chi_+\psi] + \K_-\circ \bar{r}_s[\chi_-\psi] \in
	\Secs_{SC}(P)$.

\item
	\emph{If $\tilde{\alpha}^*\in \Secs_{SC}(\tilde{F}^*)$, then
	$\tilde{\alpha}^* = \J[\psi]$, for some $\psi\in \Secs_{SC}(F)$, is
	equivalent to $\tilde{\alpha}^* = \dot{g}^{\prime*}[\tilde{\beta}^*]$,
	for some $\tilde{\beta}^* \in \Secs_{SC}(\tilde{P}^{\prime*})$.}

	If $\tilde{\alpha}^* = \J[\psi]$, then by Lem.~\ref{lem:J-fact} ($\J =
	\dot{g}^{\prime*}\circ \J_g^*$) we certainly have $\tilde{\alpha}^*
	\in \im\dot{g}^{\prime*}$. On the other hand, if $\tilde{\alpha}^* =
	\dot{g}^{\prime*}[\tilde{\beta}^*]$, let $\xi =
	\K^{\prime*}_\chi[\tilde{\beta}^*] \in \Secs_{SC}(P')$. Direct
	calculation then shows
	\begin{align}
		\J[\bar{r}^*\circ\dot{s}^{\prime*}[\xi]]
		&= (\J^*\circ \bar{r}^*)\circ\dot{s}^{\prime*}[\xi] \\
		&= \dot{f}^*\circ \dot{s}^{\prime*}[\xi]
			- (\dot{c}_g^*\circ \bar{r}_c^*\circ \dot{s}^{\prime*})[\xi] \\
		&= \dot{g}^{\prime*}\circ
				(\dot{k}^{\prime*} \circ \K^{\prime*}_\chi)[\tilde{\beta}^*]
		= \dot{g}^{\prime*}[\tilde{\beta}^*] \\
		&= \tilde{\alpha}^* ,
	\end{align}
	where we have used the formal adjoint versions of
	Eqs.~\eqref{eq:Jgf-equiv} and the gauge fixing compatibility
	condition~\eqref{eq:gf}. Therefore, the desired conclusion holds, with
	$\psi = \bar{r}^*\circ\dot{s}^{\prime*}[\xi] \in \Secs_{SC}(F)$.
	Simplifying the last expression, we get $\psi = \bar{r}^*\circ
	(\dot{s}^{\prime*}\circ \K^{\prime*}_\chi)[\tilde{\beta}^*]
	= \bar{r}^*\circ \G^*_\chi[\dot{g}^{\prime*}[\tilde{\beta}^*]]
	= \bar{r}^*\circ \G^*_\chi[\tilde{\alpha}^*]$ and hence
	\begin{equation}
		\J\circ \EE_\chi[\tilde{\alpha}^*]
		= \J\circ \bar{r}^*\circ \G^*_\chi[\tilde{\alpha}^*]
		= \tilde{\alpha}^* .
	\end{equation}
	Therefore, we have established that, after taking vertical
	cohomologies, the cohomology of~\eqref{eq:JE-seq} at
	$\Secs_{SC}(\tilde{F}^*)$ is isomorphic to
	$H^{g^*}_{SC}(\tilde{F}^*)$.

	%\qed
\end{enumerate}
\end{proof}
\begin{remark}\label{rem:loc-par-rec}
The hypotheses of Thm.~\ref{thm:Jexsplit} make use of the notion of
hyperbolizability (Def.~\ref{def:hyper}), which in turn requires the
corresponding constraints to be parametrizable (Sec.~\ref{sec:constr})
and the gauge transformations to be recognizable (Sec.~\ref{sec:gauge}).
But only the local versions of these were used. That is,
Thm.~\ref{thm:Jexsplit} holds even if \emph{global} parametrizability
and recognizability fail. The global conditions instead will appear in
the representation~\eqref{eq:symp-rep} of the formal presymplectic form
with respect to the natural pairing between the formal tangent and
cotangent spaces.
\end{remark}

As mentioned before, it is easy to see from its variational nature that
the Jacobi operator is self-adjoint $\J^* = \J$. If it were directly
invertible, the Green functions $\EE_\pm$ would satisfy the same
relation with their adjoints as shown in Sec.~\eqref{sec:green-adj},
making the causal Green function anti-self-adjoint, $(\EE)^* = -\EE$.
However, due to gauge invariance the relation of the gauge fixed Green
functions to their adjoints is more complicated.
\begin{lemma}\label{lem:E-anti}
When restricted to act on gauge invariant dual densities, the causal
Green function of the gauge fixed Jacobi system is anti-self-adjoint up
to gauge:
\begin{equation}
	(\EE)^* = -\EE \pmod{\im \dot{g}} .
\end{equation}
\end{lemma}
\begin{proof}
First, note that from identities~\eqref{eq:Jgf-equiv}
and~\eqref{eq:rinv1} we have
\begin{align}
	\J\circ\EE_\pm[\tilde{\alpha}^*]
	&= \sum_\pm (r\circ \dot{f} + r_c\circ \dot{c})\circ \G_\pm \circ
		\bar{r}[\tilde{\alpha}^*] \\
	&= r\circ \bar{r}[\tilde{\alpha}^*] + r_c\circ \bar{r}_\J[\tilde{\alpha}^*]\\
	&= (\id + p_\J\circ \dot{g}^*)[\tilde{\alpha}^*] .
\end{align}
It then follows that
\begin{align}
	(\EE_\mp)^* \circ \J \circ \EE_\pm
		&= (\EE_\mp)^* \circ (\J \circ \EE_\pm)
		= (\EE_\mp)^* + (\EE_\mp)^*\circ p_\J \circ \dot{g}^* , \\
	(\EE_\mp)^* \circ \J \circ \EE_\pm
		&= ((\EE_\mp)^*\circ \J^*) \circ \EE_\pm
		= \EE_\pm + \dot{g}\circ p_\J^* \circ \EE_\pm , \\
	\text{and hence}\quad
	\EE_\pm &= (\EE_\mp)^* + (\EE_\mp)^*\circ p_\J \circ \dot{g}^*
		- \dot{g}\circ p_\J^* \circ \EE_\pm .
\end{align}
Given that $\EE = \EE_+ - \EE_-$, we then have
\begin{equation}\label{eq:E-anti}
	\EE = -(\EE)^* - \dot{g}\circ p_\J^*\circ \EE
		- (\EE)^*\circ p_\J^* \circ \dot{g}^* ,
\end{equation}
which gives the desired conclusion. %\qed
\end{proof}

We conclude this section by drawing attention to the fact that the kind
of gauge fixing that features in a hyperbolization, as discussed in
Sec.~\ref{sec:constr-gf} is a special kind of partial gauge fixing. We
refer to it as \emph{purely hyperbolic}. Any further gauge fixing
conditions are then called \emph{residual}. We leave the consideration
of residual gauge fixing to future work. A principal difficulty in
dealing with residual gauge fixing conditions is that the resulting
constraints are no longer parametrizable (such as operators that are
elliptic on a family of spatial slices). Thus, the kernel of the gauge
fixing conditions may contain very few, if any solutions with spacelike
compact support, which would be difficult to fit into the current formal
framework for tangent and cotangent spaces to the space of solutions.

\subsubsection{Formal symplectic structure}\label{sec:formal-symp}
Below, we construct a formal symplectic form $\bar{\Omega}$ using the
covariant phase space formalism. That is, we will integrate the
presymplectic current density $\omega$, derived in
Sec.~\ref{sec:var-sys}, over a Cauchy surface. Any Cauchy surface would
do, giving the same result. The resulting form can in general be
degenerate, though, and only becomes symplectic once projected to the
gauge invariant formal tangent space $T_\phi \bar{\S}$.
\begin{definition}\label{def:formal-symp}
Consider a variational system (Sec.~\ref{sec:var-sys}) with presymplectic form $\omega \in
\Forms^{n-1,2}(F)$ (\ref{sec:jets}). Suppose that it is hyperbolizable
(Sec.~\ref{sec:constr-gf}) and $\phi\in \S_H(F)$ is a background
solution with good causal behavior (Sec.~\ref{sec:linear}), so that the
linearized equations of motion endow $M$ with a globally hyperbolic
causal structure (\ref{sec:conal}). Then, given a Cauchy surface
$\Sigma\sso M$, we define the formal \emph{presymplectic} $2$-form
$\Omega$ on the formal tangent space $T_\phi\bar{\S}$ by the formula
\begin{equation}
 	\Omega(\psi,\xi)
 	= \int_\Sigma\omega[\psi,\xi]
 	= \int_\Sigma(j^\oo\phi)^* [\iota_{\hat{\xi}}\iota_{\hat{\psi}}\omega] .
\end{equation}
\end{definition}
Recall that for any section $\psi,\xi\in \Secs(F)$ we can define the
prolonged evolutionary vector fields $\hat{\psi},\hat{\xi}$ on $J^\oo
F)$ (\ref{sec:jets}), which can be then be contracted with $\omega \in
\Forms^*(J^\oo F)$.

Ideally, we would now show that $\Omega$ defines a smooth, closed
differential form on the possibly infinite dimensional space of
solutions $\S_H(F)$. However, we would then need to make explicit use of
the infinite dimensional differential structure on $\S_H(F)$ and
$T\S_H(F)$, which we have consistently avoided doing in this review,
preferring a formal approach, with minimal analytical details. So
instead, we will settle for showing that it is \emph{formally smooth}
and \emph{closed}. These names are simply place holders for the
identities demonstrated in the proof of the following
\begin{lemma}
Under the hypotheses of Def.~\ref{def:formal-symp}, the definition of
$\Omega$ is independent of the choice of Cauchy surface $\Sigma \sso M$.
Moreover, $\Omega$ is formally closed.
\end{lemma}
\begin{proof}
First, we note that if $\chi,\xi\in T_\phi\S$ then the integral defining
$\Omega$ is necessarily finite, since the integrand $\omega[\chi,\xi] =
(j^\oo\phi)^*[\iota_{\hat{\xi}}\iota_{\hat{\chi}}\omega]$ has spacelike
compact support as both $\chi$ and $\xi$ do. Independence of the choice
of $\Sigma$ follows if we can show that $\omega[\chi,\xi]$ is de~Rham
closed on $M$. This follows directly from the horizontal, on-shell
closedness of $\omega$ in the variational bicomplex
(Lem.~\ref{lem:omega-cl}):
\begin{equation}
	\d[(j^\oo\phi)^*\iota_{\hat\xi}\iota_{\hat\chi}\omega]
	= (j^\oo\phi)^*\dh[\iota_{\hat\xi}\iota_{\hat\chi}\omega]
	= (j^\oo\phi)^*\iota_{\hat\xi}\iota_{\hat\chi}[\dh\omega]
	= 0 .
\end{equation}

For a fixed background solution $\phi\in \S_H(F)$, the formal $2$-form
$\Omega(\chi,\xi)$ is defined as a Cauchy surface integral of a
bidifferential operator $\omega[\chi,\xi]$, which is defined by a form
$\omega\in \Omega^*(J^\oo F)$. Hence we are happy to declare $\Omega$ to
be \emph{formally smooth} in its dependence on $\phi$, as long as
$\omega$ itself is smooth, which it is by construction. Also, in this
simple case, we are justified in declaring the formal de~Rham
differential $\delta$ on $\S_H(F)$ to act on $\Omega$ as the vertical
differential $\dv$ under the integral sign. Therefore, in this context,
it is straight forward to check that $\Omega$ is \emph{formally closed}
since $\omega$ is on-shell, vertically closed (Lem.~\ref{lem:omega-cl}):
\begin{align}
	(\delta\Omega)(\chi,\xi,\psi)
	&= \int_\Sigma (j^\oo\phi)^*
		[\iota_{\hat\psi}\iota_{\hat\xi}\iota_{\hat\chi}\dv\omega]
	= \int_\Sigma (j^\oo\phi)^*
		[\iota_{\hat\psi}\iota_{\hat\xi}\iota_{\hat\chi}\dv\omega] \\
	&= \int_\Sigma (j^\oo\phi)^*
		[\iota_{\hat\psi}\iota_{\hat\xi}\iota_{\hat\chi} (\dv\omega)]
	= 0 .
\end{align}
This concludes the proof. %\qed
\end{proof}
Though this was not attempted in Refs.~\citen{fr-bv,rejzner-thesis,bfr},
their rigorous setting for infinite dimensional geometry can be used to
remove the formal character of the above lemma. Also, it is quite clear
from the proof that the integration surface $\Sigma$ in the definition
of $\Omega$ need not actually be a Cauchy surface. It need only be in
the same homology%
 	\footnote{The appropriate homology theory here should correspond to a
 	variant of locally finite Borel-Moore homology, where one considers
 	only chains whose intersection with every spacelike compact set
 	is compact. This variant does not appear to have gotten any attention
 	in the literature and thus deserves further study.} %
class as a Cauchy surface. In particular, it is enough that $\Sigma$
coincides with some Cauchy surface outside a compact set.

A bilinear form defines a linear map from a vector space to its
algebraic dual. A similar statement holds for a continuous bilinear form
and the topological dual space. However, our formal cotangent spaces
$T^*_\phi\S$ and $T^*_\phi\bar{\S}$ are neither the algebraic nor the
topological duals of the formal tangent spaces $T_\phi\S$ and
$T_\phi\bar{\S}$. Thus we have to check this property for $\Omega$ by
hand. This will be accomplished using one of the splitting maps for the
Jacobi system from Thm.~\ref{thm:Jexsplit}, which is analogous
Lem.~\ref{lem:exsplit} for hyperbolic systems.  The argument in the
proof was inspired by Sec.~3.3 of Ref.~\citen{fr-pois} and Lem.~3.2.1 of
Ref.~\citen{wald-qft}.
\begin{lemma}\label{lem:bOmega}
Provided the constraints are globally parametrizable
(Sec.~\ref{sec:constr}) and the gauge transformations are globally
recognizable (Sec.~\ref{sec:gauge}), the presymplectic form $\Omega$
defines the following map from the formal tangent space to the formal
cotangent space:
\begin{align}
	\Omega\colon & T_\phi\S \to T^*_\phi\S \\
		& \psi \mapsto [\tilde{\alpha}^*] , \\
	\text{with} ~~
	\tilde{\alpha}^* &= \J_\chi[\psi] = \pm\J[\chi_\pm\psi] ,
\end{align}
where $\J\colon \Secs(F) \to \Secs(\tilde{F}^*)$ the Jacobi differential
operator and $\{\chi_\pm\}$ is a partition of unity adapted to a Cauchy
surface $\Sigma$.
\end{lemma}
\begin{proof}
Using the adapted partition of unity, we can write any spacelike
compactly supported solution $\psi$ of $\dot{f}[\psi] = 0$ as $\psi =
\psi_+ + \psi_-$, with $\psi_\pm = \chi_\pm \psi$ now being of retarded
and advanced supports. If $\psi$ also satisfies the constraints
$\dot{c}[f] = 0$, then by Thm.~\ref{thm:Jexsplit} it also satisfies the
Jacobi equation $\J[\psi] = 0$. Hence $\J[\psi_+ + \psi_-] = 0$ or
$\J[\psi_+] = -\J[\psi_-] = \J_\chi[\psi]$. Note that the support of
$\J[\psi_\pm]$ is compact, since $\psi_\pm$ satisfy the Jacobi equation
away from the intersection $S^+\cap S^-\cap \supp \psi$, which is by
hypothesis compact.

Next, we want to find a compactly supported dual density
$\tilde{\alpha}^*$ that satisfies $\Omega(\xi,\psi) = \langle
\xi,\tilde{\alpha}^*\rangle$ for any $\xi\in T_\phi\S$, which in
particular satisfies $\J[\xi]=0$. Recall that an adapted partition of
unity also depends on two additional Cauchy surfaces $\Sigma^\pm\sso
I^\pm(\Sigma)$ and the supports of the partition are contained in
$\supp\chi_\pm \sse S^\pm = I^\pm(\Sigma^\mp)$. The following direct
calculation helps us identify $\tilde{\alpha}^*$.
\begin{align}
	\Omega(\xi,\psi)
	&= \int_\Sigma \omega(\xi,\psi)[\phi]
	= \sum_\pm \int_\Sigma (j^\oo\phi)^*\omega(\xi,\psi_\pm) \\
	&= \sum_\pm\int_{\Sigma^\mp} (j^\oo\phi)^*\omega(\xi,\psi_\pm)
		+ \sum_\pm \pm\int_{S^\pm\cap I^\mp(\Sigma)}\d(j^\oo\phi)^*\omega(\xi,\psi_\pm) \\
	&= \sum_\pm \pm\int_{I^\mp(\Sigma)} (j^\oo\phi)^*(\dh\omega)(\xi,\psi_\pm) \\
	&= \sum_\pm \pm\int_{I^\mp(\Sigma)}
		(j^\oo\phi)^*(-\dv\EL_a\wedge\dv u^a)(\xi,\psi_\pm) \\
	&= \sum_\pm \mp\int_{I^\mp(\Sigma)}
		[(\J^I_{ab}\del_I\xi^b)\psi_\pm^a - (\J^I_{ab}\del_I\psi_\pm^b)\xi^a] \\
	&= \int_{I^-(\Sigma)} \xi\cdot\J[\psi_+]
		- \int_{I^+(\Sigma)} \xi\cdot\J[\psi_-] \\
	&= \int_{I^-(\Sigma)} \xi\cdot \J_\chi[\psi]
		- \int_{I^+(\Sigma)} \xi\cdot(-\J_\chi[\psi]) \\
	&= \int_M \xi\cdot \J_\chi[\psi]
\end{align}
Note that after the integration by parts, the boundary integrals over
$\Sigma^\pm$ were dropped since they did not intersect the support of
their integrands. Then, since $\supp \psi_\pm \sse S^\pm$, the
integration over $S^\pm\cap I^\mp(\Sigma)$ was extended to all of
$I^\mp(\Sigma)$. Finally, the term $\psi_\pm\cdot \J[\xi]$ was dropped
since $\xi$ is a linearized solution.

To complete the proof, we use the non-degeneracy of the natural pairing
between $T_\phi\S$ and $T^*_\phi\S$ (Lem.~\ref{lem:ginv-sols-nondegen},
which we can invoke because of the global parametrizability and
recognizability hypotheses) to define the operator $\Omega$ by the
formula
\begin{equation}
	\langle \xi, \Omega \psi \rangle
	= \Omega(\xi,\psi)
	= \langle \xi, \tilde{\alpha}^* \rangle
	= \langle \xi, [\tilde{\alpha}^*] \rangle ,
\end{equation}
so that $\Omega\psi = [\tilde{\alpha}^*] \in T^*_\phi\S$, with
$\tilde{\alpha}^* = \J_\chi[\psi]$. %\qed
\end{proof}

\begin{corollary}\label{cor:bOmega}
Provided the constraints are globally parametrizable
(Sec.~\ref{sec:constr}) and the gauge transformations are globally
recognizable (Sec.~\ref{sec:gauge}), the $2$-form $\Omega$ on $T_\phi\S$
projects to a $2$-form $\bar\Omega$ on $T_\phi\bar\S$ and hence defines
a map
\begin{align}
	\bar{\Omega}\colon & T_\phi\bar\S \to T^*_\phi\bar\S \\
		& [\psi] \mapsto [\tilde{\alpha}^*] , \\
	\text{with} ~~
	\tilde{\alpha}^* &= \J_\chi[\psi] .
\end{align}
Moreover, this map is independent of the choice of adapted partition of
unity $\{\chi_\pm\}$.
\end{corollary}
\begin{proof}
Notice that in the presence of gauge symmetries (residual gauge freedom
is present after a purely hyperbolic gauge fixing) the form $\Omega$ is
degenerate, since every pure gauge solution lies in its kernel:
\begin{equation}
	\Omega(\dot{g}[\eps],\psi)
	= \langle \dot{g}[\eps], \J_\chi[\chi_\pm\psi] \rangle
	= \pm\langle \eps, \dot{g}^*\circ\J[\chi_\pm\psi] \rangle
	= 0
\end{equation}
for any $\psi$, since Noether's second theorem implies\cite{lw} that
$\dot{g}^*\circ \J = 0$. So, the first part is established.

For the second part, recall that we are not interested in the dual
density $\tilde{\alpha}^* = \J_\chi[\psi]$ specifically, which
explicitly depends on the adapted partition of unity, but rather the
equivalence class $[\tilde{\alpha}^*]\in T^*_\phi\S$, which is defined
modulo $\im \dot{f}^*$ and $\im \dot{c}^*$. Equivalently, following a
conclusion of Thm.~\ref{thm:Jexsplit}, since we actually want
$[\tilde{\alpha}^*] \in T^*_\phi\bar{\S} \cong T^*_\phi\S / \im
\dot{g}$, it is enough to consider $\tilde{\alpha}^*$ modulo $\im\J$.
Consider another adapted partition of unity $\{\chi'_\pm\}$. Because
each partition of unity provides a splitting map
(Thm.~\ref{thm:Jexsplit}), if we consider equivalence classes of
solutions modulo $\im\dot{g}$, we have $[\psi] = [\EE\circ\J_\chi[\psi]]
= [\EE\circ \J_{\chi'}[\psi]]$. Then,
$[\EE[\J_{\chi'}[\psi]-\J_\chi[\psi]]] = [\psi] - [\psi] = [0]$. So, by
exactness of the sequence in Thm.~\ref{thm:Jexsplit}, $\J_\chi[\psi]$
and $\J_{\chi'}[\psi]$ must differ by an element of $\im \J$; in other
words, they represent the same equivalence class in $T^*_\phi\bar{\S}$.

Finally, the projected map $\bar{\Omega}\colon T_\phi \bar{S} \to
T^*_\phi \bar{S}$ is defined by the formula
\begin{equation}\label{eq:symp-rep}
	\langle [\xi] , \bar{\Omega}[\psi] \rangle
	= \langle \xi , \tilde{\alpha}^* \rangle
	= \langle [\xi], [\tilde{\alpha}^*] \rangle ,
\end{equation}
with $\tilde{\alpha}^* = \J_\chi[\psi]$, which is sufficient because the
natural pairing between $T_\phi\bar{\S}$ and $T^*_\phi\bar{S}$ is
non-degenerate (Lem.~\ref{lem:ginv-sols-nondegen}, again which we can
invoke by the global parametrizability and recognizability hypotheses).
\end{proof}

So, formally, the quotient projection to the physical phase space
effects a presymplectic reduction $(\S_H(F),\Omega) \to
(\bar{S}_H(F),\bar\Omega)$. We shall see later on that $\bar\Omega$ is
non-degenerate and hence symplectic.

\begin{remark}\label{rem:gl-par-rec}
Note that the relation between $\bar{\Omega}$ as a bilinear form on the
formal tangent space $T_\phi\bar{\S}$ and the linear map $\J_\chi\colon
T_\phi\bar{\S} \to T_\phi^*\bar{\S}$ relies on the non-degeneracy of the
natural pairing between the formal tangent and cotangent spaces. This
non-degeneracy, as proven in Sec.~\ref{sec:tt*-sols}, relies on the
cohomological conditions that we call \emph{global parametrizability}
and \emph{global recognizability} (Secs.~\ref{sec:constr}
and~\ref{sec:gauge}). It is clear that, if these conditions fail and
hence the natural pairing is degenerate, the form
$\bar{\Omega}(\psi,\xi) = \langle \psi, \J_\chi[\xi]$ may be degenerate,
even if the operator $\J_\chi$ is not. This is bound to happen, because,
as one of the conclusions of Thm.~\ref{thm:Jexsplit}, $\J_\chi$ is
invertible under the weaker hypotheses of \emph{local parametrizability}
and \emph{local recognizability}. Such a degeneracy has already been
noted, for instance, in Refs.~\citen{hs-gauge} and~\citen{sdh}.
\end{remark}

\subsubsection{Formal Poisson bivector, Peierls formula}\label{sec:formal-pois}
Below, we construct a formal \emph{Poisson bivector} $\Pi$, using the
Peierls formula
\begin{equation}
	\Pi = \EE ,
\end{equation}
where $\EE$ is again the causal Green function of the Jacobi operator
$\J$ as defined in Sec.~\ref{sec:constr-gf}. To show that $\Pi$ is
indeed a Poisson bivector, it suffices to show that (a) it is an
antisymmetric bilinear form on the formal cotangent space, (b) it
defines a map from the formal cotangent space to the formal tangent
space and (c) it is a two-sided inverse of $\bar{\Omega}$ defined in
Cor.~\ref{cor:bOmega}. We actually postpone part (c) to
Sec.~\ref{sec:symp-pois}. The fact that $\Pi$ defines a Poisson bracket,
with its Leibniz and Jacobi identities, then formally follows from
standard arguments.

\begin{lemma}\label{lem:pure-hyp-peierls}
The Peierls formula specifies a map from the formal cotangent space to
the formal tangent space:
\begin{align}
	\Pi\colon & T^*_\phi\bar{\S}\to T_\phi\bar{\S} \\
		& [\tilde{\alpha}^*] \mapsto [\psi] , \quad
		\text{with}\quad \dot{g}^*[\tilde{\alpha}^*] = 0 \\
\text{and}\quad
	\psi = {}& \EE[\tilde{\alpha}^*] .
\end{align}
\end{lemma}
\begin{proof}
The challenge is to show that $\Pi$ maps equivalence classes to
equivalence classes (Def.~\ref{def:tt*-sols-gauge}). That is, that any representative
$\tilde{\alpha}^* + \dot{f}^*[\xi] + \dot{c}^*[\tilde{\gamma}^*]$ of an
equivalence class $[\tilde{\alpha}^*]\in T^*_\phi\bar\S$, with
$\dot{g}^*[\tilde{\alpha}^*] = 0$, gets mapped to the same equivalence
class in $T_\phi\bar\S$. By linearity, it suffices to check that $[0]\in
T^*_\phi\bar\S$ is mapped to $[0]\in T_\phi\bar\S$. Recall that any
solution representing $[0]\in T_\phi\bar\S$ is pure gauge
$\dot{g}[\eps]$. Note that the equivalence~\eqref{eq:Jgf-equiv} of the
$(\dot{f}\oplus \dot{c},\tilde{F}^*\oplus E)$ and $(\J\oplus \dot{c}_g,
\tilde{F}^*\oplus E_g)$ equation forms, together with the
self-adjointness of the Jacobi operator $\J^* = \J$, allows us to
rewrite any representative of $[0]\in T^*_\phi\bar\S$ as $\J[\xi] +
\dot{c}_g^*[\tilde{\gamma}^*]$, for some $\xi\in \Secs_0(F)$ and
$\tilde{\gamma}^*\in \Secs_0(\tilde{E}^*_g)$.  This representative will
also satisfy the identity
\begin{equation}\label{eq:g*cg*}
	\dot{g}^*\circ \dot{c}^*_g[\tilde{\gamma}^*]
	= \dot{g}^*[\J[\xi] + \dot{c}^*_g[\tilde{\gamma}^*]]
	= 0 .
\end{equation}
Direct calculation then shows that
\begin{align}
	\Pi[\J[\xi] + \dot{c}^*_g[\tilde{\gamma}^*]]
	&= \EE \circ \J[\xi] + \EE\circ \dot{c}^*_g[\tilde{\gamma}^*] \\
	&= \dot{g}[\eps] - (\dot{c}_g\circ \EE)^*[\tilde{\gamma}^*]
		- \dot{g}\circ p_\J^*\circ \EE\circ \dot{c}^*_g[\tilde{\gamma}^*] \\
\notag & \qquad {}
		- (\EE)^*\circ p_\J[\dot{g}^*\circ \dot{c}^*_g[\tilde{\gamma}^*]] \\
	&= \dot{g}[\eps
		- q_\J^* \circ (\H)^* \circ r_g^*[\tilde{\gamma}^*]
		- p_\J^*\circ \EE\circ \dot{c}^*_g[\tilde{\gamma}^*]]
\end{align}
is pure gauge. We have used the identity that $\EE\circ \J[\xi] =
\dot{g}[\eps]$ for some $\eps\in \Secs_0(P)$ (Thm.~\ref{thm:Jexsplit}),
the anti-self-adjointness identity~\eqref{eq:E-anti},
that (Eqs.~\eqref{eq:Jgf-equiv} and~\eqref{eq:consist1})
\begin{align}
	\dot{c}_g\circ \EE
	&= r_g \circ (\dot{c} \circ \G) \circ \bar{r}
	= r_g \circ (\H \circ \dot{q}\circ \bar{r}) \\
	&= (r_g \circ \H \circ q_\J) \circ \dot{g}^*
\end{align}
and the identity~\eqref{eq:g*cg*}.

Therefore, we can conclude that if $[\tilde{\alpha}^*] = [0]$, then
$[\EE[\tilde{\alpha}^*]] = [0]$. %\qed
\end{proof}

\begin{lemma}\label{lem:Pi-form}
The Peierls formula defines an antisymmetric bilinear form on the formal
cotangent space:
\begin{equation}
	\Pi([\tilde{\alpha}^*],[\tilde{\beta}^*])
	= \langle \Pi[\tilde{\alpha}^*] , [\tilde{\beta}^*] \rangle
	= - \Pi([\tilde{\beta}^*],[\tilde{\alpha}^*]) ,
\end{equation}
for any $[\tilde{\alpha}^*],[\tilde{\beta}^*] \in T^*_\phi\bar{\S}$.
\end{lemma}
\begin{proof}
Recall that the representatives always satisfy
$\dot{g}^*[\tilde{\alpha}^*] = \dot{g}^*[\tilde{\beta}^*] = 0$.
Appealing directly to the anti-self-adjointness
identity~\eqref{eq:E-anti} we have
\begin{align}
	\Pi([\tilde{\alpha}^*],[\tilde{\beta}^*])
	&= \langle \Pi[\tilde{\alpha}^*], [\tilde{\beta}^*] \rangle
	= \langle \EE[\tilde{\alpha}^*] , \tilde{\beta}^* \rangle
	= \langle (\EE)^*[\tilde{\beta}^*] , \tilde{\alpha}^* \rangle \\
	&= -\langle (\EE[\tilde{\beta}^*]
			+ \dot{g}\circ p_\J^*\circ \EE[\tilde{\beta}^*]
			+ (\EE)^*\circ p_\J\circ \dot{g}^*[\tilde{\beta}^*])
			, \tilde{\alpha}^* \rangle \\
	&= -\langle [\EE[\tilde{\beta}^*]] , [\tilde{\alpha}^*] \rangle
	= - \langle \Pi[\tilde{\beta}^*] , [\tilde{\alpha}^*] \rangle \\
	&= -\Pi([\tilde{\beta}^*],[\tilde{\alpha}^*]) . \quad %\qed
\end{align}
\end{proof}

\subsubsection{The Peierls formula inverts the covariant symplectic form}
\label{sec:equiv}
Below, in Thm.~\ref{thm:peierls-inv}, we state and prove the main result
of this section, that $\Pi = \bar{\Omega}^{-1}$. It is worth pausing
here and recalling the various hypotheses, assumptions, and intermediate
results that have lead up to it.

First of all, the result itself is not completely new. On the one hand,
Peierls' original paper\cite{peierls} already outlined an argument for
the equivalence of his proposed bracket and the standard Poisson bracket
of the Hamiltonian formalism, defined with respect to a preferred time
foliation. On the other hand, when the covariant phase space formalism
was introduced, the use of the symplectic current
density\cite{lw,cw,abr,zuckerman} was justified by its agreement
with the standard symplectic structure of the Hamiltonian formalism.
These two observations were joined into a detailed argument by Barnich,
Henneaux and Schomblond,\cite{bhs} which covered the case when the
Hamiltonian formalism includes first class (gauge) and second class
constraints.

Note that both the covariant phase space and Peierls bracket formalisms
are fully covariant, but their equivalence had only been demonstrated
using a non-covariant Hamiltonian formalism as an intermediate step. So,
one reason to look for improvements is the desire to make the argument
covariant throughout and bypass the Hamiltonian formalism all together.
Another reason is to make clear all mathematical assumptions necessary
to make the intermediate constructions well defined. In particular, the
existence of advanced and retarded Green functions, needed by the
Peierls formula, is guaranteed by standard mathematical results in PDE
theory only if the field theory equations of motion (the Euler-Lagrange
equations) satisfy some local and global hyperbolicity%
	\footnote{In the spirit of being inclusive, we have equated our basic
	notion of hyperbolicity precisely with the existence of retarded and
	advanced Green functions (Green hyperbolicity). However, as pointed
	out earlier, there are large classes of PDEs easily identifiable by
	their principal symbols (including wave-like and symmetric hyperbolic
	systems) that are well known to be Green hyperbolic.} %
requirements. We have exhibited these assumptions bundled within the
notions of \emph{hyperbolizability} (Def.~\ref{def:hyper}), a global
causal condition generalizing \emph{global hyperbolicity}
(\ref{sec:conal}), as well as \emph{global parametrizability} and
\emph{recognizability} (Secs.~\ref{sec:constr} and~\ref{sec:gauge}).
Also, the formal presymplectic form (Def.~\ref{def:formal-symp}) is
defined only when the integral over the presymplectic current converges.
Again, a sufficient condition for this integral to converge is to
restrict the support of linearized solutions plugged into the
presymplectic form to be spacelike compact. This restriction is the main
reason for defining the formal tangent spaces to consist of field
sections of spacelike compact support (Secs.~\ref{sec:tt*-conf},
\ref{sec:tt*-sols}).

The main technical result leading up to the theorem below is of course
Thm.~\ref{thm:Jexsplit}, which reduces to Prp.~\ref{prp:exact} when the
Euler-Lagrange equations are directly in hyperbolic form (gauge
invariance and constraints are absent). The exactness of parts of the
horizontal sequence~\eqref{eq:JE-seq} (after taking the vertical
cohomologies) can be seen as a precise characterization of the kernel
and cokernel of the causal Green function $\EE$, defined in
Eq.~\eqref{eq:E-def}.  It is this characterization that is the main
motivation behind defining the formal cotangent spaces to consist of
dual densities of compact support (Secs.~\ref{sec:tt*-conf},
\ref{sec:tt*-sols}). If these support restrictions were relaxed, for
instance, to timelike compact support for dual densities, then the
causal Green function $\EE$ need not be invertible due to global
Aharonov-Bohm type effects.\cite{sdh} In that case, the
relation of the Peierls formula to the presymplectic form must be more
subtle. The final technical result that is used in the proof below is
the non-degeneracy of the natural pairing between the formal tangent and
cotangent spaces (Sec.~\ref{sec:tt*-sols}), which rely on rather
technical sufficient conditions that we have dubbed global
parametrizability of constraints (Sec.~\ref{sec:constr}) and global
recognizability of gauge transformations (Sec.~\ref{sec:gauge}).  If
they fail, then the symplectic form $\bar{\Omega}$ is in fact
degenerate. However, as can be seen from its representation in
Cor.~\ref{cor:bOmega}, that is not because $\J_\chi\colon T_\phi\bar{\S}
\to T^*_\phi\bar{\S}$ becomes non-invertible, but because the natural
pairing $\langle -,- \rangle$ becomes degenerate. Recall that, according
to Thm.~\ref{thm:Jexsplit}, $\J_\chi$ remains invertible even when only
local parametrizability and recognizability hold.

The proof given below was inspired by Sec.~3.3 of Ref.~\citen{fr-pois}
as well as the exact sequence of Prp.~\ref{prp:exact} (see references
near its statement), though similar ideas can already be found in
Lem.~3.2.1 of Ref.~\citen{wald-qft}. The main limitation of the argument
given by Forger~\& Romero is that it only treats the case when
Euler-Lagrange equations are already in hyperbolic form. Our argument is
generalized to the case where a hyperbolization may be required, and
constraints and gauge may be present. Due to the more complicated
hypothesis, the argument itself has been fine grained and split into
multiple steps. Also, we show that $\bar{\Omega}$ and $\Pi$ are
two-sided (as opposed to one-sided) inverses of each other.  The
technical content of the proof of the main Thm.~3 of
Ref.~\citen{fr-pois} is split between our Lem.~\ref{lem:bOmega} and
Cor.~\ref{cor:bOmega} (rewriting the formal symplectic form),
Thm.~\ref{thm:Jexsplit} (two-sided inversion), and
Lem.~\ref{lem:ginv-sols-nondegen} (natural pairing non-degeneracy).

\begin{theorem}\label{thm:peierls-inv}
Global parametrizability (Sec.~\ref{sec:constr}) and global
recognizability (Sec.~\ref{sec:gauge}) conditions hold, the Peierls
formula gives a two-sided inverse to the formal symplectic form,
$\bar{\Omega}\Pi = \id$ on $T^*_\phi\bar{\S}$ and $\Pi\bar{\Omega} =
\id$ on $T_\phi\bar{\S}$.
\end{theorem}
\begin{proof}
The proof uses in an essential way the splitting identities of
Thm.~\ref{thm:Jexsplit}. Consider any $[\psi]\in T_\phi\bar\S$ and
$[\tilde{\alpha}^*]\in T^*_\phi\bar\S$. To use these splitting
identities, we introduce a Cauchy surface $\Sigma \sso M$ and a
partition of unity $\{\chi_\pm\}$ adapted to it. Then
\begin{align}
	\langle \Pi \bar{\Omega} [\psi], [\tilde{\alpha}^*] \rangle
	&= \langle \EE\circ \J_\chi[\psi], \tilde{\alpha}^* \rangle
			\quad \text{(using Eq.~\eqref{eq:Jsplit})} \\
	&= \langle \psi + \dot{g}[\eps] , \tilde{\alpha}^* \rangle
			\quad \text{(for some $\eps\in \Secs_{SC}(P)$)} \\
	&= \langle [\psi] , [\tilde{\alpha}^*] \rangle .
\end{align}
Therefore, from the non-degeneracy of the natural pairing between
$T_\phi\bar\S$ and $T^*_\phi\bar\S$
(Lem.~\ref{lem:ginv-sols-nondegen}), we concluded that
$\Pi\bar{\Omega} = \id$. Similarly, we have
\begin{equation}
	\langle [\psi], \bar{\Omega} \Pi [\tilde{\alpha}^*] \rangle
	= \langle \psi , \J_\chi \circ \EE [\tilde{\alpha}^*] \rangle .
\end{equation}
But then
\begin{equation}
	\EE[\J_\chi \circ \EE[\tilde{\alpha}^*] - \tilde{\alpha}^*]
	= (\EE \circ \J_\chi)\circ \EE[\tilde{\alpha}^*] - \EE[\tilde{\alpha}^*]
	= \dot{g}[\eps] ,
\end{equation}
for some $\eps\in \Secs_{SC}(P)$. But, by the exactness (after taking
vertical cohomologies) of the horizontal sequence~\eqref{eq:JE-seq} of
Thm.~\ref{thm:Jexsplit} at $\Secs_0(F)$, this means that $\J_\chi\circ
\EE[\tilde{\alpha}^*] - \tilde{\alpha}^* = \J[\xi]$ for some $\xi\in
\Secs_0(F)$. In other words,
\begin{equation}
	\langle [\psi] , \bar{\Omega} \Pi [\tilde{\alpha}^*] \rangle
	= \langle \psi, \tilde{\alpha}^* + \J[\xi] \rangle
	= \langle [\psi], [\tilde{\alpha}^*] \rangle .
\end{equation}
Therefore, from the non-degeneracy of the natural pairing between
$T_\phi\bar\S$ and $T^*_\phi\bar\S$, we concluded that $\bar{\Omega}\Pi
= \id$. %\qed
\end{proof}

It is interesting to note that the construction of the Poisson bivector
$\Pi$ via the Peierls formula requires gauge fixing the equations of
motion. On the other hand, the construction of the symplectic form
$\bar{\Omega}$ does not. Since, after gauge reduction, the two are
mutual inverses, the Poisson bivector on the gauge invariant solutions
space ultimately does not depend on gauge fixing. There is another way
to see that result. For recognizable gauge transformations, the gauge
invariant field combination $\xi=\dot{g}'[\psi]$ of a solution $\psi\in
\Secs(F)$ of $\dot{f}[\psi] = 0$ itself satisfies the hyperbolic PDE
system $\dot{k}'[\xi] = \dot{s}'\circ \dot{f}[\psi] = 0$. On the other
hand, the equivalence formulas~\eqref{eq:Jgf-equiv} and the gauge fixing
compatibility condition~\eqref{eq:gf} imply that the same is true even
if only $\J[\psi]=0$. In other words, the system $\dot{k}'[\xi] = 0$
depends on $\J$ and $\dot{g'}$ but not on the choice of gauge fixing.
This is the case, for example, for Maxwell electrodynamics, where
Maxwell's equations for the gauge invariant field strength $F=F[A]$,
where $A$ is the gauge variant vector potential, by themselves
constitute a (constrained) hyperbolic system. It is then not surprising
that we can express the Poisson bivector acting on gauge invariant observables formed
with respect to the gauge invariant field combinations directly in terms of
the causal Green for the $\dot{k}'$ PDE system. This observation was
known already to Peierls and this example of Maxwell electrodynamics
appeared in his original paper.\cite{peierls} Of course, if appropriate
cohomologies in diagram~\eqref{eq:glrec*} do not vanish, there may be
gauge invariant observables not of that form, for which the gauge fixed
Peierls Green function $\EE$ would be necessary.
\begin{corollary}\label{cor:ginv}
Given two gauge invariant dual densities of the form $\tilde{\alpha}^* =
\dot{g}^{\prime*}[\tilde{\alpha}^{\prime*}]$ and $\tilde{\beta}^* =
\dot{g}^{\prime*}[\tilde{\beta}^{\prime*}]$, with $\tilde{\alpha}^*,
\tilde{\beta}^*\in \Secs_0(\tilde{F}^*)$ and $\tilde{\alpha}^{\prime*},
\tilde{\beta}^{\prime*} \in \Secs_0(\tilde{P}^{\prime*})$, we have the
following identity
\begin{equation}
 \langle \EE[\tilde{\alpha}^*] , \tilde{\beta}^* \rangle
 = \langle \K' \circ \bar{r}'[\tilde{\alpha}^{\prime*}] ,
		\tilde{\beta}^{\prime*} \rangle ,
\end{equation}
where $\bar{r}' = \dot{s}'\circ\bar{r}\circ\dot{g}^{\prime*}$.
\end{corollary}
\begin{proof}
Direct calculation shows
\begin{align}
	\langle \EE[\tilde{\alpha}^*] , \tilde{\beta}^* \rangle
	&= \langle \G\circ \bar{r}\circ\dot{g}^{\prime*}[\tilde{\alpha}^{\prime*}] ,
		\dot{g}^{\prime*}[\tilde{\beta}^*] \rangle \\
	&= \langle (\dot{g}'\circ \G)\circ \bar{r}\circ\dot{g}^{\prime*}
			[\tilde{\alpha}^{\prime*}] , \tilde{\beta}^{\prime*} \rangle \\
	&= \langle \K'\circ(\dot{s}'\circ\bar{r}\circ\dot{g}^{\prime*})
			[\tilde{\alpha}^{\prime*}] , \tilde{\beta}^{\prime*} \rangle ,
\end{align}
which concludes the proof.
\end{proof}

We conclude with a simple corollary that is sometimes known as
\emph{classical microcausality}.
\begin{corollary}\label{cor:microcaus}
Consider two on-shell gauge invariant dual density classes
$[\tilde{\alpha}^*], [\tilde{\beta}^*] \in T^*_\phi\bar{\S}$ whose
supports are spacelike separated,%
	\footnote{There exist representatives $\tilde{\alpha}^*,
		\tilde{\beta}^* \in \Secs_0(\tilde{F}^*)$ whose supports are
		genuinely spacelike separated, that is they satisfy $\supp
		\tilde{\alpha}^* \cap \overline{I(\supp \tilde{\beta}^*)} =
		\varnothing$ and $\supp \tilde{\beta}^* \cap \overline{I(\supp
		\tilde{\alpha}^*)} = \varnothing$.}
then
\begin{equation}
	\Pi([\tilde{\alpha}^*],[\tilde{\beta}^*]) = 0.
\end{equation}
\end{corollary}
\begin{proof}
Picking representatives $\tilde{\alpha}^*, \tilde{\beta}^* \in
\Secs_0(\tilde{F}^*)$ with genuinely spacelike separated supports and
using the definition of the Poisson bivector, we have
$\Pi([\tilde{\alpha}^*],[\tilde{\beta}^*]) = \langle \tilde{\alpha}^*,
\G[\bar{r}[\tilde{\beta}^*]] \rangle = 0$. This is obvious because
$\supp \G[\bar{r}[\tilde{\beta}^*]] \sse I(\supp{\tilde{\beta}^*})$, by
the properties of the causal Green function $\G$. Hence, the arguments
in the natural pairing have non-overlapping supports and give zero.
%\qed
\end{proof}

\section{Examples}\label{sec:examples}
In this section, we give a few examples of common relativistic field
theories and show how they fit into the framework presented in this
review. In particular, we make explicit the various identities needed to
show that they are \emph{hyperbolizable} according to
Def.~\ref{def:hyper}. We freely use the notation introduced in
Sec.~\ref{sec:constr-gf}. In all the examples, we will concentrate on
linear theories or the linearizations of non-linear ones, as discussed
in Sec.~\ref{sec:linear}.

\subsection{Scalar field}
The field bundle $F = M \times \mathbb{R}$ is the trivial $\R$-bundle.
The Lagrangian density is $\L[\phi] = -\frac{1}{2}\sqrt{-|g|} [g^{ij}
(\del_i \phi) (\del_j \phi) + V(\phi)]\, \d\tilde{x}$, where we used
coordinates $(x^i)$ on $M$, $g$ is a globally hyperbolic Lorentzian
metric on $M$, and $|g| = \det g_{ij}$. The Jacobi equations
\begin{equation}
	\J[\psi]
	= \left(\del_i (\sqrt{-|g|} g^{ij} \del_j \psi)
		- 2\sqrt{-|g|}V'(\phi)\psi\right)\, \d\tilde{x}
\end{equation}
have a wave-like principal symbol and, given the global hyperbolicity of
the metric, are well known to be Green
hyperbolic,\cite{bgp,waldmann-pde} so $\dot{f} = \J$. The constraints
and the gauge transformations are trivial, $\dot{c} = 0$ and $\dot{g} =
0$. Note that, for the existence of Green functions, no constraints need
to be imposed on $V'(\phi)$ beyond smoothness, so tachyonic theories and
theories with variable mass are hyperbolizable as well.

A more detailed treatment can be found for instance in
Ref.~\citen{fr-pois}.

\subsection{Maxwell $p$-form}
The field bundle $F = \Lambda^p M$ is the bundle of differential
$p$-forms ($p>0$). The Lagrangian is the generalization of the Maxwell
Lagrangian density $\L[\phi] = -\frac{1}{4} \d\phi \wedge {*}\d\phi$,
where $*$ is the Hodge star with respect to a globally hyperbolic metric
$g$ on $M$. Below, we identify the densitized dual bundle of the bundle
of$p$-forms with $(n-p)$-forms via the pairing formula $\psi\cdot
\tilde{\alpha}^* = \psi \wedge \tilde{\alpha}^*$. The Jacobi equations
are
\begin{equation}
	\J[\psi] = \frac{1}{2} {*}\delta\d\psi ,
\end{equation}
where $\delta = {*}\d{*}$ is the de~Rham co-differential. They are
invariant under gauge transformations with generator $\dot{g}[\eps] =
\d\eps$, where the gauge parameter bundle is $P = \Lambda^{p-1} M$. The
Lorenz gauge plays the role of a purely hyperbolic gauge fixing,
$c_g[\psi] = \delta \psi$. The equivalent constrained hyperbolic system
is $\dot{f}[\psi] = {*}\square\psi$, $\dot{c} =
\delta\psi$; the equivalence is effected by the operators $\bar{r} =
2\,\id$, $\bar{r}_c = {*}\d$, $\bar{r}_\J = 0$ and $\bar{r}_g =
\id$. The operator $\square = (\delta\d + \d\delta)$ is the Laplace-Beltrami
operator (also known as the $p$-form d'Alambertian) and is well known,
again when the metric $g$ is globally hyperbolic, to be Green
hyperbolic.\cite{bgp,waldmann-pde} The parametrizability
(diagram~\eqref{eq:glpar}) and recognizability
(diagram~\eqref{eq:glrec}) identities are generated by the following
commutative diagrams (which hold possibly up to sign factors):
\begin{gather}
\vcenter{\xymatrix{
	\Secs(\Lambda^{p+1}M) \ar[d]^{{*}\square} \ar[r]^{\delta}
		& \Secs(\Lambda^p M) \ar[d]^{{*}\square} \ar[r]^{\dot{c}=\delta}
		& \Secs(\Lambda^{p-1} M) \ar[d]^{{*}\square} \\
	\Secs(\Lambda^{n-p-1}M) \ar[r]^{\d}
		& \Secs(\Lambda^{n-p}M) \ar[r]^{\d}
		& \Secs(\Lambda^{n-p+1}M)
}} , \\
\vcenter{\xymatrix{
	\Secs(\Lambda^{p-1}M) \ar[d]^{{*}\square} \ar[r]^{\dot{g}=\d}
		& \Secs(\Lambda^p M) \ar[d]^{{*}\square} \ar[r]^{\d}
		& \Secs(\Lambda^{p-1} M) \ar[d]^{{*}\square} \\
	\Secs(\Lambda^{n-p+1}M) \ar[r]^{\delta}
		& \Secs(\Lambda^{n-p}M) \ar[r]^{\delta}
		& \Secs(\Lambda^{n-p-1}M)
}} .
\end{gather}
The gauge fixing compatibility condition~\eqref{eq:gf} is obviously
satisfied as $\delta \circ {*}\d \circ \delta = {*}\d^2\delta = 0$, up
to sign.

A more detailed treatment can be found for instance in Refs.~\citen{sdh}
and~\citen{hs-gauge}.

\subsection{Proca field}
The field bundle $F = T^*M$ is the bundle of $1$-forms and the
Lagrangian density differs from the Maxwell one by a mass term,
$\L[\phi] = -\frac{1}{4} \d\phi\wedge {*}\d\phi - \frac{1}{2}m^2\phi
\wedge {*}\phi$, where $*$ is the Hodge star with respect to a globally
hyperbolic metric $g$ on $M$. The Jacobi equations are
\begin{equation}
	\J[\psi] = \frac{1}{2}{*}(\delta\d\psi - m^2 \psi) ,
\end{equation}
where $\delta = {*}\d{*}$ is the de~Rham co-differential. There is no
gauge invariance, $\dot{g} = 0$, but there are integrability conditions.
The equivalent constrained hyperbolic system is $\dot{f}[\psi] =
{*}(\square - m^2)\psi$, $\dot{c}[\psi] = \delta\psi$; the equivalence
is effected by the operators $\bar{r} = 2\,\id$, $\bar{r}_c =
-\frac{2}{m^2}\delta\d$, $\bar{r}_\J = -\frac{2}{m^2}{*}\d$ and
$\bar{r}_g = 0$. Again, $\square = (\delta\d+\d\delta)$ is the
Laplace-Beltrami operator, which is known to be Green
hyperbolic\cite{bgp,waldmann-pde} when the metric $g$ is globally
hyperbolic. The parametrizability diagram~\eqref{eq:glpar} identities
are generated by the following commutative diagram (which holds possibly
up to sign factors):
\begin{equation}
\vcenter{\xymatrix{
	\Secs(\Lambda^{p+1}M) \ar[d]^{{*}\square} \ar[r]^{\delta}
		& \Secs(\Lambda^p M) \ar[d]^{{*}\square} \ar[r]^{\dot{c}=\delta}
		& \Secs(\Lambda^{p-1} M) \ar[d]^{{*}\square} \\
	\Secs(\Lambda^{n-p-1}M) \ar[r]^{\d}
		& \Secs(\Lambda^{n-p}M) \ar[r]^{\d}
		& \Secs(\Lambda^{n-p+1}M)
}} .
\end{equation}

A more detailed treatment can be found for instance in
Ref.~\citen{dappiaggi}.

\subsection{Graviton}
The field bundle $F = S^2T^*M$ is the bundle of symmetric, rank-2
covariant tensors and the Lagrangian density is the usual
Einstein-Hilbert action $\L[\phi] = {*}_\phi (R[\phi]-2\Lambda)$, where
we of course interpret $\phi$ as a Lorentzian metric, ${*}_\phi$ is the
corresponding Hodge star operator, $R[\phi]$ the corresponding Ricci
scalar and $\Lambda$ the cosmological constant.\cite{wald-gr} We will of
course use $\nabla$ to denote the covariant derivative compatible with
$\phi$. Consider a background solution $\phi\in \S_H(F)$, whose good
causal behavior property we take to coincide with the usual notion of
Lorentzian global hyperbolicity. Let us denote the corresponding volume
form as ${*} 1 = {*}_\phi1$. Here it convenient to identify any tensor
bundle with its densitized dual, with the natural pairing $\langle -,-
\rangle$ constructed by contracting corresponding indices using the
metric $\phi$ or its inverse, multiplying by the volume form ${*}1$ and
integrating over $M$.

The Jacobi equations are $\J[\psi] = L[\psi]$,
where $L$ also known as the \emph{Lichnerowicz}
operator. In local coordinates $(x^i)$ on
$M$, the components of the Lichnerowicz operator are
\begin{equation}
	L_{ij}[\psi] =
		- \frac{1}{2}\phi_{ij}(\nabla^k\nabla^l\psi_{kl}-\square\psi-\Lambda\psi)
		- \square \psi_{ij} - \Lambda \psi_{ij}
		- \frac{1}{2}\nabla_i\nabla_j \psi + \nabla^k\nabla_{(i}\psi_{j)k} ,
\end{equation}
where indices are raised and lowered using $\phi$, $\psi = \phi^{kl}
\psi_{kl}$ and $\square = \nabla^k \nabla_k$ is the tensor
d'Alambertian.

Before proceeding, we introduce some key linear differential operators
and their adjoints. To start, the \emph{trace reversal operator}
$\rho\colon \Secs(S^2T^*M) \to \Secs(S^2T^*M)$ does not actually involve
any derivatives and in components is given by $\rho_{ij}[\psi] =
\psi_{ij} - \frac{1}{2}\phi_{ij}\psi$. With our conventions, it is
self-adjoint, $\rho^* = \rho$, and also idempotent, $\rho\circ \rho =
\id$. Given, a $1$-form $v$, we define $K_{ij}[v] = \nabla_{(i} v_{j)}$
and call $K\colon \Secs(T^*M) \to \Secs(S^2T^*M)$ the \emph{Killing
operator}.  It's adjoint, $K^*_{j}[\psi] = -\nabla^i \psi_{ij}$ (recall
our identification of each tensor bundle with its own densitized dual)
is the divergence operator on symmetric covariant $2$-tensors, $K^*
\colon \Secs(S^2T^*M) \to \Secs(T^*M)$. Another important operator is
the linearized Riemann curvature (cf.~Sec.~7.5 of Ref.~\citen{wald-gr}).
If $\phi' = \phi + \lambda \psi$, the components of the Riemann tensor
of $\phi'$ are given by $R_{ijkl}[\phi'] = R_{ijkl} +
\lambda\dot{R}_{ijkl}[\psi] + O(\lambda^2)$, where $R_{ijkl}$ is the
Riemann tensor of $\phi$ and
\begin{equation}
	\dot{R}_{ijkl}[\psi]
		= -2 \nabla_{[i} \psi_{j][l;k]} + R_{ij[k}{}^m \psi_{l]m},
\end{equation}
with the usual notation $(-)_{;i} = \nabla_i(-)$, is the
\emph{linearized Riemann curvature operator} $\Secs(S^2T^*M) \to
\Secs(RM)$ and $RM \to M$ is the sub-bundle of $(T^*)^{\otimes 4}M$ that
satisfies the algebraic symmetries of the Riemann tensor. It's adjoint
operator $\dot{R}^* \colon \Secs(RM) \to \Secs(S^2T^*M)$ is then
\begin{equation}
	\dot{R}^*_{ij}[\xi]
		= 2\nabla^l\nabla^k \xi_{k(ij)l} - R^{klm}{}_{(i} \xi_{j)mkl} .
\end{equation}
Finally, we define the following self-adjoint hyperbolic differential
operator
\begin{equation}
	W_{ij}[\psi] = \square \psi_{ij} - 2R^{k}{}_{ij}{}^l \psi_{kl} ,
\end{equation}
with $W\colon \Secs(S^2T^*M) \to \Secs(S^2T^*M)$. Note that $W$ has a
wave-like principal symbol so it is known to be Green
hyperbolic.\cite{fewster-hunt,bgp,waldmann-pde}

The Jacobi equations are invariant under gauge transformations
(linearized diffeomorphisms) with generator $\dot{g}[v] = K[v]$, the
Killing operator, where the gauge parameter bundle is $P = T^*M$. The
de~Donder gauge plays the role of a purely hyperbolic gauge fixing
$\dot{c}_g = K^*\circ \rho$, or $(c_g)_j[\psi] = \nabla^i
\rho[\psi]_{ij}$ in coordinate components. The equivalent constrained
hyperbolic system is $\dot{f}[\psi] = W[\psi]$, $\dot{c}[\psi] =
K^*\circ\rho[\psi]$. The equivalence is effected by the operators
$\bar{r} = -2\rho$, $\bar{r}_c = 2K$,
$\bar{r}_\J = 0$ and $\bar{r}_g = \id$. The parametrizability
(diagram~\ref{eq:glpar}) and recognizability (diagram~\ref{eq:glrec})
identities are generated by the following commutative diagrams:
% K' = \dot{R}
\begin{gather}
\vcenter{\xymatrix{
	\Secs(RM) \ar[d]^{W^{\prime*}} \ar[r]^-{\rho\circ K^{\prime*}}
		& \Secs(S^2T^*M) \ar[d]^W \ar[r]^-{K^*\circ \rho}
		& \Secs(T^*M) \ar[d]^{\square + \Lambda} \\
	\Secs(RM) \ar[r]^{\rho\circ K^{\prime*}}
		& \Secs(S^2T^*M) \ar[r]^{K^*\circ\rho}
		& \Secs(T^*M)
}} , \\
\vcenter{\xymatrix{
	\Secs(T^*M) \ar[d]^{\square + \Lambda} \ar[r]^-{K\circ\rho}
		& \Secs(S^2T^*M) \ar[d]^W \ar[r]^-{K'}
		& \Secs(RM) \ar[d]^{W'} \\
	\Secs(T^*M) \ar[r]^{K\circ \rho}
		& \Secs(S^2T^*M) \ar[r]^{K'}
		& \Secs(RM)
}} .
\end{gather}
%
%It is well known that $\dot{R}\circ K = 0$ and that the linearized
%Riemann curvature operator and the Killing operator form an formally exact
%complex. We do not give here the explicit form of the operator $W'$, but
%it must exist for abstract reasons.
We do not give explicit general expressions for the operators $K'$ and
$W'$, simply because they do not seem to be available in the literature.
On the other hand, they must exist for abstract reasons. Namely, if we
define $K'$ as differential operator extending $K$ to a formally exact
sequence, then it always exists, as mentioned in \ref{sec:formal-exact}.
Further, the composition of operators
\begin{equation}
	K' \circ W
	= K' \circ (-2\rho\circ L + 2K\circ K^*\circ \rho)
	= -2 K' \circ\rho\circ L
\end{equation}
clearly annihilates the image of the Killing operator $K$, due to the
gauge invariance of $L$. Therefore, by Lem.~\ref{lem:fec-fact}, there
must exist a factorization $K' \circ W = W'\circ K'$, which we will
conjecture to have a wave-like principal symbol (like $W$ and $\square +
\Lambda$ do) and hence be Green hyperbolic.

If we assume that the background metric tensor $\phi$ to have constant
curvature, we can say much more. In particular we know that $K'=
\dot{R}$ and that the composition $\dot{R}'\circ K = 0$ forms part of a
larger elliptic complex,\cite{calabi} in many ways analogous to the
de~Rham complex. In the even more special case of zero curvature (and
hence also $\Lambda=0$), all operators can be expressed with constant
coefficients in local inertial coordinate systems formal, which makes it
easier to check formal exactness directly. Also, in that case all the
hyperbolic operators become equal to $\square$, the wave operator.

A more detailed treatment of the quantization of the graviton field on
arbitrary cosmological vacuum backgrounds can be found in
Ref.~\citen{fewster-hunt}, though without introducing operators
analogous to $K'$ and $W'$.

\section{Discussion}\label{sec:discuss}
We have reviewed in detail the covariant phase space formalism and the
Peierls formula, which endow the space of solutions of a classical field
theory, respectively, with symplectic and Poisson structures, thus
giving it the structure of a phase space, well known to be equivalent to
the canonical phase space. Each of these constructions is covariant and
does not require the non-covariant, canonical Hamiltonian formalism as
an intermediary. In distinction with much of the existing literature,
where the following aspects have often been left implicit, we have
spelled out precise conditions under which these constructions succeed
without mathematical ambiguities or difficulties. While it has long been
known that the resulting symplectic and Poisson structures are
equivalent (the symplectic form and the Poisson bivector are mutual
inverses), despite the covariant construction, existing proofs still
required the canonical Hamiltonian formalism as an intermediary. The
main result in our presentation, which also happens to be novel, is a
detailed and completely covariant proof of the equivalence under a
precise set of sufficient conditions. The proof follows the ideas of the
previous work of Forger~\& Romero,\cite{fr-pois} but is generalized to
field theories more general than scalar fields. Our argument holds for
theories that also include constraints and that may have gauge
symmetries. The list of examples to which the argument is applicable
includes essentially all relativistic field theories of physical
interest.

Despite the fact that the phase space of a field theory in more than one
spacetime dimension (which corresponds to ordinary mechanical systems)
is infinite dimensional, we have systematically avoided a discussion of
functional analytical details needed in a theory of infinite dimensional
geometry. Instead, we have treated formally the minimal geometric
details needed in our presentation. Essentially,
we have restricted our discussion to linear PDEs (or
rather, linearizations of non-linear ones) and their solution spaces by appealing to the
fact that the inversion of a symplectic form or a Poisson bivector
requires only the tangent or cotangent space at a single point of the
phase space (a background solution). However, the precise algebraic and
differential geometric identities given here can be used as a core in a
future investigation that would fill in the missing functional analytic
details. In fact, some attempts along these lines have already been made
elsewhere. For instance, Ref.~\citen{bfr} has done precisely that but
only for the more restrictive class of scalar field theories. On the
other hand, Refs.~\citen{fr-bv,rejzner-thesis} have considered more
general theories, including those with gauge theories. Incidentally,
these references have concentrated on the so-called \emph{off-shell}
formalism and, while heavily relying on the Peierls formula, did not
consider its relation to the corresponding covariant symplectic
structure, which requires restriction to solutions to be well defined.

The sufficient conditions we have introduced for the Peierls inversion
formula to hold, the \emph{(global) parametrizability} of constraints
and the \emph{(global) recognizability} or gauge transformations, have
two aspects. See Remarks~\ref{rem:loc-par-rec} and~\ref{rem:gl-par-rec}
regarding the subtle interplay between these conditions and the
hypotheses that are sufficient to establish non-degeneracy of symplectic
and Poisson structures described in this review. The local version is
expected to hold generically for relativistic field theories of physical
interest, as illustrated by the examples of Sec.~\ref{sec:examples}. The
global version, on the other hand has a cohomological character and it
is actually known to fail in spacetimes with certain topological
properties.\cite{hs-gauge,sdh} The main examples of these problematic
cases have come from studying Maxwell electrodynamics on spacetimes with
non-trivial spatial topology.\cite{sdh,hs-gauge} It would be nice to
identify more key examples and study their properties. This would
require the computation of cohomologies of the de~Rham and other
formally exact complexes with causally restricted supports (e.g.,\ advanced,
retarded, spacelike compact, timelike compact). The techniques needed
for such computations go a bit beyond the standard treatments of de~Rham
cohomology with unrestricted or compact supports, as presented in
standard differential geometry and differential topology texts. They
will be addressed elsewhere.\cite{kh-cohom}

More generally, compact or spacelike compact supports, featuring in the
sufficient conditions discussed above, may be too restrictive for
physical purposes, for example when dealing with infrared issues on
spatially non-compact spacetimes. In those cases, the solution, of
course, is to introduce boundary conditions at infinity.  However, as is
well known, there may not always be a uniquely preferred set of boundary
conditions. In fact, boundary conditions are expected to be dictated by
detailed physical considerations, which may vary from problem to
problem. The main difficulty in relaxing the spacelike compact support
condition on linearized solutions is the divergence of the integral in
the definition of the covariant symplectic form,
Def.~\ref{def:formal-symp}. This situation is reminiscent of the problem
of extending unbounded, symmetric operators on a Hilbert space to larger
domains, while maintaining their self-adjointness.\cite{grubb} Perhaps a
similar approach can be applied to the symplectic form, where its
anti-symmetry would replace the self-adjointness condition, can be used
to study the space of possible boundary conditions at infinity.
Notably, an attempt in a direction implicitly similar to this suggestion
can be found in Sec.~5.1 of Ref.~\citen{barnich-brandt}. These ideas will be explored
further in future work.

\section*{Acknowledgments}
The author would like to thank Claudio Dappiaggi, Thomas-Paul Hack,
Alexander Schenkel and Urs Schreiber for fruitful discussions, also
B\'eatrice Bonga for feedback on the manuscript, and acknowledges
support from the Netherlands Organisation for Scientific Research (NWO)
(Project No.\ 680.47.413).

\appendix

\section{Jet bundles and the variational bicomplex}\label{sec:jets}
In this appendix, we briefly introduce jet bundles and fix the relevant
notation. For simplicity, we restrict ourselves to fields taking values
in vector bundles. However, the discussion could be straightforwardly
generalized to general smooth bundles. More details, as well as a
coordinate independent definition, can be found in the standard
literature.\cite{olver-lie,kms,spring-convex}

Fix a vector bundle $F\to M$, with $\dim M=n$, with fibers modeled on a
vector space $U$, and consider an adapted coordinate patch $\R^n\times
U$, with coordinates $(x^i,u^a)$.  Extend this patch to a \emph{$k$-jet
patch} $\R^n\times U \times U^{n_k}$ by adding extra copies of $U$, with
new coordinates $(x^i,u^a,u^a_i,u^a_{ij},\ldots,u^a_{i_1\cdots i_k})$,
which formally denote the derivatives of $\del_{i_1i_2\cdots}\phi^a(x)$
of a section $\phi$ at $x$. To keep track of all the derivatives, we
introduce \emph{multi-index} notation. A multi-index $I=i_1i_2\cdots
i_k$ replaces the corresponding set of symmetric covariant coordinate
indices (the multi-index does not change when the defining $i$'s are
permuted).  The \emph{order} of this multi-index is given by $|I|=k$,
with $|\varnothing|=0$. To augment a multi-index by adding another
index, we use the notation $Ij = jI = i_1\cdots i_k j$. Thus we can
write higher order derivatives as $\del_{i_1\cdots i_k} \phi(x) =
\del_I\phi(x)$, the higher order jet coordinates as $u^a_{i_1\cdots i_k}
= u^a_I$ and the total set of coordinates on a $k$-jet patch as
$(x^i,u^a_I)$, $|I|\le k$. In particular the empty multi-index
$I=\varnothing$ corresponds to $u^a_\varnothing = u^a$.

Since the higher derivatives are symmetric in all indices, the number of
extra coordinates is given by $n_k = \sum_{l=1}^k \dim S^k \R^n$, with
$S^k$ denoting the symmetric tensor product. Given two different
coordinate patches on $F$, we define the transition maps between the
corresponding $k$-jet patches according to the usual calculus chain rule
applied to higher order derivatives. These $k$-jet patches can be glued
together into the total space of the \emph{$k$-jet bundle} $J^kF\to M$,
which includes $J^0F \cong F$.

Since $F\to M$ is a vector bundle, so is $J^kF\to M$. It is isomorphic
to $F\oplus_M (F\otimes_M S^1 T^*M) \oplus_M \cdots \oplus_M (F\otimes_M
S^k T^*M)$, but not naturally. Jet bundles come with natural projections
$J^kF\to J^{k-1}F$, which simply discard all derivatives of order $k$.
This projection gives $J^kF$ the structure of an affine bundle over the
base $J^{k-1}F$, with fibers modeled on the vector bundle $(F\otimes_M
S^k T^*M)^{k-1}\to J^{k-1}F$ (see Def.~\ref{def:jet-pb} next). The
bundle $J^kF\to J^{k-1}F$ is affine because, in general, bundle
morphisms of $J^kF\to J^kF$ induced by vector bundle automorphisms of
$F$ are not linear but affine.

Given a vector bundle $E\to M$ it can be pulled back to the $k$-jet
bundle along the projection $J^kF\to M$. We introduce a convenient
notation for this pullback.
\begin{definition}\label{def:jet-pb}
We denote by $(E)^k\to J^kF$ the pullback of $E\to M$ to $J^k F$, which
then fits into the pullback commutative square
\begin{equation}
\vcenter{\xymatrix{
	(E)^k \ar[r] \ar[d] & E \ar[d] \\
	J^kF  \ar[r]        & M ~ .
}}
\end{equation}
\end{definition}

Any smooth section $\phi\colon M\to F$ automatically gives rise to its
\emph{$k$-jet prolongation} or \emph{$k$-prolongation} $j^k\phi\colon
M\to J^kF$. Namely $j^k\phi$ is a section of the bundle $J^kF\to M$ that
is defined in a local adapted coordinate patch as
\begin{equation}
	j^k\phi(x) = (x^i,\phi^a(x),\del_i\phi^a(x),\ldots,
		\del_{i_1\cdots i_k}\phi^a(x))
	= (x^i,\del_I \phi^a(x)), ~~ |I|\le k.
\end{equation}
One can think of the $k$-prolongation symbol as a differential operator
\begin{equation}
	j^k\colon \Secs(F) \to \Secs(J^kF)
\end{equation}
of order $k$. In fact, any (not necessarily linear) differential operator of order $k$,
\begin{equation}
	f\colon \Secs(F) \to \Secs(E), ~~
	f\colon \phi \mapsto f[\phi] ,
\end{equation}
can be written as a composition of $j^k$ with an order $0$ (not
necessarily linear) operator $f\colon J^kF\to E$, such that $f[\phi] =
f(j^k\phi)$. Note that we are slightly abusing notation by denoting both
the differential operator and the bundle morphism by the same symbol
$f$.

Further, we can define an $l$-prolongation of a differential operator
$f$ of order $k$,
\begin{equation}
	p^l f \colon J^{k+l}F \to J^lE ,
\end{equation}
which is then a differential operator of order $k+l$, by composing with
$j^l$: $p^l f[\phi] = j^l f[\phi]$. Prolongation is discussed briefly
using coordinate-wise operations in Sec.~\ref{sec:integrability}. The
$k$-jet prolongation $j^k\phi$ can now be thought of as a special case
of bundle morphisms, that is, $j^k\phi = p^k \phi$, where on the right
hand side we interpret $\phi$ as the base fixing bundle morphism to
$F\to M$ from the trivial $0$-dimensional bundle $\id\colon M\to M$.
\begin{equation}
\xymatrix{
	M \ar[d]_\id \ar[r]^\phi & F \ar[d] \\
	M \ar[r] & M ~ .
}
\end{equation}

Given the sequence of projections of $k$-jet bundles over $M$,
\begin{equation}
	\cdots \to J^2F \to J^1F \to J^0F\cong F ,
\end{equation}
it is convenient to introduce the \emph{infinite order jet} (or
\emph{$\oo$-jet}) bundle $J^\oo F$ defined as the projective limit over
the jet order $k$
\begin{equation}
	J^\oo F = \varprojlim J^k F .
\end{equation}
This limit implicitly defines $J^\oo F$ as an infinite dimensional
smooth manifold. The main advantage of working with $\oo$-jets is that
any function or tensor on $J^kF$ for finite $k$ can be pulled back to
$J^\oo F$. Conversely, any smooth function or tensor on $J^\oo F$
depends only on jets up to some finite order, say $k$, and can be
faithfully projected to $J^kF$. Another major convenience of working on
$J^\oo F$ is the ability to decompose the usual de~Rham differential
into its \emph{horizontal} and \emph{vertical} parts
\begin{equation}
	\d = \dh + \dv.
\end{equation}
The defining property of $\dh$ is the following. Given a section
$\phi\colon M \to F$, we must have the identity
\begin{equation}
	(j^\oo\phi)^* \dh \alpha = \d (j^\oo\phi)^* \alpha,
\end{equation}
where $\alpha$ is any differential form on $J^\oo F$ and $\d$ is the
usual de~Rham differential on $M$. On the other hand, $\dv$ is
characterized by the fact that its image is annihilated by the pullback
to $M$ along any section $\phi$,
\begin{equation}
	(j^\oo \phi)^* \dv \alpha = 0.
\end{equation}
It can be checked that the horizontal
and vertical differentials anti-commute and are separately nilpotent,
\begin{equation}
	\dh\dv + \dv\dh = 0, \quad \dh^2 = 0 = \dv^2 .
\end{equation}
Note that, to apply $\dv$ or $\dh$ to forms defined on a finite order
jet bundle $J^k F$, the pullback and projection operations mentioned
above will often be applied implicitly. Thus the application of say
$\dh$ to a differential form on $J^k F$ may yield that a differential
form that projects to $J^{k+1}F$ but not to $J^k F$. In local
coordinates $(x^i,u^a)$ on $F$, and the induced coordinates
$(x^i,u^a_I)$ on $J^\oo F$, a convenient basis for differential forms is
\begin{equation}
	\dh x^i = \d x^i, \quad
	\dv u^a_I = \d u^a_I - \dh u^a_I
		= \d u^a_I - u^a_{Ii} \d x^i .
\end{equation}
We can also define two special kinds of vector fields. A vector field
$\hat{\xi}$ is \emph{horizontal} if its action in local coordinates is
\begin{equation}
	\hat{\xi}(x^i) = \xi^i, \quad
	\hat{\xi}(u^a_I) = \xi^i u^a_{iI} .
\end{equation}
for some $\xi^i=\xi^i(x,u^a_I)$. In particular, the vector field
$\hat{\del}_j$, with $\xi^i = \delta^i_j$, is horizontal. Note that
$[\hat{\del}_i,\hat{\del}_j]=0$. A vector field $\hat{\psi}$ is
\emph{evolutionary} if its action in local coordinates is
\begin{equation}
	\hat{\psi}(x^i) = 0, \quad
	\hat{\psi}(u^a_I) = \hat{\del}_I (\psi^a),
\end{equation}
for some $\psi^a = \psi^a(x,u^b_I)$, where $\hat{\del}_I(f) =
\hat{\del}_{i_1}(\hat{\del}_{i_2}(\cdots\hat{\del}_{i_k}(f)\cdots))$ for
multi-index $I=i_1i_2\cdots i_k$ (the order of application of these
vector fields does not matter since they commute). Note that the
$\psi^a$ can be seen as the fiber coordinate components of a section of
the bundle $(F)^\oo\to J^\oo F$. These definitions can be checked to be
coordinate independent.

One can show that for a horizontal vector field $\hat{\xi}$ on $J^\oo F$
there exists a vector field $\xi_\phi$ on $M$ such that their
actions on scalar functions are intertwined by the pullback along the
jet prolongation $j^\oo\phi$ of a section $\phi\colon M\to F$,
\begin{equation}
	\hat{\xi}(f)(j^\oo\phi) = \xi_\phi(f(j^\oo\phi)) ,
\end{equation}
for any scalar function $f$ on $J^\oo F$. Namely, in local coordinates,
$\xi_\phi = \xi_\phi^i\del_i$ with $\xi_\phi^i = (\iota_{\hat{\xi}} \d
x^i)(j^\oo\phi) = \hat{\xi}(x^i)(j^\oo\phi) = \xi^i(j^\oo\phi)$. On the
other hand, evolutionary vector fields $\hat{\psi}$ satisfy the identities
\begin{gather}
\label{eq:idh-comm}
	\iota_{\hat{\psi}} (\dh \alpha) + \dh (\iota_{\hat{\psi}} \alpha) = 0 , \\
% XXX: Proof: contraction i_\psi and d_h are odd derivations on the
% form algebra. Their anti-commutator is an even derivation. Therefor
% its action need only be checked on generators: scalars, exact
% horizontal forms and exact vertical forms. The main non-trivial case
% is for exact vertical forms.
\label{eq:jet-Lie}
	\Lie_\psi (j^\oo\phi)^*\alpha
		= \left.\frac{\d}{\d\eps}\right|_{\eps=0} [j^\oo(\phi+\eps\psi)]^* \alpha
		= (j^\oo\phi)^* \Lie_{\hat\psi} \alpha ,
\end{gather}
for any form $\alpha\in\Forms^*(J^\oo F)$ and section $\psi\colon M\to F$.
Actually, $\psi$ could be a section of $(F)^k\to J^kF$, that is, it
could depend on $\phi^a(x)$ and its derivatives and not only on $x\in
M$. The only corresponding change in the above formula would be to
replace $\eps\psi$ by $\eps(j^k\phi)^*\psi$. Ostensibly, $\Lie_\psi$
should stand for the Lie derivative on the infinite dimensional manifold
of sections of $F\to M$, where the section $\psi$ is identified with the
vector field whose action on local coordinates is $\Lie_\psi \phi^a(x) =
\psi^a(x)$. However, since we do not delve into the differential
geometry of infinite dimensional manifolds here, we keep the symbol
$\Lie_\psi(j^\oo\phi)^*$ primitive and defined as above.

Integrations or differentiations by parts are carried out using the
following basic identity
\begin{align}
\label{eq:byparts}
	\dv u^a_{Ii}\wedge \d x^i\wedge\alpha
	&= \dv (u^a_{Ii} \d x^i) \wedge \alpha \\
	&= (\dv\dh u^a_I) \wedge \alpha \\
	&= -(\dh\dv u^a_I) \wedge \alpha \\
	&= -\dv u^a_I \wedge \dh\alpha - \dh(\dv u^a_I\wedge\alpha) .
\end{align}

This split of the de~Rham differential into horizontal and vertical
differentials also splits the de~Rham complex $\Forms^*(J^\oo F)$ of
differential forms on $J^\oo F$ into a
\emph{bicomplex}.\cite{anderson-small,anderson-big} Since the
horizontal and vertical $1$-forms generate the graded commutative
algebra of differential forms, any form $\lambda\in \Forms^*(J^\oo F)$
can be uniquely written as
\begin{equation}
	\lambda = \sum_{h,v} \lambda_{h,v} ,
\end{equation}
where $0\le h \le n$ and $0\le v$ are respectively the horizontal and
vertical form degrees. We have thus turned the differential forms into a
bigraded complex $\Forms^*(J^\oo F) = \bigoplus_{h,v} \Forms^{h,v}(F)$,
with the $\dh$ differential increasing $h$ by $1$ and the $\dv$
differential increasing $v$ by $1$. This complex is called the
\emph{variational bicomplex}.\cite{anderson-small,anderson-big} As with
any bicomplex, we can consider its cohomology with respect to either or
any combination of the two differentials. The \emph{horizontal
cohomology} is $H^{h,v}(\dh) = H(\Forms^*(J^\oo F),\dh)$ in degrees
$(h,v)$. The \emph{vertical cohomology} is $H^{h,v}(\dv) =
H(\Forms^*(J^\oo F),\dv)$ in degrees $(h,v)$. Both $(H^{h,*}/\dh
H^{h-1,*},\dv)$ and $(H^{*,v}/\dv H^{*,v-1},\dh)$ still form complexes,
therefore we can also consider their cohomologies. The \emph{relative
cohomologies} are $H^{h,*}(\dv|\dh) = H(H^{h,*}/\dh H^{h-1,*},\dv)$ and
$H^{*,v}(\dh|\dv) = H(H^{*,v}/\dv H^{*,v-1},\dh)$.

\section{Jet bundles and systems of PDEs}\label{sec:jets-pdes}
This appendix outlines the description of PDEs as submanifolds of the jet
bundle. Jet bundles are briefly introduced in \ref{sec:jets},
where also notation is fixed (not all of it being completely standard)
and standard literature references are given.  Such a description of
PDEs is more intrinsic than than the usual one in terms of equations,
but is essentially equivalent. This approach is well known in the
geometric and formal theory of differential
systems.\cite{seiler-inv,bcggg,vk}

From now on, fix $M$ to be finite dimensional manifold and let $n=\dim
M$. Also fix a vector bundle $F\to M$. We refer to $M$ as the
\emph{spacetime} manifold and to $F$ as the \emph{field bundle}.

We restrict our attention to \emph{regular} PDEs in the following sense.
\begin{definition}
A \emph{PDE system $\E$ of order $k$} is a smooth, closed sub-bundle of
$J^kF\to M$, $\E\sso J^kF$.
\end{definition}
Note that $\E$ need not be a vector sub-bundle of $J^kF$. The above
definition may seem unfamiliar to some, but can be cast in more
recognizable form using the following
\begin{proposition}
Given a PDE system $\E$ of order $k$, there exists (up to a global
obstruction) a vector bundle $E\to M$, a smooth sub-bundle $E'\sse E$
containing the zero section of $E\to M$, and a smooth base fixing smooth
bundle morphism $f\colon J^kF \to E$ such that the image of $f$ is
contained in $E'$, the image of $f$ is transverse in $E'$ to the zero
section of $E$ and $\E$ is precisely the preimage of the zero section,
that is, $\E$ satisfies $f=0$.
\end{proposition}
The proof follows from basic differential topology. The obstruction is
of a global topological nature~\cite[\textsection 7]{goldschmidt} and is
related to the fact that not every embedded submanifold can be
represented as the zero-set of a section of a vector bundle. Clearly,
the equation form is not unique. For instance, applying any invertible
transformation to the equations $f=0$ gives another equation form
$f'=0$, which describes exactly the same PDE system. 

We refer to $E\to M$ as the \emph{equation bundle} and to $f$ or the
pair $(f,E)$ as the \emph{equation form} of the PDE system $\E$. A
section $\phi\colon M \to F$, also referred to as a \emph{field
configuration}, is said to satisfy the PDE system $\E$ if the $k$-jet
prolongation of $\phi$ is contained in $\E$, $j^k\phi(x)\in \E_x\sso
J^k_x(F,M)$. Then, equivalently, $j^k\phi$ is a section of $\E\to M$. We
denote the space of all solution sections by $\S(F)\sso \Secs(F)$ or
$\S_\E(F)$ when the PDE system needs to be mentioned explicitly. Using
the above proposition, we can equivalently say that $\phi$ is a solution of
the PDE system $\E$ if
\begin{equation}
	f[\phi] = f(j^k\phi) = 0.
\end{equation}
Expressing the $k$-jet in local coordinates, $j^k\phi(x) = (x,\phi^a(x),
\del_i \phi^a(x), \ldots)$, it is clear that $f(x,\phi^a(x),
\del_i\phi^a(x), \ldots)=0$ is a system of partial differential
equations in the usual sense of the term. Starting with a PDE system in
the usual sense, its geometric form as a sub-bundle of the jet bundle
can be obtained by a converse of the above lemma. At this point, the
regularity assumptions on both $\E$ and $f$ become important. Namely,
the transversality properties of $f$ ensure that the zero set of $f=0$
is a submanifold of $J^kF$ and vice versa.

The linear and affine structures on $J^kF$ give us the possibility of
defining the notion of \emph{linear} and \emph{quasilinear} PDE systems.
\begin{definition}
A PDE system $\E\sso J^kF$ is called \emph{linear} if $\E\to M$ is a
vector sub-bundle of the vector bundle $J^kF\to M$. The PDE system is
called \emph{quasilinear} if $\E\to J^{k-1}F$ is an affine sub-bundle of
the affine bundle $J^kF\to J^{k-1}F$.
\end{definition}
The connection to the usual meanings of these terms can be seen through
adapted equation forms.
\begin{lemma}
The PDE system $\E\sso J^kF$ is linear iff it has an equation form
$(f,E)$, where $f\colon J^kF\to E$ is a morphism of vector bundles over
$M$.

The PDE system $\E\sso J^kF$ is quasilinear iff it has an equation form
$(f,E)$, where $f\colon J^kF\to E$ is a morphism of affine bundles,
which fits into the commutative diagram
\begin{equation}
\vcenter{\xymatrix{
	J^kF \ar[d] \ar[r]^f & E \ar[d] \\
	J^{k-1}F    \ar[r]   & M ~ ,
}}
\end{equation}
where the vertical maps define the affine bundles, with the vector bundle
$E\to M$ naturally considered an affine one.
\end{lemma}
The proof is immediate. Alternatively, the quasilinear case can be cast
into the form of a base fixing affine bundle morphism $f\colon J^kF\to
(E)^{k-1}$, where both bundles are over $J^{k-1}F$. Such equation forms
are called \emph{adapted}.

In the more common language of adapted local coordinates, the conditions
of linearity and quasilinearity are expressed as follows. Consider
adapted local coordinates $(x^i,v_A)$ on the equation bundle $E$,
$(x^i,u^a)$ on the field bundle $F$, and the corresponding $(x^i,u^a_I)$
on the $k$-jet bundle $J^kF$.
Let
$\phi\colon M\to F$ be a field configuration, then its $k$-jet in local
coordinates is $j^k\phi(x) = (x^i,\del_I\phi^a(x))$. The
above lemma asserts the existence of an equation form that looks like
\begin{equation}\label{lin-eqform}
	f_{Aa}^I(x) \del_I\phi^a(x) = 0 .
\end{equation}
Note that this equation is linear in $\phi(x)$ and its derivatives and
that the coefficients $f_{Aa}^I(x)$, with multi-indices $I$, depend only
on the base space coordinates $x$. On the other hand, for a quasilinear
equation, the lemma asserts the existence of an equation form that looks
like (with $|I|=k$)
\begin{equation}\label{qlin-eqform}
	f_{Aa}^I(x,j^{k-1}\phi(x)) \del_I \phi^a(x) + f_A(x,j^{k-1}\phi(x)) = 0 .
\end{equation}
Note that, in the linear case, the fact that the coefficients of
$f\colon J^k(F,M)\to E$ only depend on the base space coordinates $x$ is
captured by the requirement that it is a morphism of vector bundles over
$M$. In the quasilinear case, the coefficients of $f$ can obviously
depend on both $x$, $\phi(x)$ as well as all derivatives $\del_I\phi(x)$
up to order $|I|=k-1$, which is captured by allowing $f\colon J^kF\to
(E)^{k-1}$ to be a (base fixing) bundle morphism over $J^{k-1}F$. It is
worth remarking that any linear PDE system is also naturally
quasilinear.

Recall that the affine bundle $J^kF\to J^{k-1}F$ is modeled on
the vector bundle $(S^kT^*M\otimes_M F)^{k-1}\to J^{k-1}F$.
Therefore, an adapted equation form $(f,E)$ of a quasilinear PDE system
$\E\sso J^kF$ naturally singles out a section
\begin{equation}\label{eq:prsym}
	\bar{f}\colon J^{k-1}F
		\to (E\otimes_M F^* \otimes_M S^kTM)^{k-1} .
\end{equation}
In local coordinates, $\bar{f}$ corresponds to the coefficient
$f_{Aa}^{I}$ of the highest derivative term $\del_{I}\phi(x)$ with
$|I|=k$ in Eq.~\eqref{qlin-eqform}. This section $\bar{f}$ is called the
\emph{principal symbol} of the given equation form of $\E$. If the
equation is linear, rather than quasilinear, $\bar{f}$ can be projected
from a section on $J^{k-1} F$ to a section on $M$. Moreover, if we fix
$x\in M$ and $p\in T^*_x M$, we can define the linear map $\bar{f}_{x,p}
= \bar{f}_x \cdot p^{\otimes k}$,
\begin{equation}\label{eq:ps-map}
	\bar{f}_{x,p}\colon E_x \to F_x ,
\end{equation}
which we also refer to as the principal symbol.

\subsection{Prolongation, integrability, equivalence}
\label{sec:integrability}
One reason to discuss PDE systems as submanifolds of a jet bundle is
independence of a particular equation form. Any two equation forms are
equivalent if they define the same PDE system manifold. We should
specify our notion of equivalence.
\begin{definition}\label{def:pde-equiv}
Consider two field bundles $F_i\to M$, $i=1,2$, and two PDE systems
$\E_i \sse J^{k_i}F_i$. Denote the corresponding spaces of smooth
solution sections by $\S_i(F_i)$. The PDE systems $\E_1$ and $\E_2$ are
said to be equivalent if there exist bundle morphisms $e_{ij}\colon
J^{l_i}F_i\to F_j$, $i\ne j$, such that
\begin{equation}
	\phi_i \in \S_i(F_i) ~~\text{and}~~
	\phi_j = e_{ij}\circ j^{l_i}\phi_i ~~\text{implies}~~
	\phi_j \in \S_j(F_j) ,
\end{equation}
as well as that $e_{12}\circ j^{l_1}$ and $e_{21}\circ j^{l_2}$ are
mutual inverses when restricted to the solution spaces $\S_1(F_1)$ and
$\S_2(F_2)$.
\end{definition}
We can easily extend the notion of equivalence to equation forms of PDE
systems. In that case two different equations forms that define the same
PDE system manifold are trivially equivalent. Note that neither the
field bundles nor the orders of the PDE systems need to be same for
equivalence to hold.

Let us restrict to the case that is of importance elsewhere in this
review, namely of $F_1=F_2=F$ and $e_{12}$ and $e_{21}$ respectively
equal to the canonical projections $J^{l_1}F\to F$ and $J^{l_2}F\to F$,
which are in a sense trivial. In this case, it is certainly sufficient
that $\E_1 = \E_2$ for equivalence to hold, but it is not necessary. In
fact $\E_1$ and $\E_2$ could be of different orders. To obtain necessary
conditions for equivalence, we need to consider \emph{prolongation} of
PDE systems and the possible resulting \emph{integrability conditions}.

A discussion of these notions in the setting of the jet bundle
description of PDE systems can be rather technical. On the other hand,
the theory of equivalence of PDE systems formulated in these terms has
become quite mature and has yielded some important results. The
technical details of this theory can be found for example in
Refs.~\citen{bcggg,seiler-inv}. Below we give a brief non-technical
introduction to this theory and state some simplified results relevant
for hyperbolic systems.

The step by step derivation and inclusion of integrability conditions
into a PDE is called prolongation. It is easiest to define prolongation
in equation form and in local coordinates. Consider an equation form
$(f,E)$ of a PDE system $\E\sse J^kF$, as well as local coordinates
$(x^i,u^a)$ on $F$ and $(x^i,v_A)$ on $E$. If the section $\phi\colon
M\to F$ satisfies the PDE system, we have the following system of
equations holding in local coordinates
\begin{equation}
	f_A(x^j,\del_J\phi^a) = 0 .
\end{equation}
These equations hold for each point $x\in M$, therefore when both sides
are differentiated with respect to the coordinates on $M$, the resulting
equations are still satisfied,
\begin{equation}
	\del_i f_A(x^j,\del_J\phi^a)
	= (\hat{\del}_i f_A)(x^j,\del_J\phi^a) = 0 ,
\end{equation}
where $\hat{\del}_i f_A$ are functions on $J^{k+1}F$ obtained by pulling
back the functions $f_A$ from $J^kF$ to $J^{k+1}F$ and applying the
horizontal vector field $\hat{\del}_i$. These new functions $f_{iA} =
\hat{\del}_i f_A$, together with the old $f_A$ ones, constitute the
local coordinate expression for the equation form $(p^1 f,J^1E)$, where
$p^1$ is the $1$-prolongation defined in Sec.~\ref{sec:jets}.
We call the corresponding PDE system $\E^1 = \E_{p^1 f} \sso J^{k+1}F$
the \emph{first prolongation} of $\E$ or also its \emph{prolongation to
order $k+1$}. Prolongations to any higher order, $(p^l f,J^lE)$ and
$\E^{l} \sso J^{k+l}F$, are defined iteratively.

Let $p_l\colon J^{k+l}F\to J^kF$ be the canonical jet projection, which
restricts to $p_l\colon \E^l\to \E$.  Notice that we necessarily have
$p_l(\E^{l}) \sse \E$, since the prolonged system contains the original
one as a subsystem. We have just shown that sections satisfying $\E$
automatically satisfy $\E^{l}$, and vice versa. In other words,
$\S_{\E}(F) = \S_{\E^{l}}(F)$ and the two PDE systems are equivalent.
However, the inclusion $p_l(\E^{l}) \sse \E$ may be strict, which would
mean that there exist non-trivial \emph{integrability conditions}. There
exists an equation form $(f\oplus g,E\oplus G)$ for $p_l(\E^{l})\sso J^k
F$, where $(g,G)$ is an equation form for the integrability conditions.
These observations provide another sufficient condition for the
equivalence of two PDE systems, namely that there exists an order $l\ge
k_1,k_2$ such that $\E_1^{l-k_1} = \E_2^{l-k_2}$ as subsets of $J^lF$.

Prolongation can be iterated indefinitely. Taking this process to its
limit, we obtain the infinite order prolongation $\E^\oo \sso J^\oo F$
from the equation form $(p^\oo f,J^\oo E)$, which takes all possible
integrability conditions into account. One can then show that the
equality $\E_1^{\oo} = \E_2^{\oo}$, as subsets of $J^\oo F$, is both a
necessary and a sufficient condition for the equivalence of two PDE
systems. It is a deep theorem of the geometric theory of PDE
systems\cite{seiler-inv,bcggg,goldschmidt}
that, for any given PDE system, there exists a finite order $l$ such
that prolongations above that order introduce no new integrability
conditions. Therefore, this restricted version of the equivalence
problem can be decided in finitely many steps.

\begin{remark}
In mathematical physics, PDEs systems are often obtained in variational
form (as Euler-Lagrange equations of some Lagrangian). However, this
form need not be one for which the existence of Green functions are
readily available. Thus, it becomes important to formalize, as is done
above, how these systems can be brought into an equivalent form that can
be shown to be Green hyperbolic (cf.~Sec.~\ref{sec:green-hyper}) using
standard methods. For example, the Klein-Gordon equation is normal
hyperbolic,\cite{bgp} but not symmetric
hyperbolic.\cite{geroch-pde,baer-green} The Dirac and Proca equations
are neither. However, each of these equations can be shown to be
equivalent to either a symmetric or normal hyperbolic system with
constraints.\cite{geroch-pde,bgp,baer-green,waldmann-pde} The inclusion
of constraints into the description of a hyperbolic system is discussed
in Sec.~\ref{sec:caus-green}.
\end{remark}

\subsection{Formal exactness}\label{sec:formal-exact}
An important concept used in this work is that of a formally exact
complex (or sequence) of differential operators. There are several
related concepts, which we discuss briefly below. More details can be
found in the Refs.~\citen{tarkhanov,pommaret,goldschmidt-lin}.

A \emph{complex of differential operators} or \emph{differential
complex} consists of vector bundles $E,F,G \to M$ and differential operators
$f\colon E\to F$ and $g\colon F\to G$ such that $g\circ f = 0$; it is
often written as
\begin{gather}
\label{eq:sec-seq}
\xymatrix{
	\Secs(E) \ar[r]^{f} & \Secs(F) \ar[r]^{g} & \Secs(G)
} \\
\label{eq:jet-seq}
	\text{or} \quad
\xymatrix{
	J^\oo E \ar[r]^{p^\oo f} & J^\oo F \ar[r]^{p^\oo g} & J^\oo G
} .
\end{gather}
Of course, the sequence of operators in a differential complex could
have any length, not just two. Given a complex, we are of course free to
define its cohomology, $\ker g / \im f$. If this cohomology vanishes,
then the complex is said to be \emph{exact} or to form an \emph{exact
sequence}. Different kinds cohomologies can be defined by considering
different spaces on which the differential operators are defined. There
are of course different kinds of exactness associated to them.

The complex is \emph{locally exact} if for every $x\in M$ there exists a
neighborhood $U\sse M$ of $x$ such that the sequence
\begin{equation}
\xymatrix{
	\Secs(E|_U) \ar[r]^{f} & \Secs(F|_U) \ar[r]^{g} & \Secs(G|_U)
}
\end{equation}
is exact. The complex is \emph{globally exact} if the
sequence~\eqref{eq:sec-seq} is exact. Of course there are different
versions of global exactness if we replace arbitrary smooth sections by
say sections with compact support, or other restriction. In practical
applications, it is global exactness or the knowledge of the global
cohomology that is important. Local exactness is important because it
allows the use of sheaf-theoretic methods to compute the global
cohomology.

Local exactness by itself is a difficult property to check, because it
is essentially a functional analytical condition. A simpler geometric
condition is formal exactness. The complex is \emph{formally exact} if
the sequence~\eqref{eq:jet-seq} is exact as a sequence of (infinite
dimensional) vector bundles, which is the same as exactness
each of the sequences
\begin{equation}
\xymatrix{
	J^{s+k+l} E \ar[r]^-{p^{s+l} f} & J^{s+l} F \ar[r]^-{p^s g} & J^s G
}
\end{equation}
of (finite dimensional) vector bundles, with $k$ and $l$ being the
respective orders of the operators $f$ and $g$. Formal exactness is not
a sufficient condition for local exactness, but it is necessary and is a
first step in trying to establish the stronger condition. For a given
differential operator $f\colon \Secs(E) \to \Secs(F)$, the existence of
a differential operator $g\colon \Secs(F) \to \Secs(G)$ extending $f$ to
a formally exact sequence is assured for abstract
reasons~\cite{goldschmidt-lin,tarkhanov}, provided the usual regularity
conditions hold.

Yet another related and simpler condition is
ellipticity~\cite[\textsection XIX.4]{hoermander-III}. The complex is
said to be \emph{elliptic} if for every $x\in M$ and non-zero $p\in
T^*_x M$ the sequence of principal symbols (cf.~Eq.~\eqref{eq:ps-map})
\begin{equation}
\xymatrix{
	E_x \ar[r]^{\bar{f}_{x,p}} & F_x \ar[r]^{\bar{g}_{x,p}} & G_x
}
\end{equation}
is exact. The de~Rham complex is a prototypical example of both an elliptic
and a formally exact complex. However, these conditions are in general
independent,\cite{smith} though in many cases ellipticity has proven helpful
in checking formal exactness.

\section{Causal structure on conal manifolds}\label{sec:conal}
In the literature on relativity, causal structure is most often studied
as a subset of Lorentzian geometry.\cite{he} The term \emph{Lorentzian
geometry} refers to the study of structures induced on spacetime
manifolds by the presence of a Lorentzian metric.  One of these
structures consists of the cones of \emph{null vectors}. In particular,
it is these cones that determine causal relationships between points in
Lorentzian spacetimes.  While causal relationships themselves can be
defined solely in terms of the null cones, the reason they deserve the
name \emph{causal} is that they also describe the maximal speed at which
disturbances can travel in solutions of hyperbolic PDEs with
wave-like\cite{bgp,waldmann-pde} principal symbols. However, a similar
property holds for more general classes of hyperbolic equations, even
those that have no relation to a Lorentzian metric. It stands to reason
then that causal relationships should be definable in terms of the
intrinsic geometry of such PDEs.  Indeed, if we consider cones of
so-called \emph{characteristic covectors}\cite{beig-hyper} and define
causal relationships in terms of them, surprisingly few changes are
necessary, with Lorentzian null cones appearing as special cases for the
class of wave-like PDEs mentioned above. It stands to reason to give the
study of such cones the name \emph{characteristic geometry}. On the
other hand, we find it convenient to generalize even further and
consider simply a priori given cones (in the tangent and cotangent
space), thus abstracting and clarifying the geometric notions that go
into the definitions and basic properties of causal relations. Thus, we
shall actually be studying \emph{conal geometry} or what have sometimes
been called \emph{conal manifolds}.\cite{lawson,neeb} This abstraction
highlights the fact that a basic tool in the study of causal structures
should be differential topology, rather than pseudo-Riemannian geometry.

Characteristic geometry and its abstraction to conal geometry are
discussed in a fuller and more integrated way in Ref.~\citen{kh-caus}. Some cues
have been taken from previous attempts to abstract the notion of causal
structure or causal order in Lorentzian
geometry.\cite{penrose,geroch-gh,he,gps} The generalization from
Lorentzian cones to more general ones, for the purposes of describing
causality in quantum field theory has been considered
before,\cite{bannier,rainer1,rainer2} but not in a concrete way.

Each point of a conal manifold, referred to here as a \emph{cone
bundle}, is smoothly assigned an open \emph{cone} (a set invariant under
multiplication by positive scalars) of tangent or cotangent vectors.
\begin{definition}
A smooth bundle $C\to M$ of finite dimensional manifolds is termed a
\emph{cone bundle} if there exists an \emph{enveloping} vector bundle
$E\to M$ and an inclusion bundle morphism $\iota\colon C \sso E$, such
that each fiber $C_x$, $x\in M$, is an open convex cone in the
corresponding fiber $E_x$ (where we have implicitly identified $C$ with
its image $\iota(C)\sso E$). A bundle map $\chi\colon C \to C'$ is a
\emph{cone bundle morphism} if, given corresponding enveloping vector
bundles $E\to M$ and $E'\to M'$, there exists a vector bundle morphism
$\psi\colon E\to E'$ such that $\chi = \psi|_C$. Namely, the following
diagrams exist and commute:
\begin{equation}
\vcenter{\xymatrix{
	E \ar[d] \ar[r]^\psi & E' \ar[d] \\
	M        \ar[r]      & M'
}} , \quad
\vcenter{\xymatrix{
	C \ar[r]^\chi \ar[d]^\sso &
	C'            \ar[d]^\sso \\
	E \ar[r]^\psi             & E'
}} .
\end{equation}
\end{definition}
Before proceeding, we need some terminology concerning cones and
operations on them. These notions are often used in convex
geometry.\cite{rockafellar}
\begin{definition}
Given a finite dimensional vector space $V$ and a convex cone $C\sso V$,
denote its closure by $\bar{C}$ and its open interior by $\mathring{C}$.
We define the \emph{convex dual} (a.k.a.\ \emph{polar dual}) $C^*\sso
V^*$ as the set
\begin{equation}
	C^* = \{ u\in V^* \mid u\cdot v \ge 0 \quad\text{for all}~v\in C \} .
\end{equation}
We define the \emph{strict convex dual} $C^\oast\sso V^*$ as the set
\begin{equation}
	C^\oast = \{ u\in V^* \mid u\cdot v > 0 \quad\text{for all}~v\in
		\bar{C}\setminus\{0\} \} .
\end{equation}
\end{definition}
The attribute \emph{strict} may be dropped from the description of
$C^\oast$ when it is clear from context. It is easy to check the
following
\begin{proposition}
Consider a convex cone $C$.
\begin{enumerate}
\item[(i)]
	The convex dual $C^*$ is always closed and also convex. In addition,
	$C^{**}=\bar{C}$. The strict convex dual $C^\oast$ is always open and
	convex.
\item[(ii)]
	$C^*\setminus\{0\}$ is non-empty iff $C$ is contained in a closed half
	space. $C^\oast$ is non-empty iff $C$ contains no affine line (it is
	\emph{salient}).
\item[(iii)]
	If $C$ is open and salient, then $C^{\oast\oast} = C$.
\item[(iv)]
	The inclusion of cones $C_1\sse C_2$ implies the reverse inclusion of
	their duals, $C_1^*\supseteq C_2^*$ and $C_1^\oast\supseteq C_2^\oast$.
\item[(v)]
	The convex dual of the intersection of closures of cones $C_1$ and
	$C_2$ is the convex union (convex hull of the union) of their duals,
	$(\bar{C}_1\cap \bar{C}_2)^* = C_1^* + C_2^*$, where the right hand
	side is written as a Minkowski sum, which for cones coincides with the
	convex hull of the union. The converse identity holds as well,
	$(\bar{C}_1+\bar{C}_2)^* = C_1^*\cap C_2^*$. Similarly, if $C_1$ and
	$C_2$ are open and salient, then $(C_1\cap C_2)^\oast = C_1^\oast +
	C_2^\oast$ and $(C_1+C_2)^\oast = C_1^\oast \cap C_2^\oast$.
\end{enumerate}
\end{proposition}
We extend these operations to cone bundles by acting fiberwise. So, if
$C\to M$ is a cone sub-bundle of a vector bundle $E\to M$, then the
\emph{convex dual cone bundle} $C^*\to M$ is the cone sub-bundle of
$E^*\to M$ such that each fiber $C^*_x$ is the convex dual of the
corresponding fiber $C_x$, for $x\in M$. Similarly, we can define the
\emph{strict convex dual cone bundle} $C^\oast\to M$. The operations of
intersection $(\cap)$ and convex union $(+)$ are extended to cone
bundles in the same way. Clearly if $C^\oast\to M$ is a smooth bundle,
it is also a cone bundle. On the other hand, $C^*\to M$ is usually not
(since it is usually not open), though for brevity we shall sometimes refer to
it as a cone bundle anyway.

In the next definition, we make a slight break with the usual
terminology concerning covectors in Lorentzian geometry. A covector
naturally defines a codimension-$1$ subspace, its annihilator, of the
tangent space. In the presence of a metric, a covector can be
canonically identified with a vector. If this vector is timelike, the
corresponding codimension-$1$ subspace is called spacelike. However, a
direct means of identifying covectors with vectors is missing in
general. On the other hand, the link between covectors and
codimension-$1$ tangent subspaces is metric independent and since the
term \emph{spacelike} still makes sense for tangent subspaces, we
transfer it to corresponding covectors naming them \emph{spacelike} as
well.
\begin{definition}
Given a manifold $M$, a \emph{chronal cone bundle} on $M$ is a cone
bundle $C\to M$ enveloped by the tangent
bundle $TM\to M$ such that each fiber $C_x$, $x\in M$, is a \emph{proper
cone} (non-empty, open, convex, salient). The elements of $C$ are
\emph{future directed, timelike vectors}.  The strict convex dual
$C^\oast\to M$, enveloped by the cotangent bundle $T^*M\to M$, is the
corresponding \emph{spacelike cone bundle}. The elements of $C^\oast$
are \emph{future oriented, spacelike covectors}. The corresponding cone
bundles $\bar{C}\to M$ and $\bar{C}^\oast\to M$ are referred to,
respectively, as the \emph{causal} and \emph{cocausal} cone bundles. A
morphism of two chronal or spacelike cone bundles $C\to M$ and $C'\to
M'$ is induced by pushforward $\chi_*\colon TM \to TM'$ (or $T^*M\to
T^*M'$) of an open embedding $\chi\colon M\to M'$.
\end{definition}
Note that the open convex dual of a proper cone is again a proper cone.
Prototypical examples of chronal and spacelike cone bundles are the
cone bundles of timelike vectors and spacelike covectors on a Lorentzian
manifold.

At this point, one may recall some standard notions of Lorentzian
geometry, as long as they are defined only in terms of timelike cones, and
apply them to the geometry of cone bundles. Many of the standard theorems
translate as well, some directly and others with some extra effort. We
restrict ourselves to those that are relevant to the issues at hand.

Fix a chronal cone bundle $C\to M$ on a spacetime manifold $M$. Note
that any chronal cone bundle is \emph{time oriented}, since by
definition the fibers consist of single cones rather than double cones
like in the Lorentzian case. By assumptions the cones are directed into
the future.
\begin{definition}\label{def:I}
A smooth curve $\gamma$ in $M$ is called \emph{future directed,
timelike} if the tangent to $\gamma$ is everywhere contained in $C$.
The \emph{chronal precedence relation} $I^+\sse M\times M$ (also
$I^+_C$) is defined as
\begin{equation}
	I^+ = \{ (x,y)\in M\times M \mid \exists \gamma,
		~\text{future directed, timelike curve from $x$ to $y$} \} .
\end{equation}
When $(x,y)\in I^+$, we say that $x$ \emph{chronologically precedes} $y$
and also write $x\ll y$.

If we replace $C$ by $\bar{C}$ in the above definitions, we obtain
\emph{causal curves} and the \emph{causal precedence relation} $J^+\sse
M\times M$, denoted $x<y$. The inverse relations are written $I^-$ and
$J^-$.
\end{definition}
It is clear that $I^+$ is an open, transitive relation ($x\ll y$ and
$y\ll z$ implies $x\ll z$). Using this relation, we can define the usual
causal hierarchy.
\begin{definition}\label{def:chron}
\begin{enumerate}
\item[(i)]
	If $I^+$ is irreflexive ($x\not\ll x$ or, equivalently, no closed
	timelike curves exist), then $C$ is \emph{chronological}.
\item[(ii)]
	Given an open $N\sse M$, $I^+_{C|_N}$ and $I^+_{C}\cap (N\times N)$ are
	both relations on $N\times N$. We say that $N$ is
	\emph{chronologically compatible} if both of these relations coincide.
	More conventionally, this means that any two points of $N$ that can be
	joined by a timelike curve in $M$ can also be joined by a timelike
	curve in $N$.
	% This means that the chronological precedence relation commutes with
	% restriction (a sheaf-like condition?)
\item[(iii)]
	An open $N\sse M$ is called \emph{chronologically convex} if any timelike
	curve that joins any two points of $N$ must also lie in $N$, which is
	a stronger condition than chronological compatibility.
	% This means that the chronological precedence *groupoid* commutes
	% with restriction.
\item[(iv)]
	The chronal cone bundle $C$ is said to be \emph{strongly chronological} if
	it is chronological and for every $x\in M$ and every neighborhood $N\sse M$
	of $x$, there exists a smaller open neighborhood $L\sse N$ that is
	chronologically convex.
	% Convex neighborhoods form a fundamental system of neighborhoods
	% that generates the manifold topology.
\item[(v)]
	The chronal cone bundle $C$ is said to be \emph{stably chronological}
	if there exists another chronal cone bundle $C'$ that is itself
	chronological and an open neighborhood of the closure of $C$, that is,
	$\bar{C}\setminus\{0\} \sso C'$. For spacelike cone bundles, stable
	chronology is equivalent to the reverse inclusion
	$\bar{C}^{\prime\oast}\setminus\{0\} \sso C^\oast$.
\end{enumerate}
Each of these definitions has an obvious analog when the adjectives
\emph{chronological} or \emph{timelike} are replaced by \emph{causal}.
\end{definition}
Note that \emph{stably chronological} is equivalent to \emph{stably
causal}, so these terms will be used interchangeably.  Moreover, the
chronological chronal cone bundle $C'$ such that it contains a stably
chronological chronal cone bundle $C$ can itself be chosen to be stably
chronological.
% XXX: Proof sketch: Use proper containment of $\bar{C}\sse C'$ to build
% a chronal cone bundle C'' such that $\bar{C}\sse C''$ and
% $\bar{C}''\sse C'$. The chronology of $C'$ is faster than the
% chronology of $C''$ that is faster than the causality of $C$.
% Therefore $\bar{C}''$ is causal and $\bar{C}''$ is a neighborhood of
% $\bar{C}$. Therefore $\bar{C}$ is stably causal. The reverse
% implication is similar.
Next we turn from curves to surfaces.
\begin{definition}\label{def:cauchy}
Each of the following concepts may be prefaced with $C$- or $C^\oast$-
to be more specific.
\begin{enumerate}
\item[(i)]
	An oriented codim-$1$ surface $S\sso M$ is called \emph{future
	oriented, spacelike} if its oriented conormals are everywhere
	contained in $C^\oast$.
\item[(ii)]
	A codim-$1$ surface $S\sso M$ is called \emph{achronal} if it
	$(S\times S)\cap I^+ = \varnothing$, that is, no two points of $S$ are
	connected by a timelike curve. Similarly, $S$ is \emph{acausal} when
	$(S\times S)\cap J^+ = \varnothing$.
\item[(iii)]
	A codim-$1$ surface $S\sso M$ is called \emph{Cauchy} if it is acausal
	and every inextensible causal curve intersects $S$ exactly once.
\item[(iv)]
	A chronal cone bundle $C\to M$ is called \emph{globally hyperbolic} if
	there exists a Cauchy surface $S\sse M$. A spacelike cone bundle
	$C^\oast\to M$ is called \emph{globally hyperbolic} if $C\to M$
	is.
\end{enumerate}
\end{definition}
Next we define some commonly used domains. Fix $S\sso M$ to be a
$C$-acausal codim-$1$ submanifold, such that either $S$ is closed or
$\bar{S}\sso M$ is a submanifold with boundary.
\begin{definition}\label{def:caus-dom}
\begin{enumerate}
\item[(i)]
	The \emph{future/past domain of influence} $I^\pm(N)$ of a subset
	$N\sse M$ is the set of points $y\in M$ such that there exists $x\in
	N$ with either $x\ll y$ ($+$) or $y\ll x$ ($-$). Let $I(N) = I^+(N)
	\cup I^-(N) \cup N$.
\item[(ii)]
	The \emph{domain of dependence} $D(S)$ is the largest open subset of $M$ for
	which $S$ is a Cauchy surface. More commonly, $D(S)$ is the set of points
	$y\in M$ such that every inextensible timelike curve through $y$
	intersects $S$. Let $D^\pm(S) = D(S) \cap I^\pm(S)$.
\item[(iii)]
	An open subset $L\sse M$ is \emph{lens-shaped} with respect to $S$ if
	it can be smoothly factored as $L\cong (-1,1)\times S$, with $t\colon
	L\to (-1,1)$ denoting the projection onto the first factor (the
	\emph{temporal function}), such that the level set $t=0$ is $S$ and all
	other level sets are spacelike as well as share the same boundary as
	$S$ in $M$ (which may be empty).
\end{enumerate}
\end{definition}

The literature in relativity and Lorentzian geometry mostly makes use of
the notion of \emph{global hyperbolicity} as given above, but
specialized to Lorentzian cone bundles. On the other hand, the
literature on hyperbolic PDE systems (including symmetric and regular
hyperbolic ones) mostly makes use of lens-shaped domains. It is a
non-trivial fact that these two notions coincide. The argument is
essentially that the temporal function of a lens-shaped domain foliates
it with Cauchy surfaces. Conversely, a globally hyperbolic cone bundle
admits a temporal function and a smooth factorization that turns it into
a lens-shaped domain (glossing over some details related to spatial
compactness). The original argument establishing the converse link in
Lorentzian geometry is due to Geroch.\cite{geroch-gh}  However, his
argument only established the existence of a continuous temporal
function. The details necessary to establish the smooth version of the
result are due to more recent work of Bernal and
Sanchez.\cite{bs-smooth1,bs-smooth2} The very recent result by Fathi and
Siconolfi,\cite{fs} using completely different methods, established the
existence of a smooth temporal function and factorization for a more
general class of cone bundles that is sufficient for our purposes.
\begin{proposition}\label{prp:globhyp-lens}
If a cone bundle $C\to M$ is globally hyperbolic, then there exists a
smooth temporal function $t\colon M\to \R$, whose level sets are all
diffeomorphic and are $C$-Cauchy surfaces. Hence, an open subset $D\sse
M$ is $C$-lens-shaped with respect to $S\sse D$ iff it is globally
$C|_D$-hyperbolic, with $S$ being $C|_D$-Cauchy.
\end{proposition}

\bibliographystyle{ws-ijmpa}
\bibliography{paper-chargeom}

\end{document}